%% file: combined.tex
\documentclass[12pt]{article}

\usepackage{scicite}

\usepackage{times}

\usepackage{graphicx}
\usepackage{multirow}
\usepackage{algorithmicx}
\usepackage[ruled,vlined]{algorithm2e}
\usepackage{booktabs}
\usepackage{xspace}
\usepackage{color}
\usepackage[colorlinks]{hyperref}
\usepackage{url}
\usepackage{tikz}
\usepackage{enumitem}
\usepackage{amstext}
\usetikzlibrary{positioning}
\usepackage{caption}
\usepackage{subcaption}
\usepackage{bm}
\usepackage{amsthm}
\usepackage{amsfonts}
\usepackage{amsmath}
\usepackage{listings}
\usepackage{siunitx}
\usepackage{epstopdf}

\newenvironment{sciabstract}{%
\begin{quote} \bf}
{\end{quote}}

\newcommand{\xhdr}[1]{\vspace{1.7mm}\noindent{{\bf #1.}}}

\newcommand{\hocond}[2]{\phi_{#1}(#2)}

\newcommand{\hocut}[1]{\text{cut}_M(#1)}
\newcommand{\hovol}[1]{\text{vol}_M(#1)}

\newcommand{\vect}[1]{\bm{#1}}

\newcommand{\hide}[1]{}

\newcommand{\motiftype}[2]{M_{#2}}
\newcommand{\mbifan}{M_{\textnormal{bifan}}}
\newcommand{\medge}{M_{\textnormal{edge}}}
\DeclareMathOperator{\truth}{\textbf{1}}
\newcommand{\indicator}[1]{\truth(#1)}

\newcommand{\newcut}[2]{\textnormal{cut}^{(#1)}(#2)}
\newcommand{\newvol}[2]{\textnormal{vol}^{(#1)}(#2)}
\newcommand{\newcond}[2]{\phi^{(#1)}(#2)}
\newcommand{\newmotifcond}[3]{\phi^{(#1)}_{#2}(#3)}
\newcommand{\newmotifcut}[3]{\textnormal{cut}^{(#1)}_{#2}(#3)}
\newcommand{\newmotifvol}[3]{\textnormal{vol}^{(#1)}_{#2}(#3)}

\newcommand{\anchorset}{\mathcal{A}}
\newcommand{\anchornodes}{\chi_{\anchorset}(\vect{v})}
\newcommand{\minstance}{(\vect{v}, \anchornodes)}

\newcommand{\motifweightedij}{(W_M)_{ij}}
\newcommand{\normmotiflap}{\mathcal{L}_M}

\newcommand{\projecturl}{\url{http://snap.stanford.edu/higher-order/}}

\newtheorem{theorem}{Theorem} 
\newtheorem{lemma}[theorem]{Lemma}

\newif\ifheaders
\headersfalse

\newif\ifarxiv
\arxivtrue

\graphicspath{{./FIG/}{./SMFIG/}}

\newcounter{lastnote}

\title{
Higher-order organization of complex networks 
}

\author
{Austin R.~Benson,$^{1}$ David F.~Gleich,$^{2}$ Jure Leskovec$^{3\ast}$\\
\\
\normalsize{$^{1}$Institute for Computational and Mathematical Engineering, Stanford University} \\
\normalsize{$^{2}$Department of Computer Science, Purdue University}\\
\normalsize{$^{3}$Computer Science Department, Stanford University}\\
\\
\normalsize{$^\ast$To whom correspondence should be addressed;}\\
\normalsize{E-mail: jure@cs.stanford.edu}
}

\date{}

\begin{document} 


\baselineskip24pt


\maketitle 



\begin{sciabstract}
\input{main-abstract}
\end{sciabstract}
\pagebreak

\input{main-intro}
\input{main-methods}
\input{main-results}
\input{main-discussion}

\clearpage
\input{main-figures}

\clearpage
\renewcommand{\thepage}{S\arabic{page}}  
\renewcommand{\thesection}{S\arabic{section}}   
\renewcommand{\thetable}{S\arabic{table}}   
\renewcommand{\thefigure}{S\arabic{figure}}
\renewcommand{\theequation}{S\arabic{equation}}

\ifarxiv
\else
\setcounter{page}{1}
Title page
\clearpage

\tableofcontents
\clearpage
\fi

\section{Derivation and analysis of the motif-based spectral clustering method}
\input{supplementary-tex/SM-theory}

\section{Computational complexity and scalability of the method}
\input{supplementary-tex/SM-complexity}

\section{Matrix-based interpretation of the motif-weighted adjacency matrix}
\input{supplementary-tex/SM-matrix-interp}

\section{Alternative clustering algorithms for evaluation}
\input{supplementary-tex/SM-alternatives}

\section{Details and comparison against existing methods for the \emph{C.~elegans} network}
\input{supplementary-tex/SM-celegans}

\section{Details and comparison against existing methods for the transportation reachability network}
\label{sec:airports}
\input{supplementary-tex/SM-airports}

\section{Additional case studies}
\input{supplementary-tex/SM-case-studies}

\section{Data availability}
\input{supplementary-tex/SM-data}

\newif\ifmanualbib
\manualbibtrue

\ifarxiv
\manualbibfalse
\fi

\ifmanualbib
\else
\nocite{wagner1993between}
\fi

\clearpage
\ifmanualbib
\input{manual-SM-bib}
\else
\bibliography{scibib}
\bibliographystyle{Science}
\fi

\clearpage
Authors would like to thank Rok Sosi\v{c} for insightful comments.
ARB acknowledges the support of a Stanford Graduate Fellowship.
DFG acknowledges the support of NSF CCF-1149756 and IIS-1422918
and DARPA SIMPLEX.
JL acknowledges the support of 
NSF IIS-1149837 and CNS-1010921, NIH BD2K,
DARPA XDATA and SIMPLEX,
Boeing, Lightspeed, and Volkswagen.



\clearpage

\end{document}

%% file: main-abstract.tex

Networks are a fundamental tool for understanding and modeling complex systems
in physics, biology, neuroscience, engineering, and social science.
Many networks are known to exhibit rich, lower-order connectivity patterns that
can be captured at the level of individual nodes and edges.
However, higher-order organization of complex networks---at the level of small
network subgraphs---remains largely unknown.
Here we develop a generalized framework for clustering networks based on higher-order
connectivity patterns.  This framework provides mathematical guarantees on the
optimality of obtained clusters and scales to networks with billions of edges.
The framework reveals higher-order organization in a number
of networks including information propagation units in neuronal networks and hub
structure in transportation networks.
Results show that networks exhibit rich higher-order organizational structures
that are exposed by clustering based on higher-order connectivity patterns.

%% file: main-intro.tex

\ifheaders
\xhdr{Motifs are important}
\fi
Networks are a standard representation of data throughout the sciences, and
higher-order connectivity patterns are essential to understanding the
fundamental structures that control and mediate the behavior of many complex
systems~\cite{milo2002network,mangan2003coherent,yang2014overlapping,holland1970method,rosvall2014memory,prvzulj2004modeling,leskovec2009community}.
The most common higher-order structures are small network subgraphs, which we
refer to as network motifs (Figure~1A).
Network motifs are considered building blocks for complex
networks~\cite{milo2002network,yaverouglu2014revealing}.  For example,
feedforward loops (Figure~1A $M_5$) have proven fundamental to understanding
transcriptional regulation networks~\cite{mangan2003structure}, triangular
motifs (Figure~1A $M_1$--$M_7$) are crucial for social
networks~\cite{holland1970method}, open bidirectional wedges
(Figure~1A $M_{13}$) are key to structural hubs in the
brain~\cite{honey2007network}, and two-hop paths (Figure~1A $M_{8}$--$M_{13}$)
are essential to understanding air traffic patterns~\cite{rosvall2014memory}.
While network motifs have been recognized as fundamental units of networks, the
higher-order \emph{organization} of networks at the level of network motifs
largely remains an open question.

\ifheaders
\xhdr{Our work}
\fi
Here we use higher-order network structures to gain new insights into the
organization of complex systems.  We develop a framework that identifies
clusters of network motifs.  For each network motif (Figure~1A), a different
higher-order clustering may be revealed (Figure~1B), which means that different
organizational patterns are exposed depending on the chosen motif.

\ifheaders
\xhdr{The objective}
\fi
Conceptually, given a network motif $M$, our framework searches for a cluster of
nodes $S$ with two goals.  First, the nodes in $S$ should participate in many
instances of $M$.  Second, the set $S$ should avoid cutting instances of $M$, which
occurs when only a subset of the nodes from a motif are in the set $S$ (Figure~1B).
More precisely, given a motif $M$, the higher-order
clustering framework aims to find a cluster (defined by a set of nodes $S$) that minimizes the following ratio:
\begin{equation}\label{eqn:mcond}
\hocond{M}{S} = \hocut{S, \bar{S}} / \min(\hovol{S}, \hovol{\bar{S}}),
\end{equation}
where $\bar{S}$ denotes the remainder of the nodes (the complement of $S$),
$\hocut{S, \bar{S}}$ is the number of instances of motif $M$ with at least one
node in $S$ and one in $\bar{S}$, and $\hovol{S}$ is the number of nodes in
instances of $M$ that reside in $S$.  Equation~\ref{eqn:mcond} is a generalization of
the conductance metric in spectral graph theory, one of the most useful graph
partitioning scores~\cite{schaeffer2007graph}. We refer to $\hocond{M}{S}$ as
the motif conductance of $S$ with respect to $M$.

%% file: main-methods.tex

\ifheaders
\xhdr{Algorithm overview}
\fi
Finding the exact set of nodes $S$ that minimizes the motif conductance is
computationally infeasible~\cite{note-spectral-np-hard}.  To approximately
minimize Equation~\ref{eqn:mcond} and hence identify higher-order clusters, we
develop an optimization framework that provably finds near-optimal
clusters~(Supplementary Materials~\cite{supp}). We extend the spectral graph clustering methodology, which
is based on the eigenvalues and eigenvectors of matrices associated with the
graph~\cite{schaeffer2007graph}, to account for higher-order structures in
networks.  The resulting method maintains the properties of traditional spectral
graph clustering: computational efficiency, ease of implementation, and
mathematical guarantees on the near-optimality of obtained clusters.
Specifically, the clusters identified by our higher-order clustering framework
satisfy the motif Cheeger inequality~\cite{note-supp-cheeger}, which means that
our optimization framework finds clusters that are at most a quadratic factor
away from optimal.

\ifheaders
\xhdr{Steps of the method}
\fi
The algorithm (illustrated in Figure~1C) efficiently identifies a cluster of
nodes $S$ as follows:

\begin{itemize}
    \item{Step 1:} Given a network and a motif $M$ of interest, form the motif
                   adjacency matrix $W_M$ whose entries $(i,j)$ are the co-occurrence
                   counts of nodes $i$ and $j$ in the motif $M$:
                   \begin{equation}\label{eqn:weighting}
                   (W_M)_{ij} = \text{number of instances of $M$ that contain nodes $i$ and $j$}.
                   \end{equation}%
    \item{Step 2:} Compute the spectral ordering $\sigma$ of the nodes from the
                   normalized motif Laplacian matrix constructed via
                   $W_M$~\cite{note-spectral-ordering}.  

    \item{Step 3:} Find the prefix set of
                   $\sigma$ with the smallest motif conductance, formally:\
                   $S := \arg\min_{r} \hocond{M}{S_r}$, where %
                   $S_r = \{\sigma_1, \ldots, \sigma_r\}$.
\end{itemize}

\ifheaders
\xhdr{Summary and extensions}
\fi
For triangular motifs, the algorithm scales to networks with billions of edges
and typically only takes several hours to process graphs of such size.  On
smaller networks with hundreds of thousands of edges, the algorithm can process
motifs up to size 9~\cite{supp}.
While the worst-case computational complexity of the algorithm for triangular
motifs is $\Theta(m^{1.5})$ , where $m$ is the number of edges in the network,
in practice the algorithm is much faster.  By analyzing 16 real-world networks
where the number of edges $m$ ranges from 159,000 to 2 billion we found the
computational complexity to scale as $\Theta(m^{1.2})$.  Moreover, the algorithm
can easily be parallelized and sampling techniques can be used to further
improve performance~\cite{seshadhri2014wedge}.

The framework can be applied to directed, undirected, and weighted networks as
well as motifs~\cite{supp}. Moreover, it can also be applied to networks with
positive and negative signs on the edges, which are common in social networks
(friend vs.~foe or trust vs.~distrust edges) and metabolic networks (edges signifying 
activation vs.~inhibition)~\cite{supp}.  
The framework can be used to identify higher-order structure in networks where domain knowledge suggests the motif of interest. In the Supplementary Material~\cite{supp} we also show that when domain-specific higher-order pattern is not known in advance, the framework can also serve to identify which motifs are important for the modular organization of a given network~\cite{supp}.
Such a general framework allows for a study of complex
higher-order organizational structures in a number of different networks using
individual motifs and sets of motifs.  The framework and mathematical theory
immediately extend to other spectral methods such as localized algorithms that
find clusters around a seed node~\cite{andersen2006local} and algorithms for
finding overlapping clusters~\cite{whang2015non}.  To find several clusters, one
can use embeddings from multiple eigenvectors and $k$-means
clustering~\cite{ng2002spectral,supp} or apply recursive
bi-partitioning~\cite{boley1998principal,supp}.

%% file: main-results.tex

\ifheaders
\xhdr{Applications}
\fi

The framework can serve to identify higher-order modular organization of networks.
We apply the higher-order clustering framework to the \emph{C. elegans} neuronal network, 
where the four-node ``bi-fan'' motif (Figure~2A) is
over-expressed~\cite{milo2002network}.  The higher-order clustering
framework then reveals the organization of the motif within the \emph{C. elegans} neuronal network.
We find a cluster of 20
neurons in the frontal section with low bi-fan motif conductance
(Figure~2B).  The cluster shows a way that nictation is
controlled.  Within the cluster, ring motor neurons (RMEL/V/R), proposed
pioneers of the nerve ring~\cite{riddle1997celegans}, propagate information to
IL2 neurons, regulators of nictation~\cite{lee2012nictation}, through the neuron
RIH and several inner labial sensory neurons (Figure~2C).  Our
framework contextualizes the sifnifance of the bi-fan motif in this control
mechanism.

The framework also provides new insights into network organization beyond the
clustering of nodes based only on edges.  Results on a transportation reachability
network~\cite{frey2007clustering} demonstrate how it finds the essential hub
interconnection airports (Figure~3).  These appear as extrema on the primary
spectral direction (Figure~3C) when two-hop motifs (Figure~3A) are used to
capture highly connected nodes and non-hubs.  (The first spectral coordinate of
the normalized motif Laplacian embedding was positively correlated with the
airport city's metropolitan population with Pearson correlation 99\% confidence
interval [0.33, 0.53]).  The secondary spectral direction identified the
West-East geography in the North American flight network (it was negatively
correlated with the airport city's longitude with Pearson correlation 99\%
confidence interval [-0.66, -0.50]).  On the other hand, edge-based methods
conflate geography and hub structure.  For example, Atlanta, a large hub, is
embedded next to Salina, a non-hub, with an edge-based method (Figure~3D).


%% file: main-discussion.tex
\ifheaders
\xhdr{Conclusion}
\fi

Our higher-order network clustering framework unifies motif analysis and network
partitioning---two fundamental tools in network science---and reveals new
organizational patterns and modules in complex systems.
%
Prior efforts along these lines do not provide worst-case performance guarantees
on the obtained clustering~\cite{serrour2011detecting}, do not reveal which
motifs organize the network~\cite{michoel2011enrichment}, or rely on expanding
the size of the network~\cite{benson2015tensor,krzakala2013spectral}.
%
Theoretical results in the Supplementary Material~\cite{supp} also explain why
classes of hypergraph partitioning methods are more general than previously
assumed and how motif-based clustering provides a rigorous framework for the
special case of partitioning directed graphs.  Finally, the higher-order network
clustering framework is generally applicable to a wide range of networks types,
including directed, undirected, weighted, and signed networks.

%% file: main-figures.tex
\begin{figure}[t]
\vspace{-1.1cm}
\centering
\includegraphics[width=0.93\textwidth]{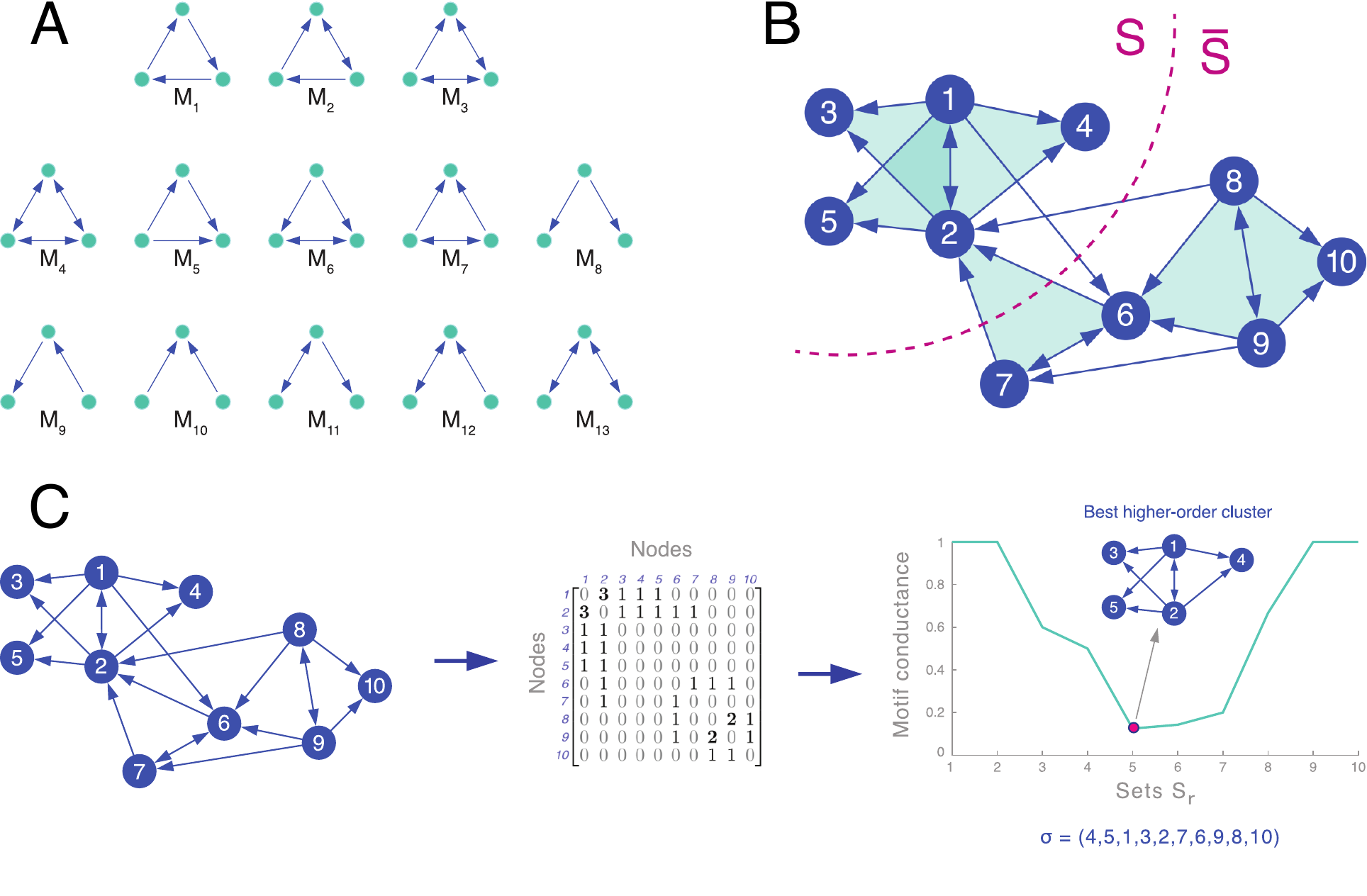}
\vspace{-0.2cm}
\caption{%
{\bf Higher-order network structures and the higher-order network clustering framework.}
{\bf A:}
Higher-order structures are captured by network motifs.  For
example, all 13 connected three-node directed motifs are shown here.
{\bf B:} 
Clustering of a network based on motif $M_7$.  For a given motif $M$, our
framework aims to find a set of nodes $S$ that minimizes motif conductance,
$\phi_M(S)$, which we define as the ratio of the number of motifs cut (filled
triangles cut) to the minimum number of nodes in instances of the motif in
either $S$ or $\bar{S}$~\cite{supp}.  In this case, there is one motif cut.
{\bf C:}
The higher-order network clustering framework.  Given a graph and a motif of
interest (in this case, $M_7$), the framework forms a motif adjacency matrix
($W_M$) by counting the number of times two nodes co-occur in an instance of the
motif.  An eigenvector of a Laplacian transformation of the motif adjacency
matrix is then computed.  The ordering $\sigma$ of the nodes provided by the
components of the eigenvector~\cite{note-spectral-ordering} produces nested sets
$S_r = \{\sigma_1, \ldots, \sigma_r\}$ of increasing size $r$.
We prove that the set $S_r$ with the smallest motif-based conductance,
$\phi_M(S_r)$, is a near-optimal higher-order cluster~\cite{supp}.
}
\label{fig:basics}
\end{figure}

\clearpage
\begin{figure}[t]
\centering
\includegraphics[width=0.9\textwidth]{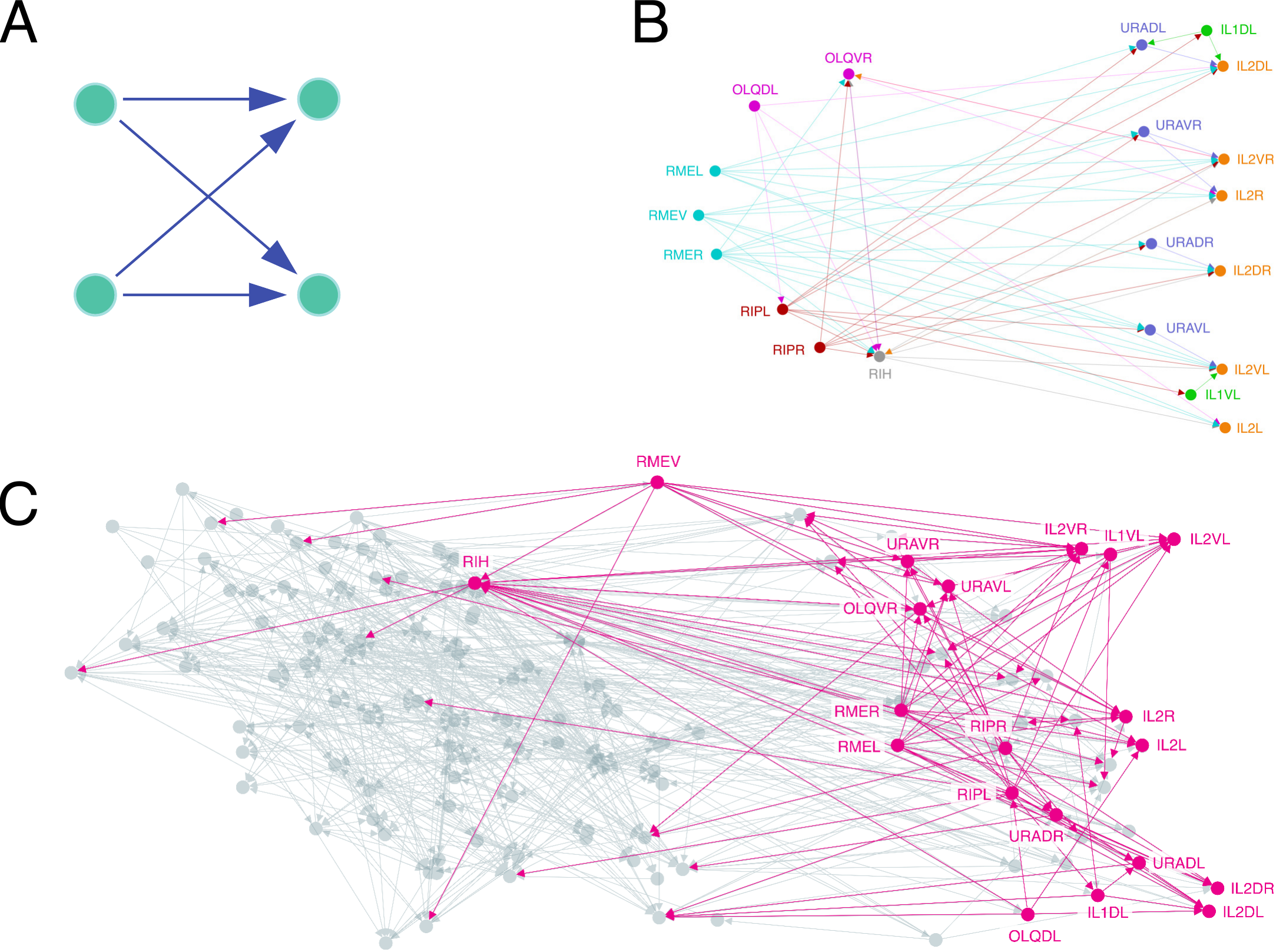}
\caption{%
{\bf Higher-order cluster in the \emph{C.~elegans} neuronal network~\cite{kaiser2006nonoptimal}.}
{\bf A:}
The 4-node ``bi-fan" motif, which is over-expressed in neuronal
networks~\cite{milo2002network}.  Intuitively, this motif describes a
cooperative propagation of information from the nodes on the left to the nodes
on the right.
{\bf B:}
The best higher-order cluster in the \emph{C.~elegans} frontal neuronal network
based on the motif in (A).  The cluster contains three ring motor neurons
(RMEL/V/R; cyan) with many outgoing connections, serving as the source of
information; six inner labial sensory neurons (IL2DL/VR/R/DR/VL; orange) with
many incoming connections, serving as the destination of information; and four
URA neurons (purple) acting as intermediaries.  These RME neurons have been
proposed as pioneers for the nerve ring~\cite{riddle1997celegans}, while the IL2
neurons are known regulators of nictation~\cite{lee2012nictation}, and the
higher-order cluster exposes their organization.  The cluster also reveals that
RIH serves as a critical intermediary of information processing.  This neuron
has incoming links from all three RME neurons, outgoing connections to five of
the six IL2 neurons, and the largest total number of connections of any neuron
in the cluster.
{\bf C:}
Illustration of the higher-order cluster in the context of the entire network.
Node locations are the true two-dimensional spatial embedding of the neurons.
Most information flows from left to right, and we see that RME/V/R/L and RIH
serve as sources of information to the neurons on the right.
}
\label{fig:celegans}
\end{figure}

\clearpage
\begin{figure}[t]
\centering
\includegraphics[width=\textwidth]{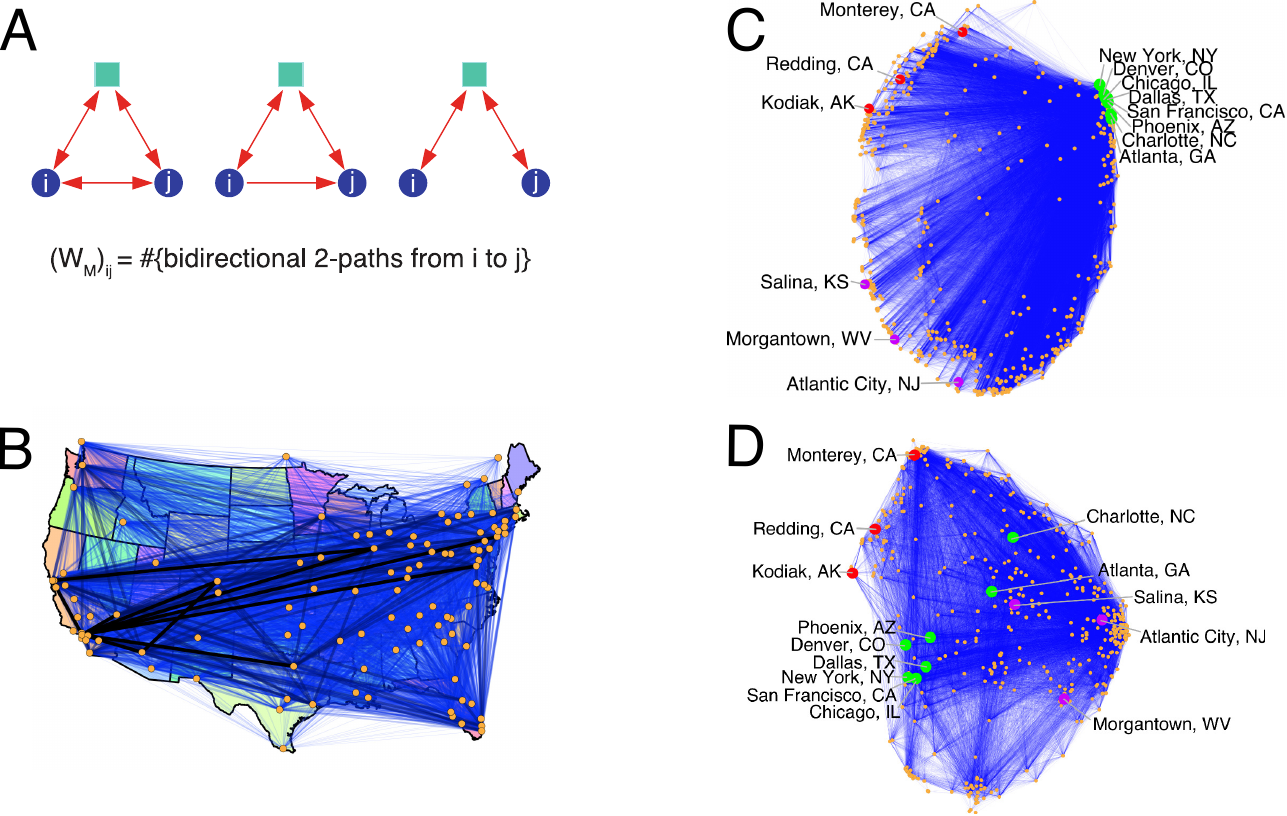}
\caption{%
{\bf Higher-order spectral analysis of a network of airports in Canada and the United States~\cite{frey2007clustering}.}
{\bf A:}
The three higher-order structures used in our analysis.  Each motif is
``anchored'' by the blue nodes $i$ and $j$, which means our framework only seeks
to cluster together the blue nodes.  Specifically, the motif adjacency matrix
adds weight to the $(i, j)$ edge based on the number of third intermediary nodes
(green squares).  The first two motifs correspond to highly-connected cities and
the motif on the right connects non-hubs to non-hubs.
{\bf B:}
The top 50 most populous cities in the United States which correspond to nodes
in the network.  The edge thickness is proportional to the weight in the motif
adjacency matrix $W_M$.  The thick, dark lines indicate that large weights
correspond to popular mainline routes.
{\bf C:}
Embedding of nodes provided by their corresponding components of the first two
non-trivial eigenvectors of the normalized Laplacian for $W_M$.  The marked
cities are eight large U.S. hubs (green), three West coast non-hubs (red), and
three East coast non-hubs (purple).  The primary spectral coordinate (left to
right) reveals how much of a hub the city is, and the second spectral coordinate
(top to bottom) captures West-East geography~\cite{supp}.
{\bf D:}
Embedding of nodes provided by their corresponding components in the first two
non-trivial eigenvectors of the standard, edge-based (non-higher-order)
normalized Laplacian.  This method does not capture the hub and geography found
by the higher-order method.  For example, Atlanta, the largest hub, is in the
center of the embedding, next to Salina, a non-hub.
}
\label{fig:airports}
\end{figure}

%% file: supplementary-tex/SM-theory.tex

We now cover the background and theory for deriving and understanding the method
presented in the main text.  We will start by reviewing the graph Laplacian and
cut and volume measures for sets of vertices in a graph.  We then define network
motifs in Section~\ref{sec:motif_def} and generalizes the notions of cut and
volume to motifs.  Our new theory is presented in
Section~\ref{sec:motif_cheeger} and then we summarize some extensions of the
method.  Finally, we relate our method to existing methods for directed
graph clustering and hypergraph partitioning.

\subsection{Review of the graph Laplacian for weighted, undirected graphs}

Consider a weighted, undirected graph $G = (V, E)$, with $\lvert V \rvert = n$.
Further assume that $G$ has no isolated nodes.  Let $W$ encode the weights, of
the graph, \emph{i.e.},
$W_{ij} = W_{ji} = \text{weight of edge (i, j)}$.
The diagonal degree matrix $D$ is defined as $D_{ii} = \sum_{j=1}^{n} W_{ij}$,
and the graph Laplacian is defined as $L = D - W$.  We now relate these matrices
to the conductance of a set $S$, $\newcond{G}{S}$:
\begin{eqnarray}
\newcond{G}{S} &=& \newcut{G}{S, \bar{S}} / \min(\newvol{G}{S}, \newvol{G}{\bar{S}}), \\
\newcut{G}{S, \bar{S}} &=& \sum_{i \in S,\;j \in \bar{S}} W_{ij}, \\
\newvol{G}{S} &=& \sum_{i \in S} D_{ii}
\end{eqnarray}
Here, $\bar{S} = V \backslash S$.
(Note that conductance is a symmetric measure in $S$ and $\bar{S}$,
\emph{i.e.}, $\newcond{G}{S} = \newcond{G}{\bar{S}}$.)
Conceptually, the cut and volume measures are defined as follows:
\begin{eqnarray}\label{eqn:cutvol_words}
\newcut{G}{S, \bar{S}} &=& \text{weighted sum of weights of edges that are cut} \label{eqn:cut_words} \\
\newvol{G}{S} &=& \text{weighted number of edge end points in $S$} \label{eqn:vol_words}
\end{eqnarray}
Since we have assumed $G$ has no isolated nodes, $\newvol{G}{S} > 0$.
If $G$ is disconnected, then for any connected component $C$, $\newcond{G}{C} = 0$.
Thus, we usually consider breaking $G$ into connected components as 
a pre-processing step for algorithms that try to find low-conductance sets.

We now relate the cut metric to a quadratic form on $L$.  Later, we will derive
a similar form for a motif cut measure.  Note that for any vector
$y \in \mathbb{R}^{n}$,
\begin{equation}
y^TLy = \sum_{(i, j) \in E}w_{ij}(y_i - y_j)^2.
\end{equation}
Now, define $x$ to be an indicator vector for a set of nodes $S$
\emph{i.e.}, $x_i = 1$ if node $i$ is in $S$ and $x_i = 0$ if node $i$ is in
$\bar{S}$.  Note that if an edge $(i, j)$ is cut, then $x_i$ and $x_j$ take
different values and $(x_i - x_j)^2 = 1$; otherwise, $(x_i - x_j)^2 = 0$.  Thus,
\begin{equation}\label{eqn:quad_cut}
x^TLx = \newcut{G}{S, \bar{S}}.
\end{equation}

\subsection{Definition of network motifs}\label{sec:motif_def}

\begin{figure}[htb]
\centering
\includegraphics[width=0.7\columnwidth]{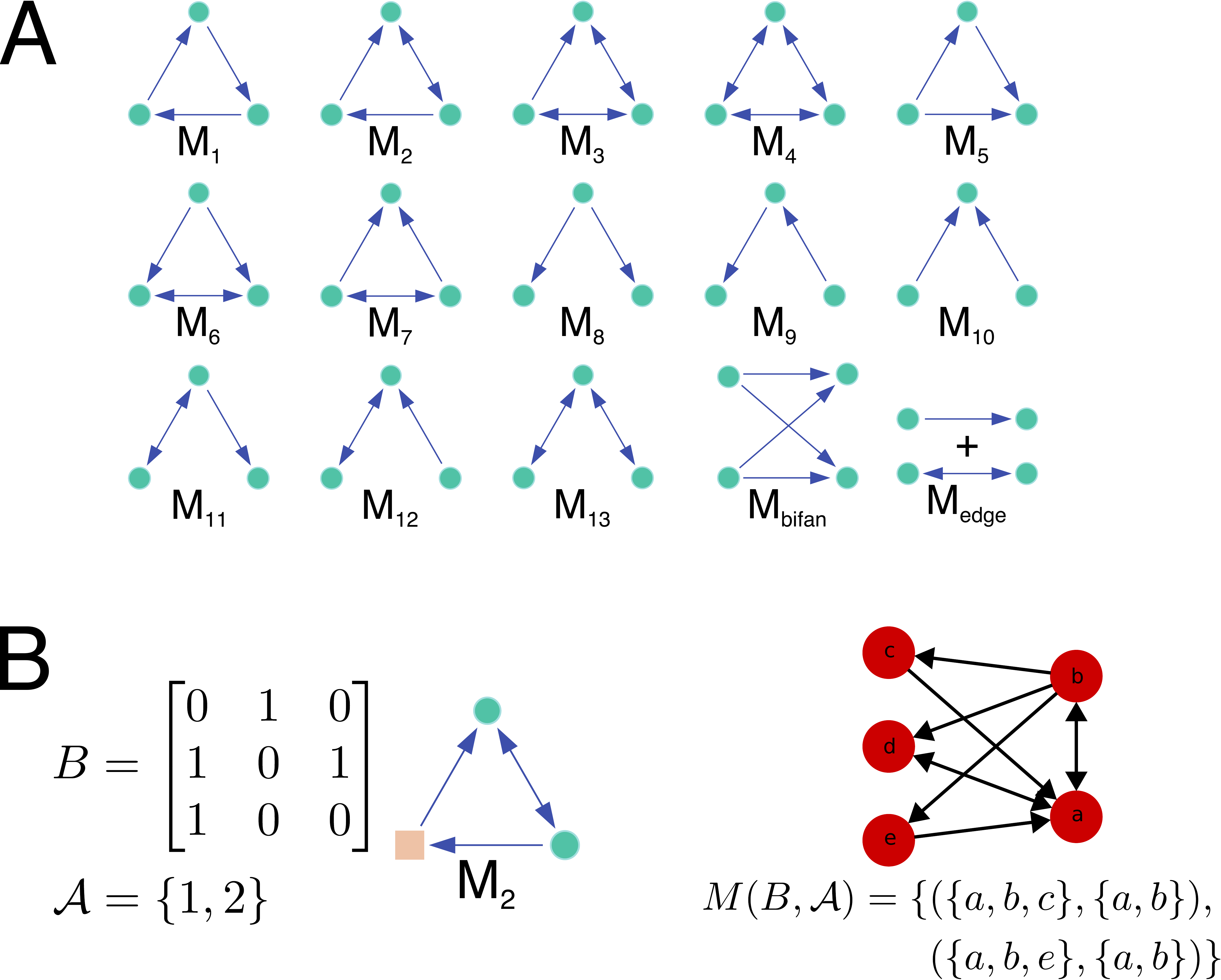}
\caption{%
{\bf A:}
Illustration of network motifs used throughout the main text and supplementary
material.  The motif $\medge$ is used to represent equivalence to undirected the
graph.
{\bf B:}
Diagram of motif definitions.  The motif is defined by a binary matrix $B$ and
an anchor set of nodes.  The figure shows an anchored version of motif
$\motiftype{3}{2}$ with anchors on the nodes that form the bi-directional edge.
There are two instances of the motif in the graph on the right.  Note that
$(\{a, b, d\}, \{a, b\})$ is not included in the set of motif instances because
the induced subgraph on the nodes $a$, $b$, and $d$ is not isomorphic to the
graph defined by $B$.
}
\label{fig:motif_defs}
\end{figure}

We now define network motifs as used in our work.  We note that there are
alternative definitions in the literature~\cite{milo2002network}.  We consider
motifs to be a pattern of edges on a small number of nodes (see
Figure~\ref{fig:motif_defs}).  Formally, we define a motif on $k$ nodes by a
tuple $(B, \anchorset)$, where $B$ is a $k \times k$ binary matrix and
$\anchorset \subset \{1, 2, \ldots, k\}$ is a set of anchor nodes.  The matrix
$B$ encodes the edge pattern between the $k$ nodes, and $\anchorset$ labels a
relevant subset of nodes for defining motif conductance.  In many cases,
$\anchorset$ is the entire set of nodes.  Let $\chi_{\anchorset}$ be a selection
function that takes the subset of a $k$-tuple indexed by $\anchorset$, and let
$\text{set}(\cdot)$ be the operator that takes an (ordered) tuple to an
(unordered) set.  Specifically,
\[
\text{set}((v_1, v_2, \ldots, v_k)) = \{v_1, v_2, \ldots, v_k\}.
\]
The set of motifs in an unweighted (possibly directed) graph with
adjacency matrix $A$, denoted $M(B, \anchorset)$, is defined by
\begin{equation}\label{eqn:motif_set}
M(B, \anchorset) =
\{
(\text{set}(\vect{v}), \text{set}(\chi_{\anchorset}(\vect{v})) ) \;\mid\;
 \vect{v} \in V^k,\quad v_1, \ldots, v_k \text{ distinct},\quad A_{\vect{v}} = B
\},
\end{equation}
where $A_{\vect{v}}$ is the $k \times k$ adjacency matrix on the subgraph
induced by the $k$ nodes of the \emph{ordered} vector $\vect{v}$.
Figure~\ref{fig:motif_defs} illustrates these definitions.  The set operator is
a convenient way to avoid duplicates when defining $M(B, \anchorset)$ for motifs
exhibiting symmetries.  Henceforth, we will just use $\minstance$ to denote
$(\text{set}(\vect{v}), \text{set}(\chi_{\anchorset}(\vect{v})))$ when
discussing elements of $M(B,\anchorset)$.  Furthermore, we call any
$\minstance \in M(B, \anchorset)$ a \emph{motif instance}.  When $B$ and
$\anchorset$ are arbitrary or clear from context, we will simply denote the
motif set by $M$.

We call motifs where $\anchornodes = \vect{v}$ \emph{simple motifs} and motifs
where $\anchornodes \neq \vect{v}$ \emph{anchored motifs}.  Motif analysis in
the literature has mostly analyzed simple motifs~\cite{alon2007network}.
However, the anchored motif provides us with a more general framework, and we
use an anchored motif for the analysis of the transportation reachability
network.

Often, a distinction is made between a \emph{functional} and a \emph{structural}
motif~\cite{sporns2004motifs} (or a subgraph and an induced
subgraph~\cite{inokuchi2000apriori}) to distinguish whether a motif specifies
simply the existence of a set of edges (functional motif or subgraph) or the
existence and non-existence of edges (structural motif or induced subgraph).  By
the definition in Equation~\ref{eqn:motif_set}, we refer to structural motifs in
this work.  Note that functional motifs consist of a set of structural motifs.
Our clustering framework allows for the simultaneous consideration of several
motifs (see Section~\ref{sec:multiple}), so we have not lost any generality in
our definitions.

\subsection{Definition of motif conductance}\label{sec:motif_cond}

Recall that the key definitions for defining conductance are the notions of cut
and volume.  For an unweighted graph, these are
\begin{eqnarray}
\newcond{G}{S, \bar{S}} &=& \newcut{G}{S, \bar{S}} / \min(\newvol{G}{S}, \newvol{G}{\bar{S}}), \\
\newcut{G}{S, \bar{S}} &=& \text{number of edges cut}, \\
\newvol{G}{S} &=& \text{number of edge end points in $S$}.
\end{eqnarray}
Our conceptual definition of motif conductance simply replaces an edge with a
motif instance of type $M$:
\begin{eqnarray}
\newmotifcond{G}{M}{S} &=& \newmotifcut{G}{M}{S, \bar{S}} / \min(\newmotifvol{G}{M}{S}, \newmotifvol{G}{M}{\bar{S}}), \\
\newmotifcut{G}{M}{S, \bar{S}} &=& \text{number of motif instances cut}, \\
\newmotifvol{G}{M}{S} &=& \text{number of motif instance end points in $S$}.
\end{eqnarray}

We say that a motif instance is cut if there is at least one anchor node in $S$
and at least one anchor node in $\bar{S}$.  We can formalize this when given a
motif set $M$ as in Equation~\ref{eqn:motif_set}:
\begin{eqnarray}
\newmotifcut{G}{M}{S, \bar{S}} &=& \sum_{\minstance \in M}\indicator{\exists\; i, j  \in \anchornodes \;\mid\; i \in S, j \in \bar{S}} \label{eqn:motif_cut},  \\
\newmotifvol{G}{M}{S} &=& \sum_{\minstance \in M} \sum_{i \in \anchornodes} \indicator{i \in S},
\end{eqnarray}
where $\indicator{s}$ is the truth-value indicator function on $s$, \emph{i.e.},
$\indicator{s}$ takes the value $1$ if the statement $s$ is true and $0$
otherwise.  Note that Equation~\ref{eqn:motif_cut} makes explicit use of the
anchor set $\anchorset$.  The motif cut measure only counts an instance of a
motif as cut if the anchor nodes are separated, and the motif volume counts the
number of anchored nodes in the set.  However, two nodes in an achor set may a
part of several motif instances.  Specifically, following the definition in
Equation~\ref{eqn:motif_set}, there may be many different $\vect{v}$ with the
same $\anchornodes$, and the nodes in $\anchornodes$ still get counted
proportional to the number of motif instances.

\subsection{Definition of the motif adjacency matrix and motif Laplacian}

Given an unweighted, directed graph and a motif set $M$, we conceptually
define the motif adjacency matrix by
\begin{equation}\label{eqn:informal_weighting}
\motifweightedij = \text{number of motif instances in $M$ where $i$ and $j$
  participate in the motif}.
\end{equation}
Or, formally,
\begin{equation}\label{eqn:formal_weighting}
\motifweightedij = \sum_{\minstance \in M} \indicator{\{i, j\} \subset \chi_{\anchorset}(\vect{v})},
\end{equation}
for $i \neq j$.  Note that weight is added to $\motifweightedij$ only if $i$ and $j$
appear in the anchor set.  This is important for the transportation reachability
network analyzed in the main text and in Section~\ref{sec:airports}, where weight is
added between cities $i$ and $j$ based on the number of intermediary cities that
can be traversed between them.

Next, we define the motif diagonal degree matrix by $(D_M)_{ii} =
\sum_{j=1}^{n} \motifweightedij$ and the motif Laplacian as $L_M = D_M - W_M$.
Finally, the normalized motif Laplacian is
$\normmotiflap = D_M^{-1/2}L_MD_M^{-1/2} = I - D_M^{-1/2}W_MD_M^{-1/2}$.
The theory in the next section will examine quadratic forms $L_M$ and derive the
main clustering method that uses an eigenvector of $\normmotiflap$.

\subsection{Algorithm for finding a single cluster}

We are now ready to describe the algorithm for finding a single cluster in a
graph.  The algorithm finds a partition of the nodes into $S$ and $\bar{S}$.
The motif conductance is symmetric in the sense that $\newmotifcond{G}{M}{S}
= \newmotifcond{G}{M}{\bar{S}}$, so either set of nodes ($S$ or $\bar{S}$) could
be interpreted as a cluster.  However, in practice, it is common that one set is
substantially smaller than the other.  We consider this smaller set to represent
a module in the network.  The algorithm is based on the Fiedler
partition~\cite{chung2007four} of the motif weighted adjacency matrix and is
presented below in Algorithm~\ref{alg:motif_fiedler}.%
\footnote{An implementation of Algorithm~\ref{alg:motif_fiedler} is available
in SNAP.  See \url{http://snap.stanford.edu/higher-order/}.}

\RestyleAlgo{boxruled}
\begin{algorithm}[H]\label{alg:motif_fiedler}
  \DontPrintSemicolon
  \KwIn{Directed, unweighted graph $G$ and motif $M$}
  \KwOut{Motif-based cluster (subset of nodes in $G$)}
   $(W_M)_{ij} \leftarrow \text{number of instances of $M$ that contain nodes $i$ and $j$}.$\;
   $G_M \leftarrow $ weighted graph induced $W_M$\;
   $D_M \leftarrow $ diagonal matrix with $(D_M)_{ii} = \sum_{j} (W_M)_{ij}$\;
   $z \leftarrow $ eigenvector of second smallest eigenvalue for $\normmotiflap = I - D_M^{-1/2}W_MD_M^{-1/2}$\;
   $\sigma_i \leftarrow$ to be index of $D_M^{-1/2}z$ with $i$th smallest value\;
   \tcc{Sweep over all prefixes of $\sigma$}
   $S \leftarrow \arg\min_{l} \newcond{G_M}{S_l}$, where $S_l = \{\sigma_1, \ldots, \sigma_l\}$ \;
  \eIf{$\lvert S \rvert < \lvert \bar{S} \rvert$}{
      \Return $S$\;
  }{
      \Return $\bar{S}$\;
  }      
  \caption{Motif-based clustering algorithm for finding a single
  cluster.}
\end{algorithm}

It is often informative to look at all conductance values found from the sweep
procedure.  We refer to a plot of $\newcond{G_M}{S_l}$ versus $l$ as a \emph{sweep
profile plot}.  In the following subection, we show that when the motif has
three nodes, $\newcond{G_M}{S_l} = \newmotifcond{G}{M}{S_l}$.  In this case, the sweep profile
shows how motif conductance varies with the size of the sets in
Algorithm~\ref{alg:motif_fiedler}.

In the following subsection, we show that when the motif $M$ has three nodes, the
cluster satisfies $\newmotifcond{G}{M}{S} \le 4\sqrt{\phi^*}$, where $\phi^*$ is the
smallest motif conductance over all sets of nodes.  In other words, the cluster
is nearly optimal.  Later, we extend this algorithm to allow for signed,
colored, and weighted motifs and to simultaneously finding multiple clusters.

\subsection{Motif Cheeger inequality for network motifs with three nodes}
\label{sec:motif_cheeger}

We now derive the motif Cheeger inequality for simple three-node motifs, or, in
general, motifs with three anchor nodes.  The crux of this result is deriving a
relationship between the motif conductance function and the weighted motif
adjacency matrix, from which the Cheeger inequality is essentially a corollary.
For the rest of this section, we will use the following notation.  Given an
unweighted, directed $G$ and a motif $M$, the corresponding weighted graph
defined by Equation~\ref{eqn:formal_weighting} is denoted by $G_M$.

The following Lemma relates the motif volume to the volume in the weighted
graph.  This lemma applies to any anchor set $\anchorset$ consisting of
at least two nodes.  For our main result, we will apply the lemma assuming
$\vert \anchorset \vert = 3$.  However, we will apply the lemma more generally
when discussing four node motifs in Section~\ref{sec:fournode}.
\begin{lemma}\label{lem:motif_vol}
Let $G = (V, E)$ be a directed, unweighted graph and let $G_M$ be the weighted
graph for a motif on $k$ nodes and $\vert \anchorset \vert \ge 2$ anchor nodes.
Then for any $S \subset V$,
\[
\newmotifvol{G}{M}{S} = \frac{1}{\vert \anchorset \vert  - 1}\newvol{G_M}{S}
\]
\end{lemma}
\begin{proof}
Consider an instance $\minstance$ of a motif.  Let
$(u_1, \ldots, u_{\vert \anchorset \vert}) = \chi_{\anchorset}(\vect{v})$.
By Equation~\ref{eqn:formal_weighting}, $(W_M)_{u_1, j}$ is incremented by one
for $j = u_2, \ldots, u_{\vert \anchorset \vert}$.  Since 
$(D_M)_{u_1, u_1} = \sum_{j} {(W_M)_{u_1,j}}$,
the motif end point $u_1$ is counted $\vert \anchorset \vert - 1$ times.
\end{proof}

\clearpage
\begin{figure}[h]
\centering
\includegraphics[width=0.9\columnwidth]{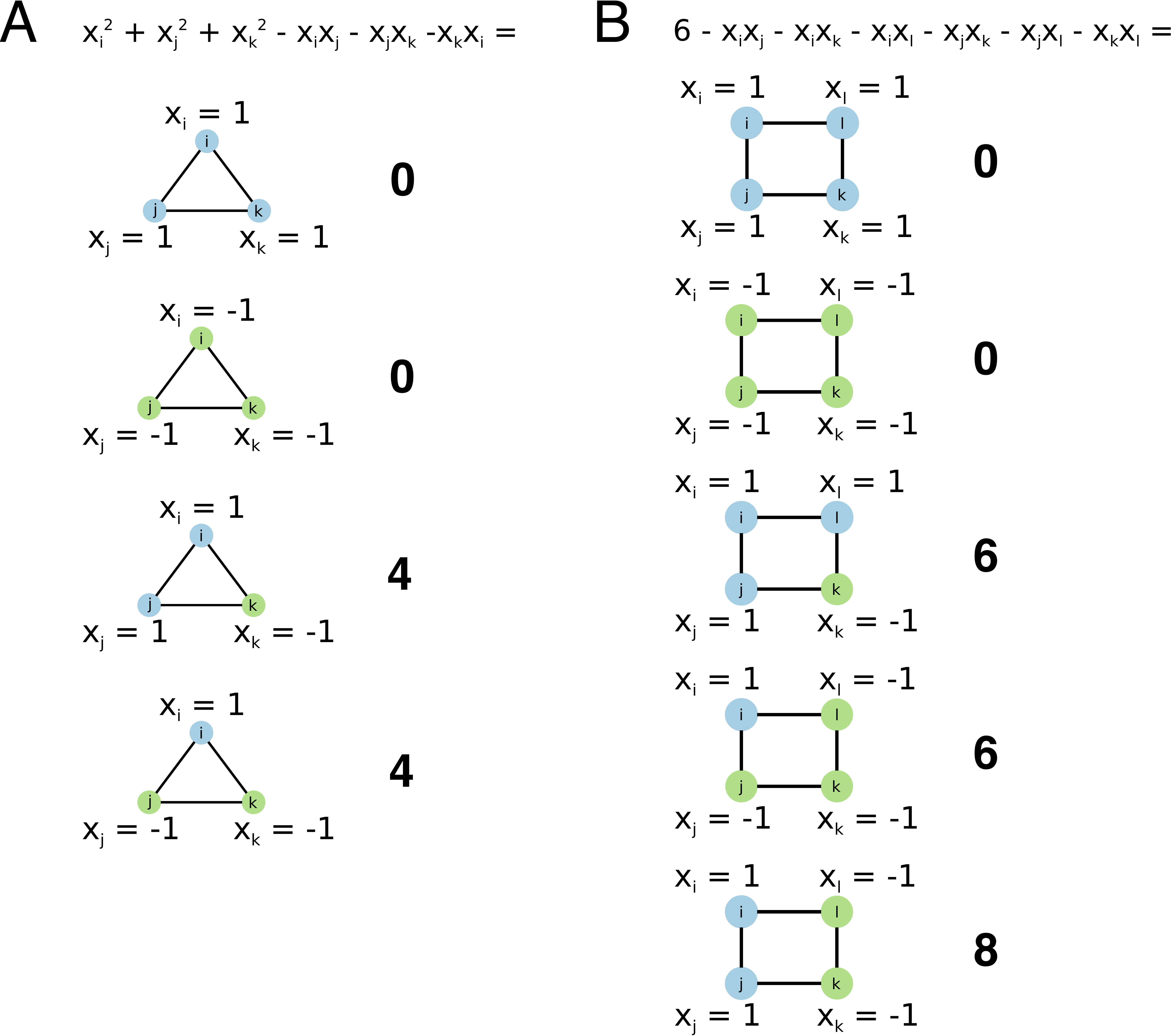}
\caption{%
Illustrations of the quadratic forms on indicator functions for set assignment.
Here, the blue nodes have assignment to set $S$ and the green nodes have
assignment to set $\bar{S}$.  The quadratic function gives the penalty for
cutting that motif.
{\bf A:}
Illustration of Equation~\ref{eqn:ind3}.  The quadratic form is proportional
to the indicator on whether or not the motif is cut.
{\bf B:}
Illustration of Equation~\ref{eqn:motif4_quadratic}.  The quadratic form is equal
to zero when all nodes are in the same set.  However, the form penalizes 2/2
splits more than 3/1 splits.
}
\label{fig:motif_cuts}
\end{figure}

The following lemma states that the truth value for determining whether three
binary variables in $\{-1, 1\}$ are not all equal is a quadratic function of the
variables (see Figure~\ref{fig:motif_cuts}).  Because this function is
quadratic, we will be able to relate motif cuts on three nodes to a quadratic
form on the motif Laplacian.

\begin{lemma}\label{lem:ind3}
Let $x_i, x_j, x_k \in \{-1, 1\}$.  Then
\begin{equation}\label{eqn:ind3}
4\cdot\indicator{x_i, x_j, x_k \text{ not all the same}}
= x_i^2 + x_j^2 + x_k^2 - x_ix_j - x_jx_k - x_kx_i. \nonumber
\end{equation}
\end{lemma}

It will be easier to derive our results with binary indicator variables taking
values in $\{-1, 1\}$.  However, in terms of the quadratic form on the
Laplacian, we have already seen how indicator vectors taking values in $\{0,
1\}$ relate to the cut value (Equation~\ref{eqn:quad_cut}).  The following lemma
shows that the $\{0, 1\}$ and $\{-1, 1\}$ indicator vectors are equivalent, up
to a constant, for defining the cut measure in terms of the Laplacian.

\begin{lemma}\label{lem:ind_pm}
Let $z \in \{0, 1\}^{n}$ and define $x$ by $x_{i} = 1$ if $z_i = 1$ and $x_{i} =
-1$ if $z_i = 0$.  Then for any graph Laplacian $L = D - W$, $4z^TLz = x^TLx$.
\end{lemma}
\begin{proof}
\[
x^TLx
= \sum_{(i, j) \in E}W_{ij}(x_{i} - x_{j})^2
= \sum_{(i, j) \in E}W_{ij}4(z_{i} - z_{j})^2
= 4z^TLz.
\]
\end{proof}

The next lemma contains the essential result that relates motif cuts in the
original graph $G$ to weighted edge cuts in $G_M$.  In particular, the lemma
shows that the motif cut measure is proportional to the cut on the weighted
graph defined in Equation~\ref{eqn:informal_weighting} when there are three
anchor nodes.
\begin{lemma}\label{lem:motif_cut}
Let $G = (V, E)$ be a directed, unweighted graph and let $G_M$ be the weighted
graph for a motif with $\vert \anchorset \vert = 3$.  Then for any $S \subset
V$,
\[
\newmotifcut{G}{M}{S, \bar{S}} = \frac{1}{2}\newcut{G_M}{S, \bar{S}}
\]
\end{lemma}
\begin{proof}
Let $x \in \{-1, 1\}^{n}$ be an indicator vector of the node set $S$.
\begin{eqnarray*}
4 \cdot \newmotifcut{G}{M}{S, \bar{S}}
&=& \sum_{(\vect{v}, \{i, j, k\}) \in M} 4\cdot\indicator{x_i, x_j, x_k \text{ not all the same}} \\
&=& \sum_{(\vect{v}, \{i, j, k\}) \in M} \left(x_i^2 + x_j^2 + x_k^2\right) - \left(x_ix_j + x_jx_k + x_kx_i\right) \\
&=& \frac{1}{2}x^TD_Mx  - \frac{1}{2}x^TW_Mx \\
&=& \frac{1}{2}x^TL_Mx \\
&=& 2\cdot\newcut{G_M}{S, \bar{S}}.
\end{eqnarray*}
The first equality follows from the definition of cut motifs
(Equation~\ref{eqn:motif_cut}).  The second equality follows from
Lemma~\ref{lem:ind3}.  The third equality follows from Lemma~\ref{lem:motif_vol}
and Equation~\ref{eqn:formal_weighting}.  The fourth equality follows from the
definition of $L_M$.  The fifth equality follows from Lemma~\ref{lem:ind_pm}.
\end{proof}

We are now ready to prove our main result, namely that motif conductance on the
original graph $G$ is equivalent to conductance on the weighted graph $G_M$ when
there are three anchor nodes.  The result is a consequence of the volume and cut
relationships provided by Lemmas~\ref{lem:motif_vol}~and~\ref{lem:motif_cut}.

\begin{theorem}\label{thm:motif_cond}
Let $G = (V, E)$ be a directed, unweighted graph and let $W_M$ be the weighted
adjacency matrix for any motif with $\vert \anchorset \vert = 3$.  Then for any
$S \subset V$,
\[
\newmotifcond{G}{M}{S} = \newcond{G_M}{S}
\]
In other words, when the number of anchor nodes is $3$, the motif conductance is
equal to the conductance on the weighted graph defined by
Equation~\ref{eqn:informal_weighting}.
\end{theorem}
\begin{proof}
When $\vert \anchorset \vert = 3$, the motif cut and motif volume are both equal
to half the motif cut and motif volume measures by
Lemmas~\ref{lem:motif_vol}~and~\ref{lem:motif_cut}.
\end{proof}

For any motif with three anchor nodes, conductance on the weighted graph is
equal to the motif conductance.  Because of this, we can use results from
spectral graph theory for weighted graphs~\cite{chung2007four} and re-interpret
the results in terms of motif conductance.  In particular, we get the following
``motif Cheeger inequality".

\begin{theorem}\label{thm:motif_cheeger}
\textnormal{\textbf{Motif Cheeger Inequality}}.
Suppose we use Algorithm~\ref{alg:motif_fiedler} to find a low-motif conductance set $S$.
Let $\phi_* = \min_{S'} \newmotifcond{G}{M}{S'}$ be the optimal motif conductance over
any set of nodes $S'$.  Then
\begin{enumerate}
\item $\newmotifcond{G}{M}{S} \le 4\sqrt{\phi^*}$ and
\item $\phi^* \ge \lambda_2 / 2$
\end{enumerate}
\end{theorem}
\begin{proof}
The result follows from Theorem~\ref{thm:motif_cond} and the standard Cheeger
ineqaulity~\cite{chung2007four}.
\end{proof}

The first part of the result says that the set of nodes $S$ is within a
quadratic factor of optimal.  This provides the mathematical guarantees that our
procedure finds a good cluster in a graph, if one exists.  The second result
provides a lower bound on the optimal motif conductance in terms of the
eigenvalue.  We use this bound in our analysis of a food web (see
Section~\ref{sec:foodweb}) to show that certain motifs do not provide good
clusters, regardless of the procedure to select $S$.

\subsection{Discussion of motif Cheeger inequality for network motifs with four or more nodes}
\label{sec:fournode}

Analogs of the indicator function in Lemma~\ref{lem:ind3} for four or more
variables are not quadratic.  Subsequently, for motifs with $\vert \anchorset
\vert > 3$, we no longer get the motif Cheeger inequalities guaranteed by
Theorem~\ref{thm:motif_cheeger}.  That being said, solutions found by
motif-based partitioning approximate a related value of conductance.  We now
provide the details.

We begin with a lemma that shows a functional form for four binary variables
taking values in $\{-1, 1\}$ to not all be equal.  We see that it is quartic,
not quadratic.
\begin{lemma}\label{lem:ind4_1}
Let $x_i, x_j, x_k, x_l \in \{-1, 1\}$.
Then the indicator function on all four elements not being equal is
\begin{eqnarray}
&& 8 \cdot \indicator{x_i, x_j, x_k, x_l \text{ not all the same}} \\
&=& \left(7 - x_ix_j - x_ix_k - x_ix_l - x_jx_k - x_jx_l - x_kx_l - x_ix_jx_kx_l\right) \nonumber.
\end{eqnarray}
\end{lemma}

We \emph{almost} have a quadratic form, if not for the quartic term
$x_ix_jx_kx_l$.  However, we could use the following related quadratic form:
\begin{align}
& 6 - x_ix_j - x_ix_k - x_ix_{l} - x_jx_k - x_jx_{l} - x_kx_{l} \nonumber \\
&= \left\{
     \begin{array}{ll}
       0 &  x_i, x_j, x_k, x_l \text{ are all the same} \\
       6 &  \text{exactly three of } x_i, x_j, x_k, x_l \text{ are the same} \\
       8 &  \text{exactly two of } x_i, x_j, x_k, x_l \text{ are $-1$}.
     \end{array}
   \right.\label{eqn:motif4_quadratic}
\end{align}

The quadratic still takes value $0$ if all four entries are the same, and takes
a non-zero value otherwise.  However, the quadratic takes a larger value if
exactly two of the entries are $-1$.  Figure~\ref{fig:motif_cuts} illustrates
this idea.  From this, we can provide an analogous statement to
Lemma~\ref{lem:motif_cut} for motifs with $\vert \anchorset \vert = 4$.

\begin{lemma}\label{lem:motif_cut4}
Let $G = (V, E)$ be a directed, unweighted graph and let $G_M$ be the weighted
graph for a motif with $\vert \anchorset \vert = 4$.  Then for any $S \subset
V$,
\[
\newmotifcut{G}{M}{S, \bar{S}} =
\frac{1}{3}\newcut{G_M}{S, \bar{S}} - 
\sum_{(\vect{v}, \{i, j, k, l\}) \in M}
\frac{1}{3} \cdot \indicator{\text{exactly two of } i, j, k, l \text{ in } S}
\]
\end{lemma}
\begin{proof}
Let $x \in \{-1, 1\}^n$ be an indicator vector of the node set $S$.
\begin{eqnarray*}
&& 6 \cdot \newmotifcut{G}{M}{S, \bar{S}} +
\sum_{(\vect{v}, \{i, j, k, l\}) \in M}
2 \cdot \indicator{\text{exactly two of } i, j, k, l \text{ in } S} \\
&=& \sum_{(\vect{v}, \{i, j, k, l\}) \in M} 6 - x_ix_j - x_ix_k - x_ix_l - x_jx_k - x_jx_l - x_kx_l \\
&=& \sum_{(\vect{v}, \{i, j, k, l\}) \in M} \frac{3}{2}
\left(x_i^2 + x_j^2 + x_k^2 + x_l^2\right) -
\left(x_ix_j + x_ix_k + x_ix_l + x_jx_k + x_jx_l + x_kx_l\right) \\
&=& \frac{1}{2}x^TD_Mx  - \frac{1}{2}x^TW_Mx \\
&=& \frac{1}{2}x^TL_Mx \\
&=& 2\cdot\newcut{G_M}{S, \bar{S}}.
\end{eqnarray*}
The first equality follows from
Equations~\ref{eqn:motif_cut}~and~\ref{eqn:motif4_quadratic}.  The third
equality follows from Lemma~\ref{lem:motif_vol}.  The fourth equality follows
from the definition of $L_M$.  The fifth equality follows from
Lemma~\ref{lem:ind_pm}.
\end{proof}

With four anchor nodes, the motif cut in $G$ is slightly different than the
weighted cut in the weighted graph $G_M$. However, Lemma~\ref{lem:motif_vol}
says that the motif volume in $G$ is still the same as the weighted volume in
$G_M$.  We use this to derive the following result.

\begin{theorem}\label{thm:motif_cond4}
Let $G = (V, E)$ be a directed, unweighted graph and let $W_M$ be the weighted
adjacency matrix for any motif with $\vert \anchorset \vert = 4$.  Then for any
$S \subset V$,
\[
\newmotifcond{G}{M}{S} =
\newcond{G_M}{S} -
\frac{
\sum_{(\vect{v}, \{i, j, k, l\}) \in M} \indicator{\text{exactly two of } i, j, k, l \text{ in } S}
}{
\newvol{G_M}{S}
}
\]
In other words, when there are four anchor nodes, the weighting scheme in
Equation~\ref{eqn:informal_weighting} models the exact conductance with an
additional penalty for splitting the four anchor nodes into two groups of two.
\end{theorem}
\begin{proof}
This follows from Lemmas~\ref{lem:motif_vol}~and~\ref{lem:motif_cut4}.
\end{proof}

To summarize, we still get a Cheeger inequality from the weighted graph, but it
is in terms of a penalized version of the motif conductance $\newmotifcond{G}{M}{S}$.
However, the penalty makes sense---if the group of four nodes is ``more split"
(2 and 2 as opposed to 3 and 1), the penalty is larger.  When $\vert \anchorset
\vert > 4$, we can derive similar penalized approximations to $\newmotifcond{G}{M}{S}$.

\subsection{Methods for simultaneously finding multiple clusters}

For clustering a network into $k > 2$ clusters based on motifs, we could
recursively cut the graph using the sweep procedure with some stopping
criterion~\cite{boley1998principal}.  For example, we could continue to cut the
largest remaining cluster until the graph is partitioned into some pre-specified
number of clusters.  We refer to this method as recursive bi-partitioning.

In addition, we can use the following method of Ng~\emph{et al.}~\cite{ng2002spectral}.

\RestyleAlgo{boxruled}
\begin{algorithm}[H]\label{alg:motif_ngetal}
  \DontPrintSemicolon
  \KwIn{Directed, unweighted graph $G$, motif $M$, number of clusters $k$}
  \KwOut{$k$ disjoint motif-based clusters}
   $(W_M)_{ij} \leftarrow \text{number of instances of $M$ that contain nodes $i$ and $j$}.$\;
   $D_M \leftarrow $ diagonal matrix with $(D_M)_{ii} = \sum_{j} (W_M)_{ij}$\;
   $z_1, \ldots, z_k \leftarrow $ eigenvectors of $k$ smallest eigenvalues for $\normmotiflap = I - D_M^{-1/2}W_MD_M^{-1/2}$\;
   $Y_{ij} \leftarrow z_{ij} / \sqrt{\sum_{j=1}^{k} z_{ij}^2}$\;
   Embed node $i$ into $\mathbb{R}^k$ by taking the $i$th row of the matrix $Y$\;
   Run $k$-means clustering on the embedded nodes\;
  \caption{Motif-based clustering algorithm for finding several clusters.}
\end{algorithm}

This method does not have the same Cheeger-like guarantee on quality.  However,
recent theory shows that by replacing $k$-means with a different clustering
algorithm, there is a performance guarantee~\cite{lee2014multiway}.  While this
provides motivation, we use $k$-means for its simplicity and empirical success.

\subsection{Extensions of the method for simultaneously analyzing several network motifs}
\label{sec:multiple}

All of our results carry through when considering several motifs simultaneously.
In particular, suppose we are interested in clustering based on motif sets $M_1,
\ldots, M_q$ for $q$ different motifs.  Further suppose that we want to weight
the impact of some motifs more than other motifs.  Let $W_{M_j}$ be the weighted
adjacency matrix for motif $M_j$, $j = 1, \ldots, q$, and let $\alpha_j \ge 0$
be the weight of motif $M_j$, then we can form the weighted adjacency matrix
\begin{equation}
W_M = \sum_{j=1}^{q} \alpha_j W_{M_j}.
\end{equation}

Now, the cut and volume measures are simply weighted sums by linearity.
Suppose that the $M_j$ all have three anchor nodes and let $G_M$ be the
weighted graph corresponding to $W_M$.  Then
\[
\newcut{G_M}{S, \bar{S}} = \sum_{j=1}^{q}\alpha_j\newmotifcut{G}{M_j}{S, \bar{S}}, \quad
\newvol{G_M}{S} = \sum_{j=1}^{q}\alpha_j\newmotifvol{G}{M_j}{S},
\]
and Theorem~\ref{thm:motif_cheeger} applies to a weighted motif conductance
equal to
\[
\frac{
\sum_{j=1}^{q}\alpha_j\newmotifcut{G}{M_j}{S, \bar{S}}
}{
\min\left(
\sum_{j=1}^{q}\alpha_j\newmotifvol{G}{M_j}{S},
\sum_{j=1}^{q}\alpha_j\newmotifvol{G}{M_j}{\bar{S}}
\right)
}.
\]

\subsection{ Extensions of the method to signed, colored, and weighted motifs}

Our results easily generalize for signed networks.  We only have to generalize
Equation~\ref{eqn:motif_set} by allowing the adjacency matrix $B$ to be signed.
Extending the method for motifs where the edges or nodes are ``colored'' or ``labeled''
is similar.  If the edges are colored, then we again just allow the adjacency
matrix $B$ to capture this information.  If the nodes in the motif are colored,
we only count motif instances with the specified pattern.

We can also generalize the notions of motif cut and motif volume for ``weighted
motifs", \emph{i.e.}, each motif has an associated nonnegative weight.  Let
$\omega_{\minstance}$ be the weight of a motif instance.  Our cut and volume
metrics are then
\begin{eqnarray*}
\newmotifcut{G}{M}{S, \bar{S}} &=&
\sum_{\minstance \in M}
\omega_{\minstance}\indicator{\exists\; i, j  \in \anchornodes \;\mid\; i \in S, j \in \bar{S}},  \\
\newmotifvol{G}{M}{S} &=&
\sum_{\minstance \in M}
\omega_{\minstance}
\sum_{i \in \anchornodes} \indicator{i \in S}.
\end{eqnarray*}
Subsequently, we adjust the motif adjacency matrix as follows:
\begin{equation}
\motifweightedij = \sum_{\minstance \in M}
\omega_{\minstance} \indicator{\{i, j\} \subset \chi_{\anchorset}(\vect{v})}
\end{equation}


\subsection{ Connections to directed graph partitioning}

Our framework also provides a way to analyze methods for clustering directed
graphs.  Existing principled generalizations of undirected graph partitioning to
directed graph partitioning proceed from graph
circulations~\cite{chung2005laplacians} or random walks~\cite{boley2011commute}
and are difficult to interpret.  Our motif-based clustering framework provides a
simple, rigorous framework for directed graph partitioning.  For example,
consider the common heuristic of clustering the symmetrized graph $W = A + A^T$,
where $A$ is the (directed) adjacency matrix~\cite{malliaros2013clustering}.
Following Theorem~\ref{thm:motif_cond}, conductance-minimizing methods for
partitioning $W$ are actually trying to minimize a weighted sum of motif-based
conductances for the directed edge motif and the bi-directional edge motif:
\[
B_1 = \begin{bmatrix} 0 & 1 \\ 0 & 0 \end{bmatrix},\quad
B_2 = \begin{bmatrix} 0 & 1 \\ 1 & 0 \end{bmatrix},
\]
where both motifs are simple ($\anchorset = \{1, 2\}$).  If $W_1$ and $W_2$ are
the motif adjacency matrices for $B_1$ and $B_2$, then $A + A^T = W = W_1 +
2W_2$.  This weighting scheme gives a weight of two to bi-directional edges in
the original graph and a weight of one to uni-directional edges.

An alternative strategy for clustering a directed graph is to simply remove the
direction on all edges, treating bi-directional and uni-directional edges the
same.  The resulting adjacency matrix is equivalent to the motif adjacency
matrix for the bi-directional and uni-directional edges (without any relative
weighting).  Formally, $W = W_1 + W_2$.  We refer to this ``motif'' as $\medge$
(Figure~\ref{fig:motif_defs}), which will later provide a convenient notation
when discussing both motif-based clustering and edge-based clustering.

\subsection{ Connections to hypergraph partitioning}

Finally, we contextualize our method in the context of existing literature on
hypergraph partitioning.  The problem of partitioning a graph based on
relationships between more than two nodes has been studied in hypergraph
partitioning~\cite{karypis1999multilevel}, and we can interpret motifs as
hyperedges in a graph. In contrast to existing hypergraph partitioning problems,
we \emph{induce} the hyperedges from motifs rather than take the hyperedges as
given \emph{a priori}.  The goal with our analysis of the Florida Bay food web,
for example, was to find which hyperedge sets (induced by a motif) provide a
good clustering of the network (see Section~\ref{sec:foodweb}).

In general, our motif-based spectral clustering methodology falls into the area
of encoding a hypergraph partitioning problem by a graph partitioning
problem~\cite{agarwal2006higher,zhou2006learning}.  With simple motifs on $k$
nodes, the motif Laplacian $\mathcal{L}_M$ formed from $W_M$
(Equation~\ref{eqn:formal_weighting}) is a special case of the Rodr\'{i}guez
Laplacian~\cite{agarwal2006higher,rodriguez2002laplacian} for $k$-regular
hypergraphs.  The motif Cheeger inequality we proved
(Theorem~\ref{thm:motif_cheeger}) explains why this Laplacian is appropriate for
$3$-regular hypergraphs.  Specifically, it respects the standard cut and volume
metrics for graph partitioning.

%% file: supplementary-tex/SM-complexity.tex

We now analyze the computation of the higher-order clustering method.  We first
provide a theoretical analysis of the computational complexity, which depends on
motif.  After, we empirically analyze the time to find clusters for triangular
motifs on a variety of real-world networks, ranging in size from a few hundred
thousand edges to nearly two billion edges.  Finally, we show that we can
practically compute the motif adjacency matrix for motifs up to size 9 on a
number of real-world networks.

\subsection{Analysis of computational complexity}

We now analyze the computational complexity of the algorithm presented in
Theorem~\ref{thm:motif_cheeger}.  Overall, the complexity of the algorithm is
governed by the computations of the motif adjacency matrix $W_M$, an
eigenvector, and the sweep cut procedure.  For simplicity, we assume that we can
access edges in a graph in $O(1)$ time and access and modify matrix entries in
$O(1)$ time.  Let $m$ and $n$ denote the number of edges in the graph.
Theoretically, the eigenvector can be computed in $O((m + n)(\log n)^{O(1)})$
time using fast Laplacian solvers~\cite{trevisan2015notes}.  For the sweep cut,
it takes $O(n\log n)$ to sort the indices given the eigenvector using a standard
sorting algorithm such is merge sort.  Computing motif conductance for each set
$S_r$ in the sweep also takes linear term.  In pratice, the sweep cut step takes
a small fraction of the total running time of the algorithm.  For the remainder
of the analysis, we consider the more nuanced issue of the time to compute
$W_M$.

The computational time to form $W_M$ is bounded by the time to find all
instances of the motif in the graph.  Naively, for a motif on $k$ nodes, we can
compute $W_M$ in $\Theta(n^k)$ time by checking each $k$-tuple of nodes.
Furthermore, there are cases where there are $\Theta(n^k)$ motif instances in
the graph, \emph{e.g.}, there are $\Theta(n^3)$ triangles in a complete graph.
However, since most real-world networks are sparse, we instead focus on the
complexity of algorithms in terms of the number of edges and the maximum degree
in the graph.  For this case, there are several efficient practical algorithms
for real networks with available software%
~\cite{demeyer2013index,houbraken2014index,wernicke2006efficient,wernicke2006fanmod,aberger2015emptyheaded}.

Theoretically, motif counting is efficient.  Here we consider four classes of
motifs: (1) triangles, (2) wedges (connected, non-triangle three-node motifs),
(3) four-node motifs, and (4) $k$-cliques.  Let $m$ be the number of edges in a
graph.  Latapy analyzed a number of algorithms for listing all triangles in an
undirected network, including an algorithm that has computational complexity
$\Theta(m^{1.5})$~\cite{latapy2008main}.  For a directed graph $G$, we can use
the following algorithm: (1) form a new graph $G_{\textnormal{undir}}$ by
removing the direction from all edges in $G$ (2) find all triangles in
$G_{\textnormal{undir}}$, (3) for every triangle in $G_{\textnormal{undir}}$,
check which directed triangle motif it is in $G$.  Since step 1 is linear and we
can perform the check in step 3 in $O(1)$ time, the same $\Theta(m^{1.5})$
complexity holds for directed networks.  This analysis holds regardless of the
structure of the networks.  However, additional properties of the network can
lead to improved algorithms.  For example, in networks with a power law degree
sequence with exponent greater than $7/2$, Berry \emph{et al.} provide a
randomized algorithm with expected running time $\Theta(m)$~\cite{berry2014why}.
In the case of a bounded degree graph, enumerating over all nodes and checking
all pairs of neighbors takes time $\Theta(nd_{\max}^2)$, where $d_{\max}$ is the
maximum degree in the graph.  We note that with triangular motifs, the number of
non-zeros in $W_M$ is less than the number of non-zeros in the original
adjacency matrix.  Thus, we do not have to worry about additional storage
requirements.

Next, we consider wedges (open triangles).  We can list all wedges by looking at
every pair of neighbors of every node.  This algorithm has $\Theta(nd_{\max}^2)$
computational complexity, where $n$ is the number of nodes and $d_{\max}$ is
again the maximum degree in the graph (a more precise bound is $\Theta(\sum_{j}
d^2_j)$, where $d_j$ is the degree of node $j$.)  If the graph is sparse, the
motif adjacency matrix will have more non-zeros than the original adjacency
matrix, so additional storage is required.  Specifically, there is fill-in for
all two-hop neighbors, so the motif adjacency matrix has $O(\sum_{j}d^2_j)$
non-zeros.  This is impractical for large real-world networks but manageable for
modestly sized networks.

Marcus and Shavitt present an algorithm for listing all four-node motifs in an
undirected graph in $O(m^2)$ time~\cite{marcus2010efficient}.  We can employ the
same edge direction check as for triangles to extend this result to directed
graphs.  Chiba and Nishizeki develop an algorithm for finding a representation
of all quadrangles (motif on four nodes that contains a four-node cycle as a
subgraph) in $O(am)$ time and $O(m)$ space, where $a$ is the arboricity of the
graph~\cite{chiba1985arboricity}.  The arboricity of any connected graph is
bounded by $O(m^{1/2})$, so this algorithm runs in time $O(m^{3/2})$.

Chiba and Nishizeki present an algorithms for $k$-clique enumeration that also
depends on the arboricity of the graph.  Specifically, they provide an algorithm
for enumerating all $k$-cliques in $O(ka^{k-2}m)$ time, where $a$ is the
arboricity of the graph.  This algorithm achieves the $\Theta(m^{3/2})$ bound
for arbitrary graphs.  (We note that the triangle listing sub-case is similar in
spirit to the algorithm proposed by Schank and Wagner~\cite{schank2005finding}).
For four-node cliques, the algorithm runs in time $O(m^2)$ time, which matches
the complexity of Marcus and Shavitt~\cite{marcus2010efficient}.

We note that we could also employ approximation algorithms to estimate the
weights in the motif adjacency matrix~\cite{becchetti2008efficient}.  Such
methods balance computation time and accuracy.  Finally, we note that the
computation of $W_M$ and the computation of the eigenvector are suitable for
parallel computation.  There are already distributed algorithms for triangle
enumeration~\cite{cohen2009graph}, and the (parallel) eigenvector computation of
a sparse matrix is a classical problem in scientific
computing~\cite{parlett1980symmetric,maschhoff1996p_arpack}.

\subsection{Experimental results on triangular motifs}

In this section, we demonstrate that our method scales to real-world networks
with billions of edges.  We tested the scalability of our method on 16 large
directed graphs from a variety of real-world applications.  These networks range
from a couple hundred thousand to two billion edges and from 10 thousand to over
50 million nodes.  Table~\ref{tab:scalability_networks} lists short descriptions
of these networks.  The wiki-RfA, email-EuAll, cit-HepPh, web-NotreDame,
amazon0601, wiki-Talk, ego-Gplus, soc-Pokec, and soc-LiveJournal1 networks were
downloaded from the SNAP collection at
\url{http://snap.stanford.edu/data/}~\cite{snapnets}.  The uk-2014-tpd,
uk-2014-host, enwiki-2013, uk-2002, arabic-2005, twitter-2010, and sk-2005
networks were downloaded from the Laboratory for Web Algorithmics collection at
\url{http://law.di.unimi.it/datasets.php}%
~\cite{boldi2004ubicrawler,boldi2004webgraph,boldi2011layered,boldi2014bubing}.
Links to all datasets are available on our project website: \projecturl.

Recall that Algorithm~\ref{alg:motif_fiedler} consists of two major computational components:
\begin{enumerate}
\item Form the weighted graph $W_M$.
\item Compute the eigenvector $z$ of second smallest eigenvalue of the matrix
  $\mathcal{L}_M$.
\end{enumerate}
After computing the eigenvector, we sort the vertices and loop over prefix sets
to find the lowest motif conductance set.  We consider these final steps as part
of the eigenvector computation for our performance experiments.

For each network in Table~\ref{tab:scalability_networks}, we ran the method for
all directed triangular motifs ($\motiftype{3}{1}$--$\motiftype{3}{7}$).
To compute $W_M$, we used a standard algorithm that meets the $O(m^{3/2})$
bound~\cite{schank2005finding,latapy2008main} with some additional pre-processing
based on the motif.  Specifically, the algorithm is:

\begin{enumerate}
\item Take motif type $M$ and graph $G$ as input.
\item (Pre-processing.)  If $M$ is $M_1$ or $M_5$, remove all bi-directional
  edges in $G$ since these motifs only contain uni-directional edges.  If $M$ is
  $M_4$, remove all uni-directional edges in $G$ as this motif only contains
  bi-directional edges.
\item Form the undirected graph $G_{\textnormal{undir}}$ by removing the
  direction of all edges in $G$.
\item Let $d_u$ be the degree of node $u$ in $G_{\textnormal{undir}}$.  Order the
  nodes in $G_{\textnormal{undir}}$ by increasing degree, breaking ties
  arbitrarily.  Denote this ordering by $\psi$.
\item For every edge undirected edge $\{u, v\}$ in $G_{\textnormal{undir}}$, if
  $\psi_u < \psi_v$, add directed edge $(u, v)$ to $G_{\textnormal{dir}}$;
  otherwise, add directed edge $(v, u)$ to $G_{\textnormal{dir}}$.
\item For every node in $u$ in $G_{\textnormal{dir}}$ and every pair of directed
  edges $(u, v)$ and $(u, w)$, check to see if $u$, $v$, and $w$ form motif $M$
  in $G$.  If they do, check if the triangle forms motif $M$ in $G$ and update
  $W_M$ accordingly.
\end{enumerate}

The algorithm runs in time $\Theta(m^{3/2})$ time in the worst case, and is also
known as an effective heuristic for real-world networks~\cite{berry2014why}.
After, we find the largest connected component of the graph corresponding to the
motif adjacency matrix $W_M$, form the motif normalized Laplacian
$\normmotiflap$ of the largest component, and compute the eigenvector of second
smallest eigenvalue of $\normmotiflap$.  To compute the eigenvector, we use
MATLAB's \texttt{eigs} routine with tolerance 1e-4 and the ``smallest
algebraic'' option for the eigenvalue type.

Table~\ref{tab:scalability_results} lists the time to compute $W_M$ and the time
to compute the eigenvector for each network.  We omitted the time to read the
graph from disk because this time strongly depends on how the graph is
compressed.  All experiments ran on a 40-core server with four 2.4 GHz Intel
Xeon E7-4870 processors.  All computations of $W_M$ were in serial and the
computations of the eigenvectors were in parallel.

Over all networks and all motifs, the longest computation of $W_M$ (including
pre-processing time) was for $\motiftype{3}{2}$ on the sk-2005 network and took
roughly 52.8 hours.  The longest eigenvector computation was for
$\motiftype{3}{6}$ on the sk-2005 network, and took about 1.62 hours. We note
that $W_M$ only needs to be computed once per network, regardless of the
eventual number of clusters that are extracted.  Also, the computation of $W_M$
can easily be accelerated by parallel computing (the enumeration of motifs can
be done in parallel over nodes, for example) or by more sophisticated
algorithms~\cite{berry2014why}.  In this work, we perform the computation of
$W_M$ in serial in order to better understand the scalability.

In theory, the triangle enumeration time is $O(m^{1.5})$.  We fit a linear
regression of the log of the computation time of the last step of the
enumeration algorithm to the regressor $\log(m)$ and a constant term:
\begin{equation}\label{eqn:regression}
\log(\text{time}) \sim a\log(m) + b
\end{equation}
If the computations truly took $cm^{1.5}$ for some constant $c$, then the
regression coefficient for $\log(m)$ would be $1.5$.  Because of the
pre-processing of the algorithm, the number of edges $m$ depends on the motif.
For example, with motifs $\motiftype{3}{1}$ and $\motiftype{3}{5}$, we only
count the number of uni-directional edges.  The pre-processing time, which is
linear in the total number of edges, is not included in the time.  The
regression coefficient for $\log(m)$ ($a$ in Equation~\ref{eqn:regression}) was
found to be smaller $1.5$ for each motif (Table~\ref{tab:scale_conf}).  The
largest regression coefficient was $1.31$ for $\motiftype{3}{3}$ (with 95\%
confidence interval $1.31 \pm 0.19$).  We also performed a regression over the
aggregate times of the motifs, and the regression coefficient was $1.17$ (with
95\% confidence interval $1.17 \pm 0.09$).  We conclude that on real-world
datasets, the algorithm for computing $W_M$ performs much better than the
worst-case guarantees.

\begin{table}[hb]
\centering
\caption{%
Summary of networks used in scalability experiments with triangular motifs.  The
total number of edges is the sum of the number of unidirectional edges and twice
the number of bidirectional edges.
}
\scalebox{0.9}{
\begin{tabular}{l l c @{\hskip 1cm} c c c}
\toprule
Name             & description                                  & \# nodes & \multicolumn{3}{l}{\# edges} \\
&                                              &          & total                        & unidir. & bidir. \\
\midrule
wiki-RfA         & Adminship voting on Wikipedia                & 10.8K & 189K  & 175K  & 7.00K \\
email-EuAll      & Emails in a research institution             & 265K  & 419K  & 310K  & 54.5K \\
cit-HepPh        & Citations for papers on arXiv HEP-PH         & 34.5K & 422K  & 420K  & 657 \\
web-NotreDame    & Hyperlinks on \texttt{nd.edu} domain         & 326K  & 1.47M & 711K  & 380K \\
amazon0601       & Product co-purchasing on Amazon              & 403K  & 3.39M & 1.50M & 944K \\
wiki-Talk        & Wikipedia users interactions                 & 2.39M & 5.02M & 4.30M & 362K \\
ego-Gplus        & Circles on Google+                           & 108K  & 13.7M & 10.8M & 1.44M \\
uk-2014-tpd      & top private domain links on \texttt{.uk} web & 1.77M & 16.9M & 13.7M & 1.58M \\
soc-Pokec        & Pokec friendships                            & 1.63M & 30.6M & 14.0M & 8.32M \\
uk-2014-host     & Host links on \texttt{.uk} web               & 4.77M & 46.8M & 33.7M & 6.55M \\
soc-LiveJournal1 & LiveJournal friendships                      & 4.85M & 68.5M & 17.2M & 25.6M \\
enwiki-2013      & Hyperlinks on English Wikipedia              & 4.21M & 101M  & 82.6M & 9.37M \\
uk-2002          & Hyperlinks on \texttt{.uk} web               & 18.5M & 292M  & 231M  & 30.5M \\
arabic-2005      & Hyperlinks on arabic-language web pages      & 22.7M & 631M  & 477M  & 77.3M \\
twitter-2010     & Twitter followers                            & 41.7M & 1.47B & 937M  & 266M \\
sk-2005          & Hyperlinks on \texttt{.sk} web               & 50.6M & 1.93B & 1.69B & 120M  \\
\bottomrule
\end{tabular}
}
\label{tab:scalability_networks}
\end{table}


\clearpage
\begin{table}[hb]
\caption{%
Time to compute the motif adjacency matrix $W_M$ and the second eigenvector of
the motif normalized Laplacian $\normmotiflap$ in seconds for each directed
triangular motif.
}
\scalebox{0.55}{
\begin{tabular}{l l l l l l l l @{\hskip 1cm} l l l l l l l}
\toprule
& \multicolumn{7}{l}{Motif adjacency matrix $W_M$} & \multicolumn{7}{l}{Second eigenvector of $\normmotiflap$} \\
Network & $\motiftype{3}{1}$ & $\motiftype{3}{2}$ & $\motiftype{3}{3}$ & $\motiftype{3}{4}$ & $\motiftype{3}{5}$ & $\motiftype{3}{6}$ & $\motiftype{3}{7}$
        & $\motiftype{3}{1}$ & $\motiftype{3}{2}$ & $\motiftype{3}{3}$ & $\motiftype{3}{4}$ & $\motiftype{3}{5}$ & $\motiftype{3}{6}$ & $\motiftype{3}{7}$ \\
\midrule
wiki-RfA         & 1.19e+00 & 2.67e+00 & 1.71e+00 & 2.06e-02 & 1.79e+00 & 2.42e+00 & 2.35e+00 & 1.14e-01 & 2.12e-01 & 1.22e-01 & 2.12e-01 & 2.12e-01 & 2.94e-01 & 2.93e-01 \\
email-EuAll      & 4.74e-01 & 8.29e-01 & 6.26e-01 & 2.46e-01 & 5.02e-01 & 5.40e-01 & 5.41e-01 & 2.29e-01 & 1.62e-01 & 2.43e-01 & 1.62e-01 & 1.62e-01 & 2.35e-01 & 1.92e-01 \\
cit-HepPh        & 7.65e+00 & 3.36e+00 & 2.73e+00 & 6.22e+00 & 8.20e+00 & 3.29e+00 & 3.35e+00 & 2.11e+00 & 2.10e+00 & 2.11e+00 & 2.10e+00 & 2.10e+00 & 2.24e+00 & 2.30e+00 \\
web-NotreDame    & 9.42e-01 & 2.39e+01 & 2.33e+01 & 2.30e+00 & 1.17e+00 & 8.29e+00 & 8.40e+00 & 1.86e-01 & 3.62e-01 & 5.97e-01 & 3.62e-01 & 3.62e-01 & 9.61e-01 & 2.06e+00 \\
amazon0601       & 2.35e+00 & 8.66e+00 & 6.91e+00 & 1.82e+00 & 2.94e+00 & 5.47e+00 & 5.73e+00 & 1.23e-01 & 6.96e-01 & 4.62e+00 & 6.96e-01 & 6.96e-01 & 4.97e+00 & 4.53e+00 \\
wiki-Talk        & 1.07e+01 & 3.00e+01 & 2.20e+01 & 3.11e+00 & 1.35e+01 & 2.09e+01 & 2.10e+01 & 1.28e+00 & 2.40e+00 & 2.51e+00 & 2.40e+00 & 2.40e+00 & 2.54e+00 & 4.52e+00 \\
ego-Gplus        & 8.55e+02 & 2.42e+03 & 1.73e+03 & 2.08e+01 & 1.63e+03 & 2.07e+03 & 2.17e+03 & 4.42e+00 & 1.68e+01 & 2.11e+01 & 1.68e+01 & 1.68e+01 & 2.57e+01 & 4.42e+01 \\
uk-2014-tpd      & 8.10e+01 & 5.31e+02 & 4.07e+02 & 2.56e+01 & 1.15e+02 & 3.04e+02 & 2.85e+02 & 3.59e+00 & 9.66e+00 & 9.92e+00 & 4.35e+00 & 9.66e+00 & 2.10e+01 & 2.16e+01 \\
soc-Pokec        & 4.17e+01 & 1.34e+02 & 1.21e+02 & 3.04e+01 & 4.88e+01 & 1.00e+02 & 1.04e+02 & 1.96e+00 & 1.75e+01 & 3.91e+01 & 1.75e+01 & 1.75e+01 & 2.39e+01 & 2.45e+01 \\
uk-2014-host     & 9.98e+02 & 4.68e+03 & 2.76e+03 & 8.90e+01 & 1.32e+03 & 2.89e+03 & 2.99e+03 & 1.81e+01 & 4.38e+01 & 6.80e+01 & 2.04e+01 & 4.38e+01 & 8.28e+01 & 8.73e+01 \\
soc-LiveJournal1 & 9.08e+01 & 7.66e+02 & 6.24e+02 & 1.24e+02 & 1.24e+02 & 4.41e+02 & 4.49e+02 & 2.32e+00 & 2.20e+01 & 1.06e+02 & 2.20e+01 & 2.20e+01 & 4.49e+01 & 6.13e+01 \\
enwiki-2013      & 8.36e+02 & 9.62e+02 & 7.09e+02 & 3.13e+01 & 9.77e+02 & 8.19e+02 & 8.38e+02 & 2.18e+01 & 7.58e+01 & 8.45e+01 & 7.58e+01 & 7.58e+01 & 2.14e+02 & 1.48e+02 \\
uk-2002          & 1.47e+03 & 8.59e+03 & 5.17e+03 & 2.45e+02 & 1.73e+03 & 4.53e+03 & 5.29e+03 & 1.66e+01 & 8.65e+01 & 2.52e+02 & 8.65e+01 & 8.65e+01 & 7.87e+02 & 5.32e+02 \\
arabic-2005      & 6.51e+03 & 7.64e+04 & 6.05e+04 & 6.08e+03 & 8.39e+03 & 3.59e+04 & 3.69e+04 & 1.98e+01 & 1.64e+02 & 4.80e+02 & 3.26e+02 & 1.64e+02 & 1.95e+03 & 1.40e+03 \\
twitter-2010     & 1.21e+04 & 1.38e+05 & 1.31e+05 & 3.33e+04 & 1.99e+04 & 8.03e+04 & 7.65e+04 & 2.23e+02 & 1.23e+03 & 1.95e+03 & 1.23e+03 & 1.23e+03 & 2.22e+03 & 2.18e+03 \\
sk-2005          & 5.52e+04 & 1.63e+05 & 1.29e+05 & 1.55e+04 & 5.23e+04 & 9.64e+04 & 8.42e+04 & 5.73e+01 & 2.94e+02 & 7.98e+02 & 2.94e+02 & 2.94e+02 & 5.83e+03 & 3.81e+03 \\
\bottomrule
\end{tabular}
}
\label{tab:scalability_results}
\end{table}

\begin{table}[h]
\centering
\caption{%
The 95\% confidence interval (CI) for the regression coefficient of the regressor
$\log(m)$ in a linear model for predicting the time to compute $W_M$, based on
the computational results for the networks in
Table~\ref{tab:scalability_networks}.  The algorithm runs is guranteed to run in
time $O(m^{3/2})$.  ``Combined'' refers to the regression coefficient when
considering all of the times.
}
\scalebox{0.8}{
\begin{tabular}{c c c c c c c c c}
  \toprule
  & Motif \\
  & $\motiftype{3}{1}$ &  $\motiftype{3}{2}$ &  $\motiftype{3}{3}$ &  $\motiftype{3}{4}$
  & $\motiftype{3}{5}$ &  $\motiftype{3}{6}$ &  $\motiftype{3}{7}$ & Combined \\
\midrule  
95\% CI & $1.20 \pm 0.19$ & $1.30 \pm 0.20$ & $1.31 \pm 0.19$ & $0.90 \pm 0.31$
& $1.20 \pm 0.20$ & $1.27 \pm 1.21$ & $1.27 \pm 0.21$ & $1.17 \pm 0.09$ \\
\end{tabular}
}
\label{tab:scale_conf}
\end{table}
\clearpage

\subsection{Experimental results on $k$-cliques}

On smaller graphs, we can compute larger motifs.  To illustrate the computation
time, we formed the motif adjacency matrix $W$ based on the $k$-cliques motif
for $k = 4, \dots, 9$.  We implemented the $k$-clique enumeration algorithm by
Chiba and Nishizeki with the additional pre-processing of computing the
($k-1$)-core of the graph.  (This pre-processing improves the running time in
practice but does not affect the asymptotic complexity.)  The motif adjacency
matrices for $k$-cliques are sparser than the adjacency matrix of the original
graph.  Thus, we do not worry about spatial complexity for these motifs.

We ran the algorithm on nine real-world networks, ranging from roughly four
thousande nodes and 88 thousand edges to over two million nodes and around five
million edges (see Table~\ref{tab:networks_cliques}.)  Each network contained at
least one $9$-clique and hence at least one $k$-clique for $k < 9$.  All
networks were downloaded from the SNAP collection at
\url{http://snap.stanford.edu/data/}~\cite{snapnets}.  All computations ran on
the same server as for the triangular motifs and again there was no parallelism.
We terminated computations after two hours.  For five of the nine networks, the
time to compute $W_M$ for the $k$-clique motif was under two hours for $k = 4,
\ldots, 9$ (Table~\ref{tab:cliques_perf}).  And for each network, the
computation finished within two hours for $k = 4, 5, 6$.  The smallest network
(in terms of number of nodes and number of edges) was the Facebook ego network,
where it took just under two hours to comptue $W_M$ for the $6$-clique motif and
over two hours for the $7$-clique motif.  This network has around 80,000 edges.
On the other hand, for the YouTube network, which contains nearly 3 million
edges, we could compute $W_M$ for the $9$-clique motif in under a minute.

We conclude that it is possible to use our frameworks with motifs much larger
than the three-node motifs on which we performed many of our experiments.
However, the number of edges is not that correlated with the running time to
compute $W_M$.  This makes sense becuse the Chiba and Nishizeki algorithm
complexity is $O(a^{k-2}m)$, where $a$ is the arboricity of the graph.  Hence,
the dependence on the number of edges is always linear.

\begin{table}[htb]
\centering
\caption{%
Summary of networks used in scalability experiments with $k$-clique motifs.  For
each graph, we consider all edges as undirected.
}
\begin{tabular}{l l c c}
\toprule
Network          & description                      & \# nodes & \# edges \\
\midrule
ego-Facebook     & Facebook friendships             & 4.04K    & 88.2K \\
wiki-RfA         & Adminship voting on Wikipedia    & 10.8K    & 182K  \\
ca-AstroPh       & author co-authorship             & 18.8K    & 198K  \\
email-EuAll      & Emails in a research institution & 265K     & 364K  \\
cit-HepPh        & paper citations                  & 34.5K    & 421K  \\
soc-Slashdot0811 & Slashdot user interactions       & 77.4K    & 469K  \\
com-DBLP         & author co-authorship             & 317K     & 1.05M \\
com-Youtube      & User friendships                 & 1.13M    & 2.99M \\
wiki-Talk        & Wikipedia users interactions     & 2.39M    & 4.66M \\
\bottomrule
\end{tabular}
\label{tab:networks_cliques}
\end{table}


\begin{table}[h]
\centering
\caption{%
Time to compute $W_M$ for $k$-clique motifs (seconds).  Only computations that
finished within two hours are listed.
}
\begin{tabular}{l c c c c c c}
\toprule
& \multicolumn{6}{c}{Number of nodes in clique ($k$)} \\
Network          & 4 & 5 & 6 & 7 & 8 & 9 \\
\midrule
ego-Facebook     & 14 & 317 & 6816 & -- & -- & -- \\
wiki-RfA         & 6 & 22 & 63 & 134 & 218 & 286 \\
ca-AstroPh       & 5 & 35 & 285 & 2164 & -- & -- \\
email-EuAll      & 1 & 2 & 4 & 5 & 6 & 6 \\
cit-HepPh        & 3 & 6 & 11 & 18 & 30 & 36 \\
soc-Slashdot0811 & 3 & 12 & 55 & 282 & 1018 & 2836 \\
com-DBLP         & 9 & 129 & 3234 & -- & -- & -- \\
com-Youtube      & 12 & 17 & 25 & 33 & 35 & 33 \\
wiki-Talk        & 64 & 466 & 2898 & -- & -- & -- \\
\bottomrule
\end{tabular}
\label{tab:cliques_perf}
\end{table}
\clearpage

%% file: supplementary-tex/SM-matrix-interp.tex

For several motifs, the motif adjacency matrix $W_M$
(Equation~\ref{eqn:informal_weighting}) has a simple formula in terms of the
adjacency matrix of the original, directed, unweighted graph, $G$.  Let $A$ be
the adjacency matrix for $G$ and let $U$ and $B$ be the adjacency matrix of the
unidirectional and bidirectional links of $G$.  Formally, $B = A \circ A^T$ and
$U = A - B$, where $\circ$ denotes the Hadamard (entry-wise) product.
Table~\ref{tab:matrix_interp} lists the formula of $W_M$ for motifs
$\motiftype{3}{1}$, $\motiftype{3}{2}$, $\motiftype{3}{3}$, $\motiftype{3}{4}$,
$\motiftype{3}{5}$, $\motiftype{3}{6}$, and $\motiftype{3}{7}$ (see
Figure~\ref{fig:motif_defs}) in terms of the matrices $U$ and $B$.  The central
computational kernel in these computations is $(X \cdot Y) \circ Z$.  When $X$,
$Y$, and $Z$ are sparse, efficient parallel algorithms have been developed and
analyzed~\cite{azad2015parallel}.  If the adjacency matrix is sparse, then
computing $W_M$ for these motifs falls into this framework.

\begin{table}[hb]
\centering
\caption{%
Matrix-based formulations of the weighted motif adjacency matrix $W_M$
(Equation~\ref{eqn:informal_weighting}) for all triangular three-node simple
motifs.  $P \circ Q$ denotes the Hadamard (entry-wise) products of matrices $P$
and $Q$.  If $A$ is the adjacency matrix of a directed, unweighted graph $G$,
then $B = A \circ A^T$ and $U = A - B$.  Note that in all cases, $W_M$ is
symmetric.
}
\begin{tabular}{c @{\hskip 1cm} l @{\hskip 1.25cm} l}
\toprule
Motif & Matrix computations & $W_M =$ \\ \midrule
$\motiftype{3}{1}$  & $C = (U \cdot U) \circ U^T$ & $C + C^T$ \\
$\motiftype{3}{2}$  & $C = (B \cdot U) \circ U^T + (U \cdot B) \circ U^T + (U \cdot U) \circ B$ & $C + C^T$ \\
$\motiftype{3}{3}$  & $C = (B \cdot B) \circ U + (B \cdot U) \circ B + (U \cdot B) \circ B$ & $C + C^T$ \\
$\motiftype{3}{4}$  & $C = (B \cdot B) \circ B$ & $C$ \\
$\motiftype{3}{5}$  & $C = (U \cdot U) \circ U + (U \cdot U^T) \circ U + (U^T \cdot U) \circ U$ & $C + C^T$ \\
$\motiftype{3}{6}$  & $C = (U \cdot B) \circ U + (B \cdot U^T) \circ U^T + (U^T \cdot U) \circ B$ & $C$ \\
$\motiftype{3}{7}$  & $C = (U^T \cdot B) \circ U^T + (B \cdot U) \circ U + (U \cdot U^T) \circ B$ & $C$ \\
%
\bottomrule
\end{tabular}
\label{tab:matrix_interp}
\end{table}

With these matrix formulations, implementing the motif-based spectral
partitioning algorithm for modestly sized graphs is straightforward.  However,
these computations become slower than standard fast triangle enumeration
algorithms when the networks are large and sparse.  Nevertheless, the matrix
formulations provide a simple and elegant computational method for the motif
adjacency matrix $W_M$.  To demonstrate, Figure~\ref{fig:sample_code_M6}
provides a complete MATLAB implementation of Algorithm~\ref{alg:motif_fiedler}
for $\motiftype{3}{6}$ (Figure~\ref{fig:motif_defs}).  The entire algorithm
including comments comrpises 28 lines of code.

\lstset{
  language=matlab,
  basicstyle=\footnotesize,
  numbers=left,
  numberstyle=\tiny\color{gray},
  keywordstyle=\color{blue},
  commentstyle=\color{purple},
  stringstyle=\color{red},
}
\begin{figure}[tb]
\centering
\lstset{language=matlab, basicstyle=\footnotesize}
\begin{lstlisting}
function [S, Sbar, conductances] = MotifSpectralPartitionM6(A)
% Spectral partitioning for motif M_6

B = spones(A & A'); % bidirectional links
U = A - B;          % unidirectional links

% Form motif adjacency matrix for motif M_6.
% For different motifs, replace this line with another matrix formulation.
W = (B * U') .* U' + (U * B) .* U + (U' * U) .* B;

% Compute eigenvector of motif normalized Laplacian
Dsqrt = full(sum(W, 2));
Dsqrt(Dsqrt ~= 0) = 1 ./ sqrt(Dsqrt(Dsqrt ~= 0));
[I, J, V] = find(W);
Ln = sparse(I, J, -V .* (Dsqrt(I) .* Dsqrt(J)), size(A, 1), size(A, 2));
[Z, lambdas] = eigs(Ln, 2, 'sa');
% Matlab's eigs is sometimes out of order
[~, eig_order] = sort(diag(lambdas));
y = Dsqrt .* Z(:, eig_order(end));

% Linear time sweep procedure
[~, order] = sort(y);
C = W(order, order);
C_sums = full(sum(C, 2));
volumes = cumsum(C_sums);
volumes_other = full(sum(sum(W))) * ones(length(order), 1) - volumes;
conductances = cumsum(C_sums - 2 * sum(tril(C), 2)) ./ min(volumes, volumes_other);
[~, split] = min(conductances);
S = order(1:split);
Sbar = order((split+1):end);
\end{lstlisting}
\caption{%
MATLAB implementation of the motif-based spectral partitioning algorithm for
motif $\motiftype{3}{6}$.  For other motifs, line 9 can be replaced with the
formulations from Table~\ref{tab:matrix_interp}.
}
\label{fig:sample_code_M6}
\end{figure}

An alternative matrix formulation comes from a motif-node adjacency matrix.  Let
$M(B, \anchorset)$ be a motif set and number the instances of the motif
$1, \ldots, \lvert M \rvert$,
so that $(\vect{v}_i, \chi_{\mathcal{A}}(\vect{v}_i))$ is the $i$th motif.
Define the $\lvert M \rvert \times n$ motif-node adjacency matrix $A_M$ by
$(A_M)_{ij} = \indicator{j \in \chi_{\mathcal{A}}(\vect{v}_i)}$.  Then
\begin{equation}\label{eqn:motif_node_adj}
(W_M)_{ij} = (A_M^TA_M)_{ij}, \quad i \neq j.
\end{equation}
This provides a convenient algebraic formulation for defining the weighted
motif adjacency matrix.  However, in practice, we do not use this formulation
for any computations.

%% file: supplementary-tex/SM-alternatives.tex
For our experiments, we compare our spectral motif-based custering to the
following methods:
\begin{itemize}
\item Standard, edge-based spectral clustering, which is a special
case of motif-based clustering.  In particular, the motifs
\begin{equation}
B_1 = \begin{bmatrix} 0 & 1 \\ 1 & 0 \end{bmatrix},\;
B_2 = \begin{bmatrix} 0 & 1 \\ 0 & 0 \end{bmatrix},\;
\anchorset = \{1, 2\}
\end{equation}
correspond to removing directionality from a directed graph.  We refer to the
union of these two motifs as $\medge$.
\item %
Infomap, which is based on the map equation~\cite{rosvall2008maps}.  Software
for Infomap was downloaded from
\url{http://mapequation.org/code.html}.
We run the algorithm the algorithm for directed links when the network under
consideration is directed.
\item %
The Louvain method~\cite{blondel2008fast}.  Software for the Louvain method was
downloaded from
\url{https://perso.uclouvain.be/vincent.blondel/research/louvain.html}
We use the ``oriented'' version of the Louvain method for directed graphs.
\end{itemize}

Infomap and the Louvain method are purely clustering methods in the sense that
they take as input the graph and produce as output a set of labels for the nodes
in the graph.  In contrast to the spectral methods, we do not have control over
the number of clusters.  Also, only the spectral methods provide embeddings of
the nodes into euclidean space, which is useful for visualization.  Thus, for
our analysis of the transportation reachiability network in
Section~\ref{sec:airports}, we only compare spectral methods.

%% file: supplementary-tex/SM-celegans.tex

We now provide more details on the cluster found for the \emph{C. elegans}
network of frontal neurons~\cite{kaiser2006nonoptimal}.  In this network, the
nodes are neurons and the edges are synapses.  The network data was downloaded
from \url{http://www.biological-networks.org/pubs/suppl/celegans131.zip}.

\subsection{Connected components of the motif adjacency matrices}

We again first onsider the connected components of the motif adjacency matrices as a
pre-processing step.  For our analysis, we consider use $\mbifan$,
$\motiftype{3}{8}$, and $\medge$ (Figure~\ref{fig:motif_defs}).  The original
network has 131 nodes and 764 edges.  The largest connected component of the
motif adjacency matrix for motif $\mbifan$ contains 112 nodes.  The remaining 19
nodes are isolated and correspond to the neurons AFDL, AIAR, AINR, ASGL/R,
ASIL/R, ASJL/R, ASKL/R, AVL, AWAL, AWCR, RID, RMFL, SIADR, and SIBDL/R.  The
largest connected component of the motif adjacency matrix for motif
$\motiftype{3}{8}$ contains 127 nodes.  The remaining 4 nodes are isolated and
correspond to the neurons ASJL/R and SIBDL/R.  The original network is weakly
connected, so the motif adjacency matrix for $\medge$ is connected.

\subsection{Comparison of bi-fan motif cluster to clusters found by existing methods}

\begin{figure}[htb]
\centering
\includegraphics[width=10cm]{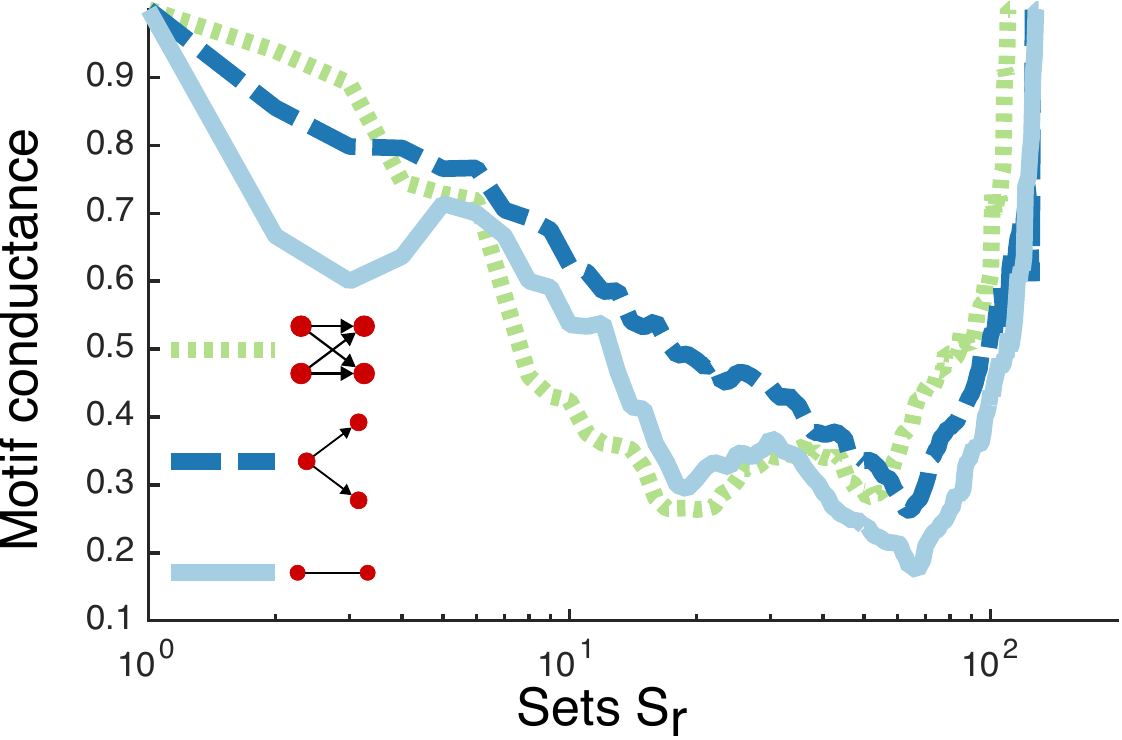}
\caption{%
Sweep profile plot ($\phi_{M}(S)$ as a function of $S$ from the sweep in
Algorithm~\ref{alg:motif_fiedler}) for $\mbifan$ (green) $\motiftype{3}{8}$
(dark blue), and $\medge$ (light blue).
}
\label{fig:celegans_ncp}
\end{figure}

We found the motif-based clusters for motifs $\mbifan$, $\motiftype{3}{8}$, and
$\medge$ by running Algorithm~\ref{alg:motif_fiedler} on the largest connected
component of the motif adjacency matrix.  Sweep profile plots ($\phi_{M}(S)$ as a
function of $S$ from the sweep in Algorithm~\ref{alg:motif_fiedler}) are shown
in Figure~\ref{fig:celegans_ncp} and show that the size of the $\mbifan$
returned by Algorithm~\ref{alg:motif_fiedler} cluster is smaller than the
clusters for $\motiftype{3}{8}$ and $\medge$.  In fact, the motif-based clusters
for $\motiftype{3}{8}$ and $\medge$ essentially bisect the graph, containing 63
of 127 and 64 of 131 nodes, respectively.  Of the 63 nodes in the
$\motiftype{3}{8}$-based cluster, only 2 are in the edge-based cluster, so
these partitions give roughly the same information.

Next, we compare the clusters found by existing methods to the $\mbifan$-based
cluster found by Algorithm~\ref{alg:motif_fiedler}.  We will show that existing
methods do not find the same group of nodes.  Let $S_{\textnormal{bifan}}$ be
the $\mbifan$-based cluster, which consists of 20 nodes.  The nodes
correspond to the following neurons: IL1DL/VL, IL2DL/DR/VL/VR/L/R, OLQDL/R, RIH,
RIPL/R, RMEL/R/V, and URADL/DR/VL/VR.  The partitions based on
$\motiftype{3}{8}$ and $\medge$ provide two sets of nodes each.  For the
subsequent analysis, we consider the set with the largest number of overlapping
nodes with $S_{\textnormal{bifan}}$.  Call these sets $S_{\motiftype{3}{8}}$ and
$S_{\textnormal{edge}}$.  We also consider the cluster found by Infomap and the
Louvain method with the largest overlap with $S_{\textnormal{bifan}}$.  Call
these sets $S_{I}$ and $S_{L}$.

To compare the most similar clusters found by other methods to
$S_{\textnormal{bifan}}$, we look at two metrics.  First, how many neurons in
$S_{\textnormal{bifan}}$ are in a cluster found by existing methods (in other
words, the overlap).  A cluster consisting of all nodes in the graph would
trivially have 100\% overlap with $S_{\textnormal{bifan}}$ but loses all
precision in the cluster identification.  Thus, we also consider the sizes of
the clusters.  These metrics are summarized as follows:
\begin{eqnarray*}
&\lvert S_{\textnormal{bifan}} \cap S_{\motiftype{3}{8}} \rvert = 20,
&\lvert S_{\motiftype{3}{8}} \rvert = 68 \\
&\lvert S_{\textnormal{bifan}} \cap S_{\textnormal{edge}} \rvert = 20,
&\lvert S_{\textnormal{edge}} \rvert = 64 \\
&\lvert S_{\textnormal{bifan}} \cap S_{L} \rvert = 13,
&\lvert S_{L} \rvert = 27 \\
&\lvert S_{\textnormal{bifan}} \cap S_{I} \rvert = 19,
&\lvert S_{I} \rvert = 114
\end{eqnarray*}

We see that $S_{\textnormal{bifan}}$ is a subset of $S_{\motiftype{3}{8}}$ and
$S_{\textnormal{edge}}$ and has substantial overlap with $S_{I}$.  However,
$S_{\textnormal{bifan}}$ is by far the smallest of all of these sets.  We
conclude that existing methods do not capture the same information as motif
$\mbifan$.

To further investigate the structure found by existing methods, we show the
clusters $S_{\textnormal{edge}}$ and $S_{\motiftype{3}{8}}$ in
Figure~\ref{fig:celegans_embed}.  From the figure, we see that spectral
clustering based on edges or motif $\motiftype{3}{8}$ simply finds a spatially
coherent cluster, rather than the control structure formed by the nodes in
$S_{\textnormal{bifan}}$.

\clearpage
\begin{figure}[ht]
\centering
\includegraphics[width=0.75\columnwidth]{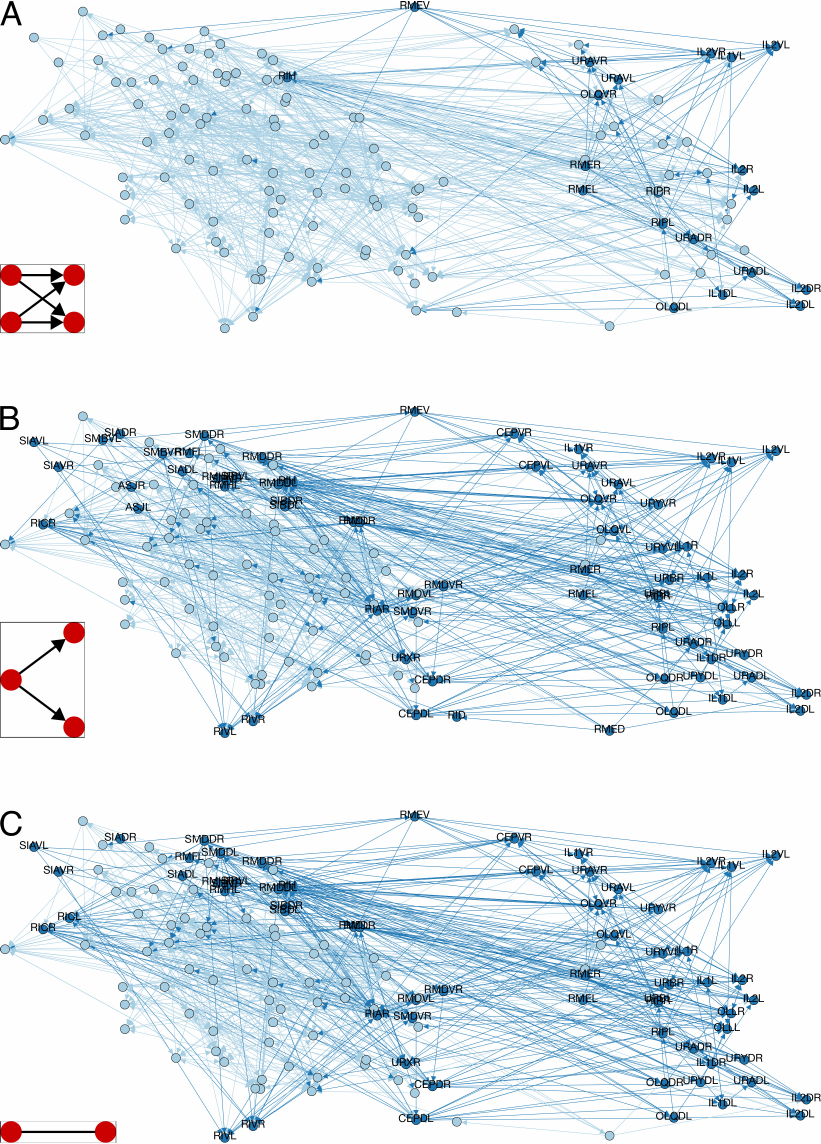}
\caption{%
Illustration of motif-based clusters with true two-dimensional spatial
dimensions of the frontal neurons of \emph{C. elegans}.
{\bf A:}
The $\mbifan$-based cluster consists of the labeled dark blue
nodes.
{\bf B:}
Partitioning the graph based on motif $\motiftype{3}{8}$, where the labeled dark
blue nodes are the nodes on the side of the partition with largest overlap of
the nodes in A.
{\bf C:}
Partitioning the graph based on edges, where the labeled dark blue nodes are the
nodes on the side of the partition with largest overlap of the nodes in A.
Note that the partitions in Figures B and C capture the cluster in Figure A, but
also contain many other nodes.  Essentially, the partitions in B and C are just
capturing spatial information.
}
\label{fig:celegans_embed}
\end{figure}
\clearpage

%% file: supplementary-tex/SM-airports.tex

The nodes in the transportation reachability network are airports in the United
States and Canada.  There is an edge from city $i$ to city $j$ if the estimated
travel time from $i$ to $j$ is less than some
threshold~\cite{frey2007clustering}.  The network is not symmetric.  The network
with estimated travel times was downloaded from \\
\url{http://www.psi.toronto.edu/affinitypropagation/TravelRouting.mat} 
and \url{http://www.psi.toronto.edu/affinitypropagation/TravelRoutingCityNames.txt}.
We collected the latitude, longitude, and metropolitan populations of the
cities using WolframAlpha and Wikipedia.  All of the data is available on
our project web page: \projecturl.

\subsection{Methods for spectral embeddings}

We compared the motif-based spectral embedding of the transportation
reachability network to spectral embeddings from other connectivity matrices.
For this analysis, we ignore the travel times times and only consider the
topology of the network.  The two-dimensional spectral embedding for a graph
defined by a (weighted) adjacency matrix $W \in \mathbb{R}^{n \times n}$ comes
from Algorithm~\ref{alg:motif_ngetal}:
\begin{enumerate}
\item %
Form the normalized Laplacian $\mathcal{L} = I - D^{-1/2}WD^{-1/2}$, where $D$
is the diagonal degree matrix with $D_{ii} = \sum_{j}W_{ij}$.
\item %
Compute the first 3 eigenvectors $z_1$, $z_2$, $z_3$ of smallest eigenvalues for
$\mathcal{L}$ ($z_1$ has the smallest eigenvalue).
\item %
Form the normalized matrix $Y \in \mathbb{R}^{n \times 3}$ by
$Y_{ij} = z_{ij} / \sqrt{\sum_{j=1}^{3}z_{ij}^2}$.
\item %
Define the primary and secondary spectral coordinates of node $i$ to be $Y_{i2}$
and $Y_{i3}$, respectively.
\end{enumerate}

We consider the following three matrices $W$.

\begin{enumerate}
\item {\bf Motif}: %
The sum of the motif adjacency matrix (Equation~\ref{eqn:formal_weighting}) for
three different anchored motifs:
\begin{equation}
B_1 = \begin{bmatrix} 0 & 1 & 1 \\ 1 & 0 & 1 \\ 1 & 1 & 0 \end{bmatrix},\;
B_2 = \begin{bmatrix} 0 & 1 & 1 \\ 1 & 0 & 1 \\ 0 & 1 & 0 \end{bmatrix},\;
B_3 = \begin{bmatrix} 0 & 1 & 0 \\ 1 & 0 & 1 \\ 0 & 1 & 0 \end{bmatrix},\;
\anchorset = \{1, 3\}.
\end{equation}
If $S$ is the matrix of bidirectional links in the graph ($S_{ij} = 1$ if and
only if $A_{ij} = A_{ji} = 1$), then the motif adjacency matrix for these motifs
is $W_M = S^2$.  The resulting embedding is shown in Figure 4C of the main text.

\item {\bf Undirected}: %
The adjacency matrix is formed by ignoring edge direction.  This is the standard
spectral embedding.    The resulting embedding is shown in Figure 4D of the main text.

\item {\bf Undirected complement}: %
The adjacency matrix is formed by taking the complement of the undirected
adjacency matrix.  This matrix tends to connect non-hubs to each other.

\end{enumerate}

The networks represented by each adjacency matrices are all connected.

\subsection{Comparison of motif-based embedding to other embeddings}

We computed 99\% confidence intervals for the Pearson correlation of the primary
spectral coordinate with the metropolitan population of the city using the
Pearson correlation coefficient.  Table~\ref{tab:correlations} lists the
confidence intervals.  (Since eigenvectors are only unique up to sign, the
confidence intervals are symmetric about $0$.  We list the interval with the
largest positive end point under this permutation to be consistent across
embeddings.)  The motif-based primary spectral coordinate has the strongest
correlation with the city populations.

We repeated the computations for the correlation between the secondary spectral
coordinate and the longitude of the city.  Again, the motif-based clustering has
the strongest correlation.  Furthermore, the lower end of the confidence
interval for the motif-based embedding was above the higher end of the
confidence interval for the other three embeddings.

\begin{table}[tb]
\centering
\caption{%
Summary of Pearson correlations for spectral embeddings of the transportation
reachability network.  We list the 99\% confidence interval for the Pearson
correlation coefficient.
}
\begin{tabular}{l @{\hskip 0.8cm} l @{\hskip 0.5cm}l}
\toprule
                      & Primary spectral coordinate & Secondary spectral coordinate \\
                      & and metropolitan population & and longitude                 \\ \midrule
Embedding             & 99\% confidence interval    & 99\% confidence interval      \\ \midrule
Motif                 & 0.43 $\pm$ 0.09             & 0.59 $\pm$ 0.08               \\
Undirected            & 0.11 $\pm$ 0.12             & 0.39 $\pm$ 0.11               \\
Undirected complement & 0.31 $\pm$ 0.11             & 0.10 $\pm$ 0.12               \\
\end{tabular}
\label{tab:correlations}
\end{table}

Finally, in order to visualize these relationships, we computed Loess
regressions of city metropolitan population and longitude against the primary
and secondary spectral coordinates for each of the embeddings
(Figure~\ref{fig:airports_loess}).  The sign of the eigenvector used in each
regression was chosen to match correlation shown in Figures 3C and 3D in the
main text (primary spectral coordinate positively correlated with population and
secondary spectral coordinate negatively correlated with longitude).  The Loess
regressions visualize the stronger correlation of the motif-based spectral
coordinates with the metropolitan popuatlion and longitude.

We conclude that the embedding provided by the motif adjacency matrix more
strongly captures the hub nature of airports and West-East geography of the
network.  To gain further insight into the relationship of the primary spectral
coordinate's relationship with the hub airports, we visualize the adjacency
matrix in Figure~\ref{fig:spectral_ordered_adj}, where the nodes are ordered by
the spectral ordering.  We see a clear relationship between the spectral ordering
and the connectivity.

\clearpage
\begin{figure}[t]
\centering
\includegraphics[width=5cm]{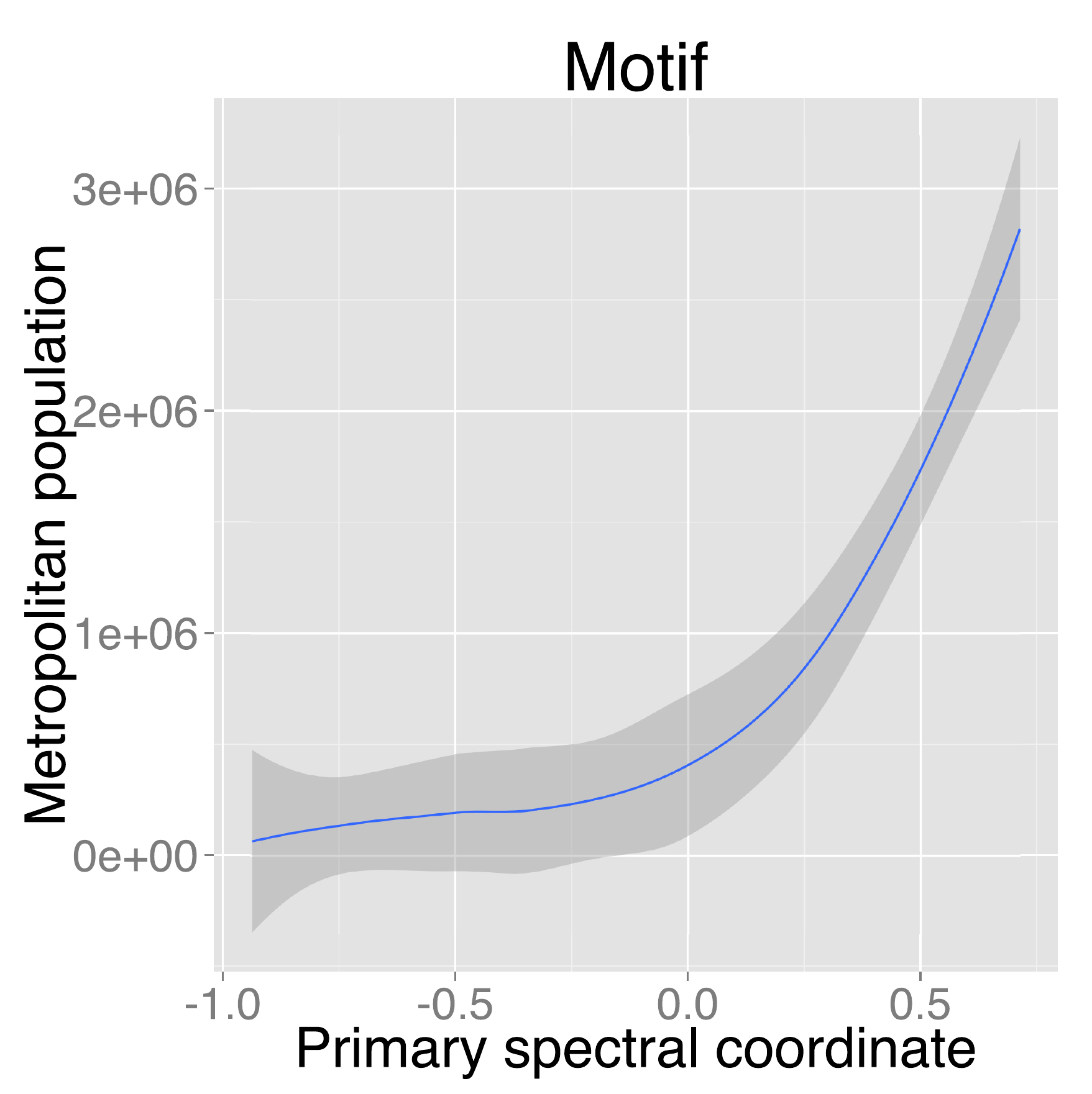}
\includegraphics[width=5cm]{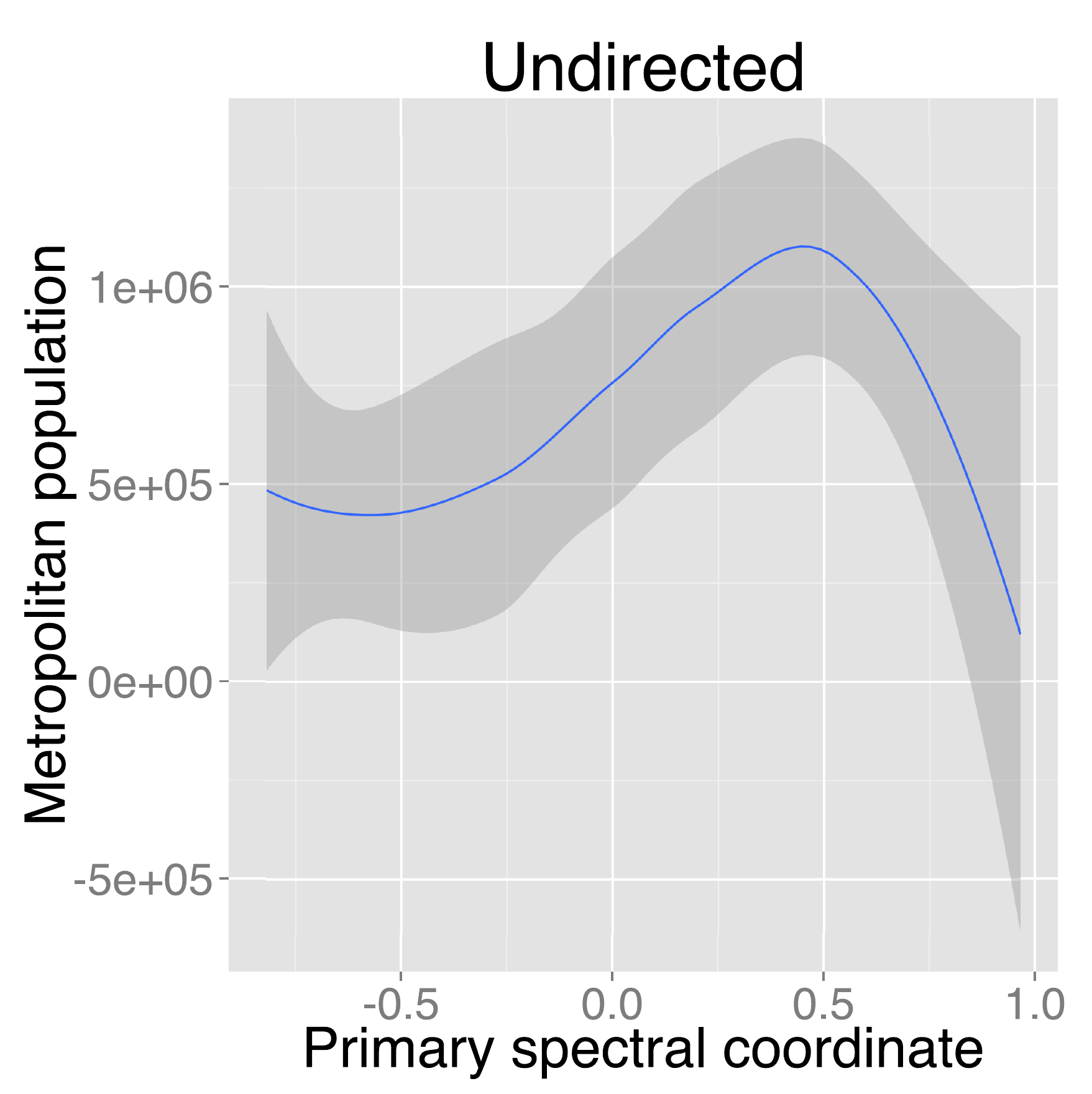}
\includegraphics[width=5cm]{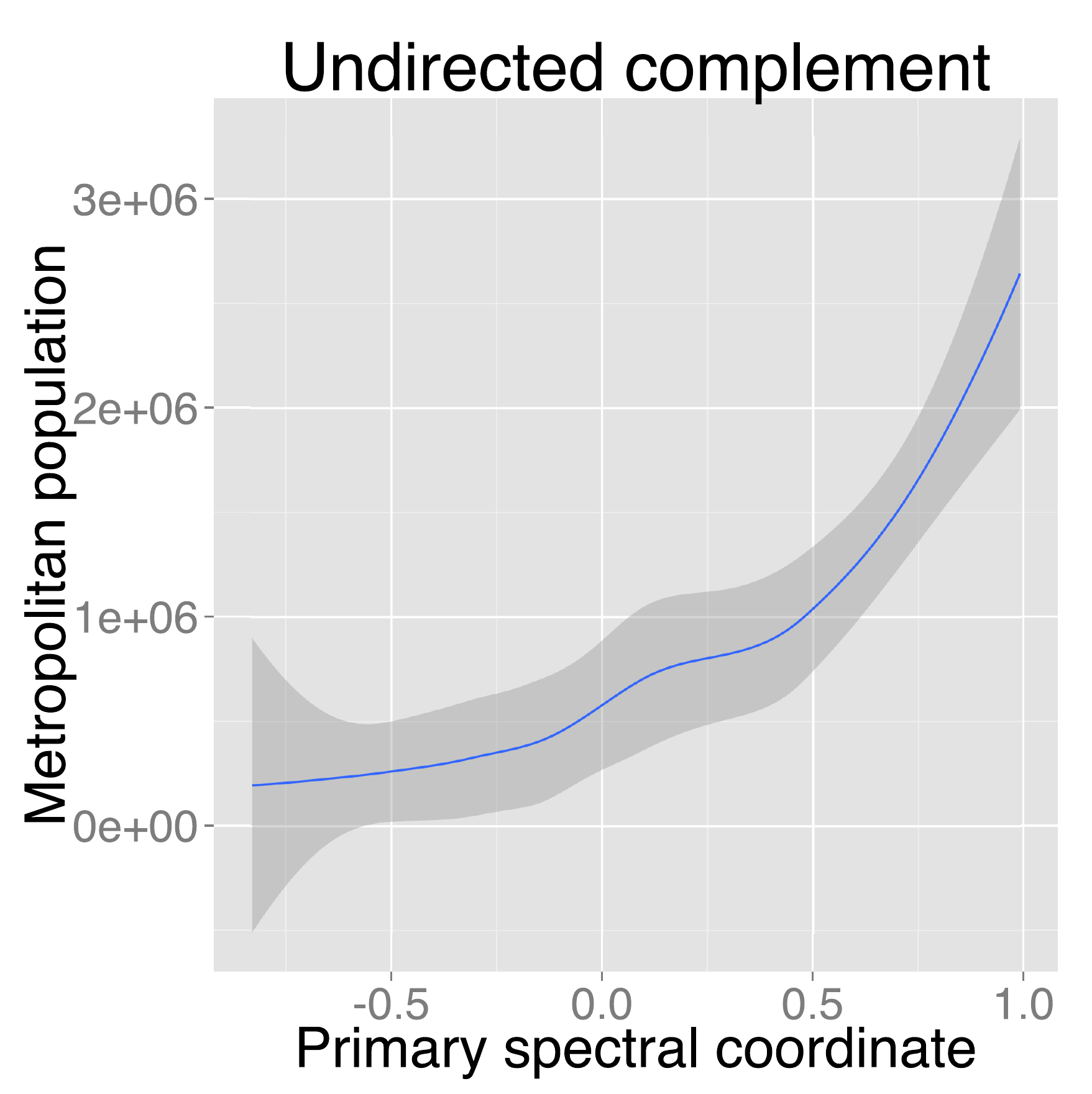}
\includegraphics[width=5cm]{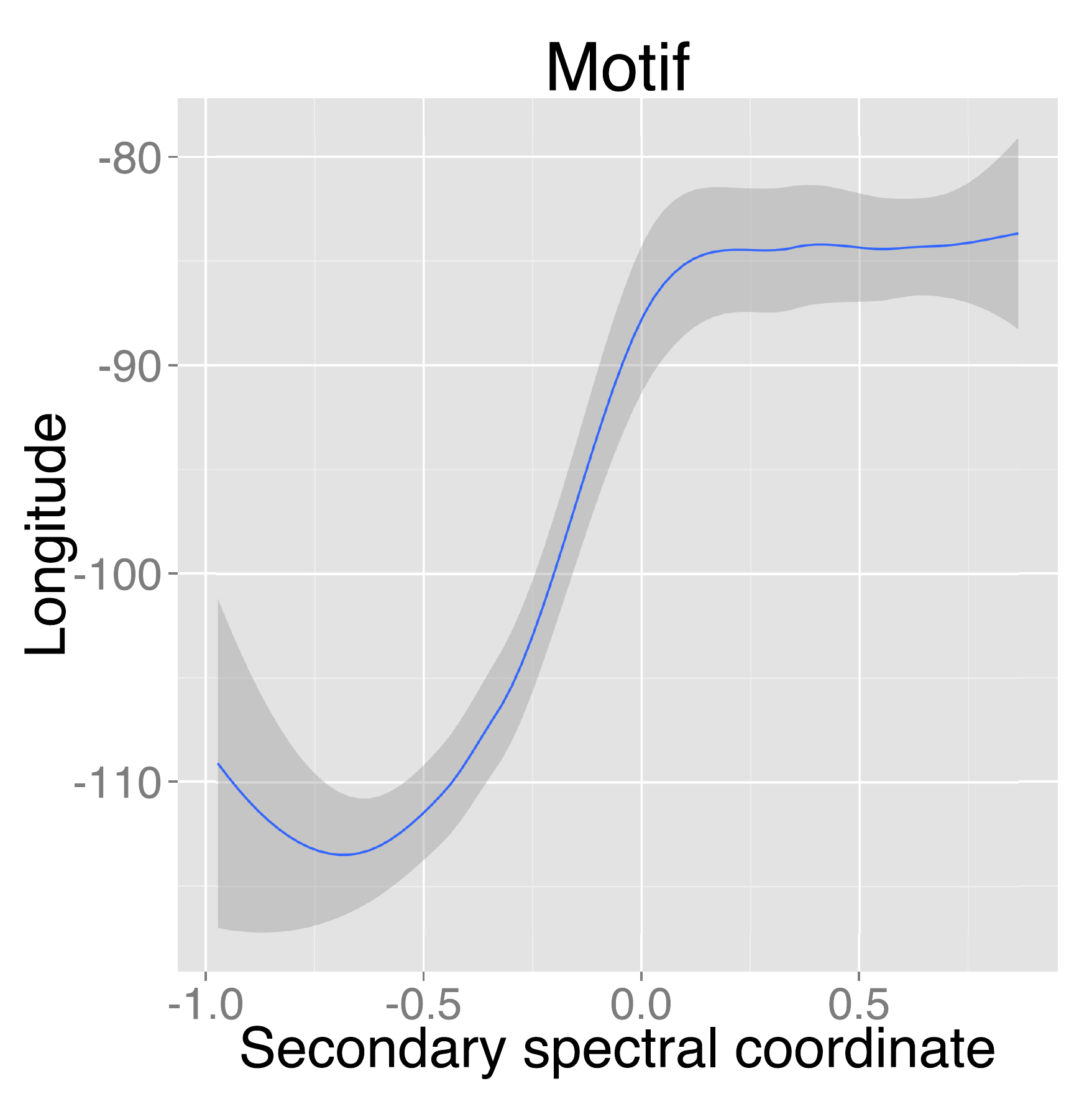}
\includegraphics[width=5cm]{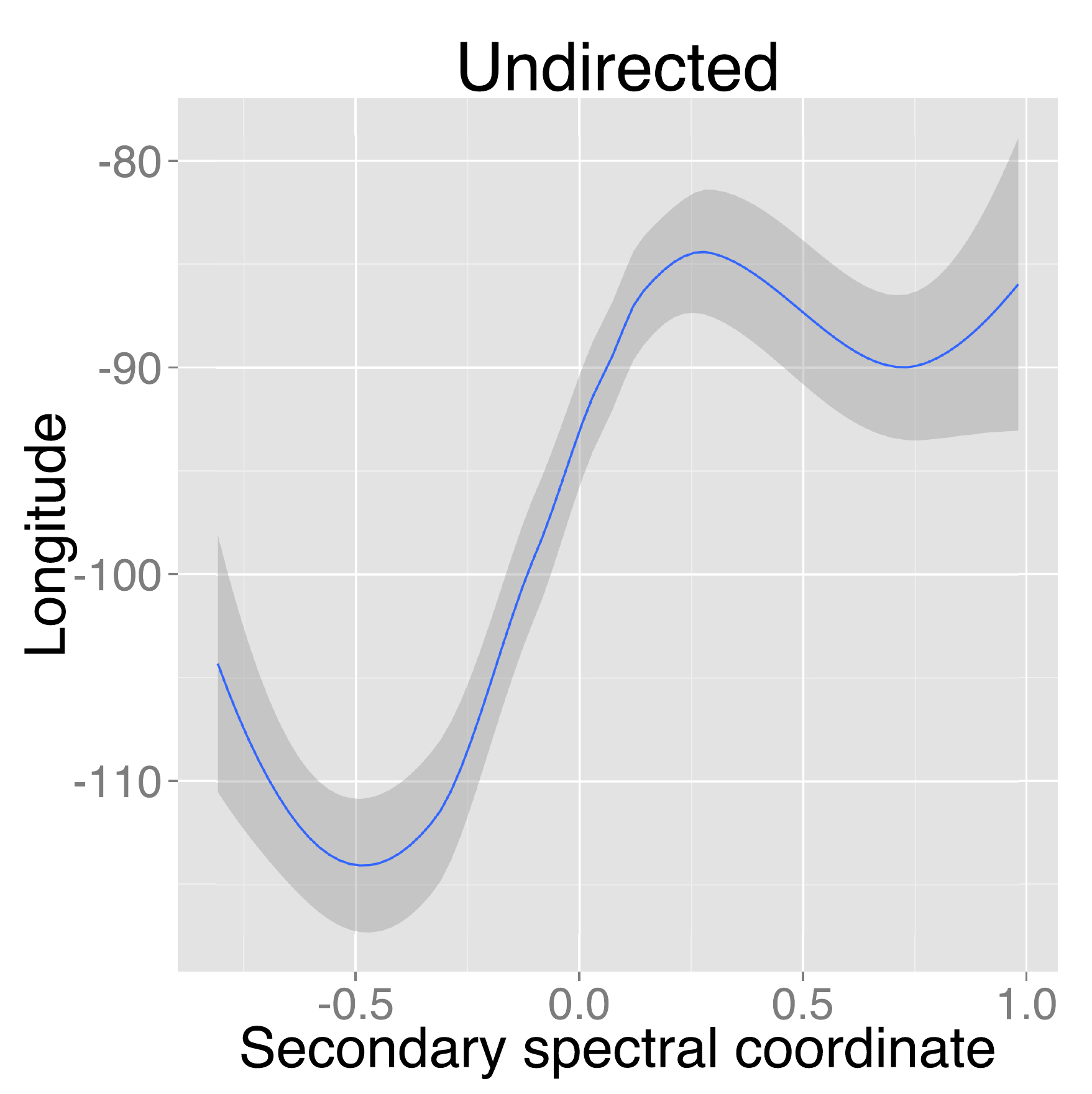}
\includegraphics[width=5cm]{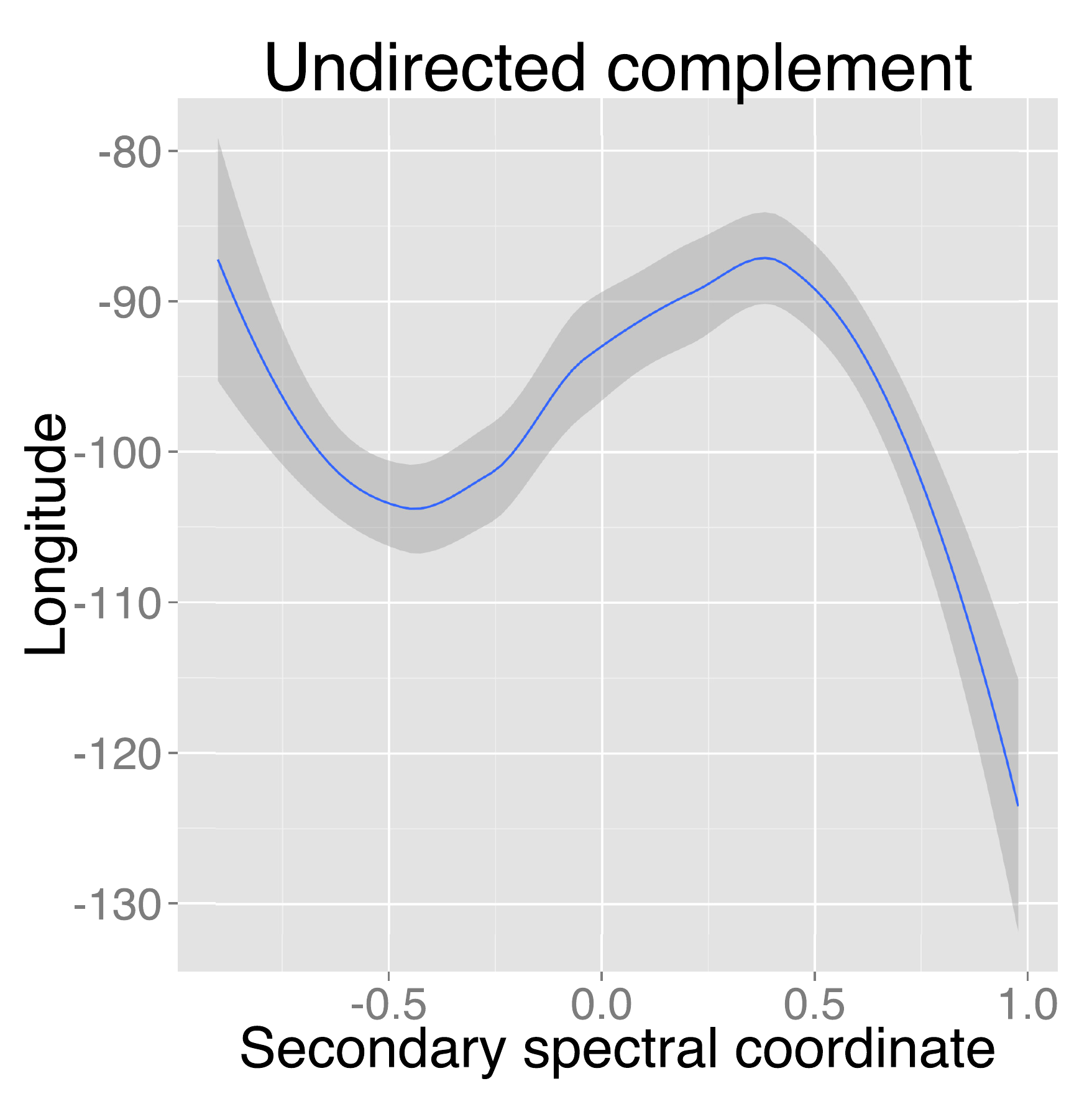}
\caption{%
Loess regressions of city metropolitan population against the primary spectral
coordinate (top) and longitude against secondary spectral coordinate (bottom)
for the motif (left), undirected (middle), and undirected complement (right)
adjacency matrices.
}
\label{fig:airports_loess}
\end{figure}

\begin{figure}[h]
\centering
\includegraphics[width=4.5cm]{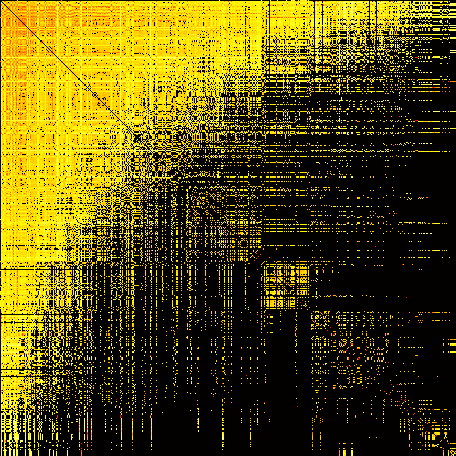}
\caption{%
Visualization of transportation reachability network.  Nodes are ordered by the
spectral ordering provided by the motif adjacency matrix.  A black dot means no
edge exists in the network.  For the edges in the network, lighter colors mean
longer estimated travel times.
}
\label{fig:spectral_ordered_adj}
\end{figure}
\clearpage

%% file: supplementary-tex/SM-case-studies.tex

We next use motif-based clustering to analyze several additional networks.  Our
main goal is to show that motif-based clusters find markedly different
structures in many real-world networks compared to edge-based clusters.  For the
case of a transcription regulation network of yeast, we also show that
motif-based clustering more accurately finds known functional modules compared
to existing methods.  On the English Wikipedia article network and the Twitter
network, we identify motifs that find anomalous clusters.  On the Stanford web
graph and in collaboration networks, we use motifs that have previously been
studied in the literature and see how they reveal organizational structure in
the networks.

\subsection{Motif $\motiftype{3}{6}$ in the Florida Bay food web}
\label{sec:foodweb}
\input{supplementary-tex/SM-foodweb}

\subsection{Coherent feedforward loops in the \emph{S.~cerevisiae}
  transcriptional regulation network}
\input{supplementary-tex/SM-yeast}

\subsection{Motif $\motiftype{3}{6}$ in the English Wikipedia article network}
\input{supplementary-tex/SM-enwiki}

\subsection{Motif $\motiftype{3}{6}$ in the Twitter follower network}
\input{supplementary-tex/SM-twitter}

\subsection{Motif $\motiftype{3}{7}$ in the Stanford web graph}
\input{supplementary-tex/SM-stanford}

\subsection{Semi-cliques in collaboration networks}
\input{supplementary-tex/SM-collaboration}

%% file: supplementary-tex/SM-foodweb.tex

We now apply the higher-order clustering framework on the Florida Bay ecosystem
food web~\cite{ulanowicz1998network}.  The dataset was downloaded from
\url{http://vlado.fmf.uni-lj.si/pub/networks/data/bio/foodweb/Florida.paj}.
In this network, the nodes are compartments (roughly, organisms and species) and
the edges represent directed carbon exchange (in many cases, this means that
species $j$ eats species $i$).  Motifs model energy flow patterns between
several species.

\subsubsection{Identifying higher-order modular organization}

In this case study, we use the framework to identify higher-order modular
organization of networks.  We focus on three motifs: $\motiftype{3}{5}$
corresponds to a hierarchical flow of energy where species $i$ and $j$ are
energy sources (prey) for species $k$, and $i$ is also an energy source for $j$;
$\motiftype{3}{6}$ models two species that prey on each other and then compete
to feed on a common third species; and $\motiftype{3}{8}$ describes a single
species serving as an energy source for two non-interacting species.  Motif
$\motiftype{3}{5}$ is considered a building block for food
webs~\cite{bascompte2005interaction,bascompte2009disentangling}, and the
prevalence of motif $\motiftype{3}{6}$ is predicted by a certain niche
model~\cite{stouffer2007evidence}.

The framework reveals that low motif conductance (high-quality) clusters only
exist for motif $\motiftype{3}{6}$ (motif conductance 0.12), whereas clusters
based on motifs $\motiftype{3}{5}$ or $\motiftype{3}{8}$ have high motif
conductance (see Figure~\ref{fig:foodweb}).  In fact, the motif Cheeger
inequality (Theorem~\ref{thm:motif_cheeger}) guarantees that clustering based on
motif $\motiftype{3}{5}$ or $\motiftype{3}{8}$ will always have larger motif
conductance that clustering based on $\motiftype{3}{6}$.  The inequality says
that the motif conductance for any cluster in a connected motif adjacency matrix
is at least half of the second smallest eigenvalue of the motif-normalized
Laplacian.  However, finding the cluster with optimal conductance is still
computationally infeasible in general~\cite{wagner1993between}.

The lower bounds using the largest connected component of the motif adjacency
matrix for motifs $\motiftype{3}{5}$, $\motiftype{3}{6}$, and $\motiftype{3}{8}$
were 0.2195, 0.0335, and 0.2191, and the clusters found by the
Algorithm~\ref{alg:motif_fiedler} had motif conductances of 0.4414, 0.1200, and
0.4145.  Thus, the cluster $S$ found by the algorithm for $\motiftype{3}{6}$ has
smaller motif $\motiftype{3}{6}$-conductance (0.12) than any possible cluster's
motif-$\motiftype{3}{5}$ or motif-$\motiftype{3}{8}$ conductance.  To state this
formally, let $C$ be the cluster found by the algorithm for motif
$\motiftype{3}{6}$ and let $H_{M}$ be the largest connected component of motif
adjacency matrix for motif $M$. Then
\begin{align}
\phi_{\motiftype{3}{6}}(H_{\motiftype{3}{6}}, C)
\le
\min\left\{
\min_{S} \phi_{\motiftype{3}{5}}(H_{\motiftype{3}{5}}, S),\;
\min_{S} \phi_{\motiftype{3}{8}}(H_{\motiftype{3}{8}}, S)
\right\}.
\end{align}
This means that, in terms of motif conductance, any cluster based on motifs
$\motiftype{3}{5}$ or $\motiftype{3}{8}$ is worse than the cluser found by the
algorithm in Theorem~\ref{thm:motif_cheeger} for motif $\motiftype{3}{6}$.
We note that the same conclusions hold for edge-based clustering.  For
motif $\medge$, the lower bound on conductance was 0.2194 and the cluster found
by the algorithm had conductance 0.4083.

\subsubsection{Analysis of higher-order modular organization}

Subsequently, we used motif $\motiftype{3}{6}$ to cluster the food web,
revealing four clusters (Figure~\ref{fig:foodweb}).  Three represent well-known
aquatic layers: (i) the pelagic system; (ii) the benthic predators of eels,
toadfish, and crabs; (iii) the sea-floor ecosystem of macroinvertebrates.  The
fourth cluster identifies microfauna supported by particulate organic carbon in
water and free bacteria. Table~\ref{tab:foodweb_classification} lists the nodes
in each cluster.

We also measured how well the motif-based clusters correlate to known ground
truth system subgroup classifications of the nodes~\cite{ulanowicz1998network}.  These classes
are microbial, zooplankton, and sediment organism microfauna; detritus; pelagic,
demersal, and benthic fishes; demseral, seagrass, and algae producers; and
macroinvertebrates (Table~\ref{tab:foodweb_classification}).%
\footnote{The classifications are also available on our project web page: \projecturl.}
We also consider a set of labels which does not include the subclassification
for microfauna and producers.  In this case, the labels are microfauna;
detritus; pelagic, demersal, and benthic fishes; producers; and
macroinvertebrates.

To quantify how well the clusters found by motif-based clustering reflect the
ground truth labels, we used several standard evaluation criteria: adjusted rand
index, F1 score, normalized mutual information, and
purity~\cite{manning2008introduction}.  We compared these results to the
clusters of several methods using the same evaluation criteria.  In total, we
evaluated six methods:
\begin{enumerate}
\item Motif-based clustering with the embedding + k-means algorithm
(Algorithm~\ref{alg:motif_ngetal}) with 500 iterations of k-means.
\item Motif-based clustering with recursive bi-partitioning (repeated application of
Algorithm~\ref{alg:motif_fiedler} on the largest remaining compoennt).  The
process continues to cut the largest cluster until there are 4 total.
\item Edge-based clustering with the embedding + k-means algorithm, again with 500
iterations of k-means.
\item Edge-based clustering with recursive bi-partitioning with the same
partitioning process.
\item The Infomap algorithm.
\item The Louvain method.
\end{enumerate}

For the first four algorithms, we control the number of clusters, which we set
to 4.  For the last two algorithms, we cannot control the number of clusters.
However, both methods found 4 clusters.

Table~\ref{tab:foodweb_performance} shows that the motif-based clustering by
embedding + k-means had the best performance for each classification criterion
on both classifications.  We conclude that the organization of compartments in
the Florida Bay foodweb are better described motif $\motiftype{3}{6}$ than by
edges.

\subsubsection{Connected components of the motif adjacency matrices}

Finally, we discuss the discuss the preprocessing step of our method, where we
compute computed connected components of the motif adjacency matrices.  The
original network has 128 nodes and 2106 edges.  The largest connected component
of the motif adjacency matrix for motif $\motiftype{3}{5}$ contains 127 of the
128 nodes.  The node corresponding to the compartment of ``roots'' is the only
node not in the largest connected component.  The two largest connected
components of the motif adjacency matrix for motif $\motiftype{3}{6}$ contain 12
and 50 nodes.  The remaining 66 nodes are isolated.
Table~\ref{tab:foodweb_m6_components} lists the nodes in each component.  We
note that the group of 12 nodes corresponds to the green cluster in
Figure~\ref{fig:foodweb}.  The motif adjacency matrix for $\motiftype{3}{8}$ is
connected.  The original network is weakly connected, so the motif adjacency
matrix for $\medge$ is also connected.

\clearpage
\begin{figure}[htb]
\centering
\includegraphics[width=0.92\textwidth]{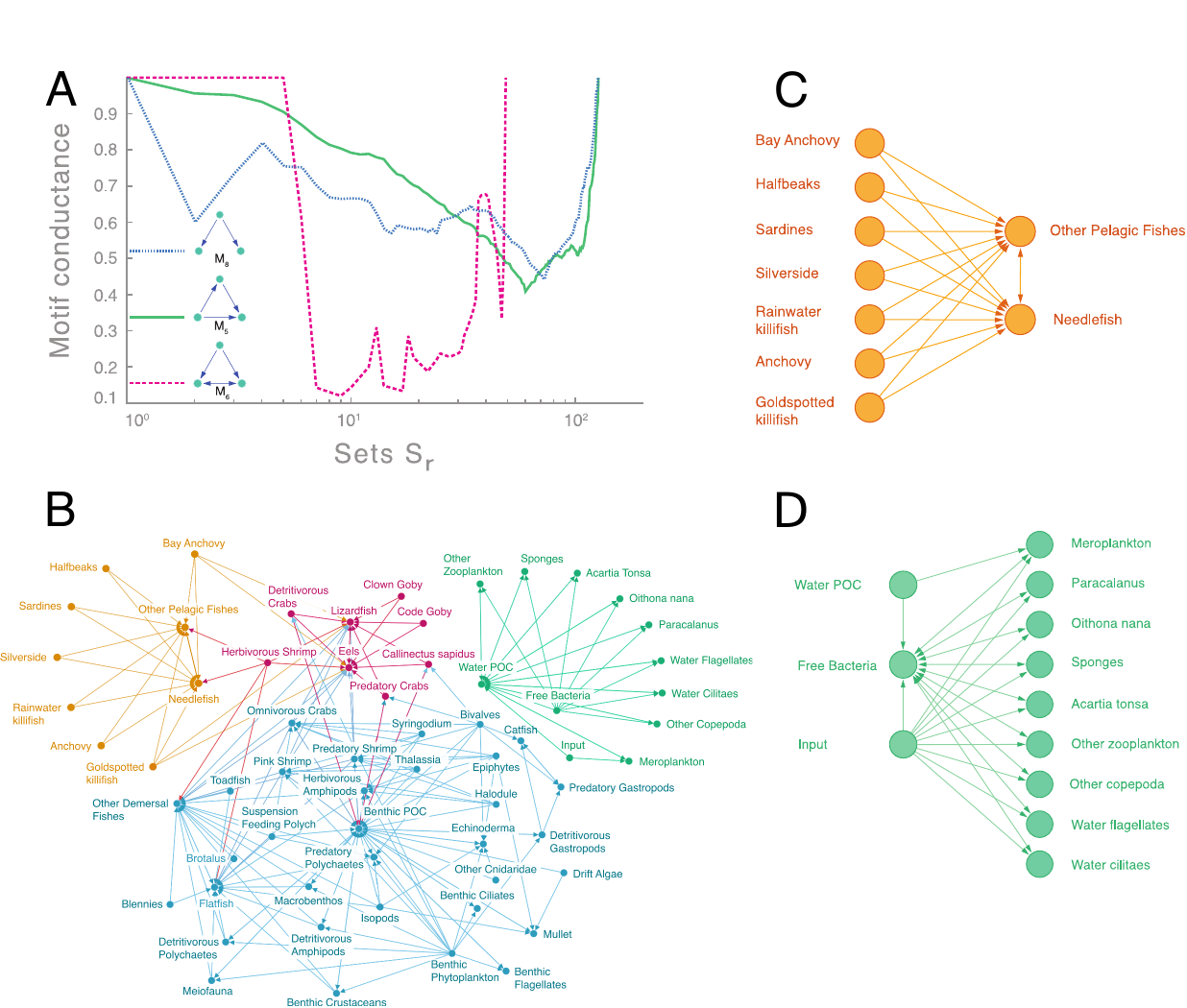}
\caption{%
{\bf Higher-order organization of the Florida Bay food web.}
{\bf A:}
Sweep profile plot ($\newmotifcond{G}{M}{S}$ as a function of $S$ from the sweep in
Algorithm~\ref{alg:motif_fiedler}) for different motifs on the Florida Bay
ecosystem food web~\cite{ulanowicz1998network}. A priori it is not clear whether
the network is organized based on a given motif.  For example, motifs
$\motiftype{3}{5}$ (green) and $\motiftype{3}{8}$ (blue) do not reveal any
higher-order organization (motif conductance has high values). However, the
downward spikes of the red curve show that $\motiftype{3}{6}$ reveals rich
higher-order modular structure~\cite{leskovec2009community}.  Ecologically,
motif $\motiftype{3}{6}$ corresponds to two species mutually feeding on each
other and also preying on a common third species.
{\bf B:}
Clustering of the food web based on motif $\motiftype{3}{6}$.  (For
illustration, edges not participating in at least one instance of the motif are
omitted.)  The clustering reveals three known aquatic layers: pelagic fishes
(yellow), benthic fishes and crabs (red), and sea-floor macroinvertebrates
(blue) as well as a cluster of microfauna and detritus (green).  Our framework
identifies these modules with higher accuracy (61\%) than existing methods
(48--53\%).
{\bf C:} 
A higher-order cluster (yellow nodes in (B)) shows how motif
$\motiftype{3}{6}$ occurs in the pelagic layer.  The needlefish and other
pelagic fishes eat each other while several other fishes are prey for these two
species.
{\bf D:}
Another higher-order cluster (green nodes in (B)) shows how motif
$\motiftype{3}{6}$ occurs between microorganisms.  Here, several microfauna
decompose into Particulate Organic Carbon in the water (water POC) but also
consume water POC.  Free bacteria serves as an energy source for both the
microfauna and water POC.
}
\label{fig:foodweb}
\end{figure}

\clearpage
\begin{table}[t]
\centering
\caption{%
Connected components of the Florida Bay foodweb motif adjacency matrix for motif
$\motiftype{3}{6}$.  There are 50 nodes in component 1, 12 nodes in component 2,
and 66 isolated nodes.
}
\scalebox{0.57}{
\begin{tabular}{l c @{\hskip 5cm} l c}
\toprule
\multicolumn{2}{l}{Two largest components} & Isolated nodes \\
Compartment (node) & Component index & Compartment (node) & \\
\midrule
\input{supplementary-tex/SM-foodweb-components}

\bottomrule
\end{tabular}
}
\label{tab:foodweb_m6_components}
\end{table}

\clearpage
\begin{table}[t]
\centering
\caption{%
Ecological classification of nodes in the Florida Bay foodweb.  Colors
correspond to the colors in the clustering of Figure~\ref{fig:foodweb}.
}
\scalebox{0.62}{
\begin{tabular}{l l l l}
\toprule
Compartment (node) & Classification 1 & Classification 2 & Assignment \\
\midrule
\input{supplementary-tex/SM-foodweb-labels}

\bottomrule
\end{tabular}
}
\label{tab:foodweb_classification}
\end{table}

\clearpage
\begin{table}[t]\def\arraystretch{1.3}
\centering
\caption{%
Comparison of motif-based algorithms against other methods in finding ground
truth structure in the Florida Bay food web~\cite{ulanowicz1998network}.
Performance for identifying the two classifications provided in
Table~\ref{tab:foodweb_classification} was evaluated based on Adjusted Rand
Index (ARI), F1 score, Normalized Mutual Information (NMI), and Purity.  In all
cases, the motif-based methods have the best performance.
}
\scalebox{0.75}{
\begin{tabular}{l @{\hskip 1.25cm} l c @{\hskip 0.5cm} c @{\hskip 0.5cm} c @{\hskip 0.5cm} c @{\hskip 0.5cm} c @{\hskip  0.5cm} c } %
\toprule
\input{supplementary-tex/SM-foodweb-eval}

\bottomrule
\end{tabular}
}
\label{tab:foodweb_performance}
\end{table}

%% file: supplementary-tex/SM-foodweb-components.tex
Benthic Phytoplankton & 1 & Barracuda & \\
Thalassia & 1 & \SI{2}{\micro\metre} Spherical Phytoplankt & \\
Halodule & 1 & Synedococcus & \\
Syringodium & 1 & Oscillatoria & \\
Drift Algae & 1 & Small Diatoms ($<$\SI{20}{\micro\metre}) & \\
Epiphytes & 1 & Big Diatoms ($>$\SI{20}{\micro\metre}) & \\
Predatory Gastropods & 1 & Dinoflagellates & \\
Detritivorous Polychaetes & 1 & Other Phytoplankton & \\
Predatory Polychaetes & 1  & Roots & \\
Suspension Feeding Polych & 1 & Coral & \\
Macrobenthos & 1  & Epiphytic Gastropods & \\
Benthic Crustaceans & 1  & Thor Floridanus & \\
Detritivorous Amphipods & 1 & Lobster & \\
Herbivorous Amphipods & 1 & Stone Crab & \\
Isopods & 1 & Sharks & \\
Herbivorous Shrimp & 1 & Rays & \\
Predatory Shrimp & 1 & Tarpon & \\
Pink Shrimp & 1 & Bonefish & \\
Benthic Flagellates & 1 & Other Killifish & \\
Benthic Ciliates & 1 & Snook & \\
Meiofauna & 1 & Sailfin Molly & \\
Other Cnidaridae & 1 & Hawksbill Turtle & \\ 
Silverside & 1 & Dolphin & \\
Echinoderma & 1 & Other Horsefish & \\
Bivalves & 1 & Gulf Pipefish & \\
Detritivorous Gastropods & 1 & Dwarf Seahorse & \\
Detritivorous Crabs & 1 & Grouper & \\
Omnivorous Crabs & 1 & Jacks & \\
Predatory Crabs & 1 & Pompano & \\
Callinectes sapidus (blue crab) & 1 & Other Snapper & \\
Mullet & 1 & Gray Snapper & \\
Blennies & 1 & Mojarra & \\
Code Goby & 1 & Grunt & \\
Clown Goby & 1 & Porgy & \\
Flatfish & 1 & Pinfish & \\
Sardines & 1 & Scianids & \\
Anchovy & 1 & Spotted Seatrout & \\
Bay Anchovy & 1 & Red Drum & \\
Lizardfish & 1 & Spadefish & \\
Catfish & 1 & Parrotfish & \\
Eels & 1 & Mackerel & \\
Toadfish & 1 & Filefishes & \\
Brotalus & 1 & Puffer & \\
Halfbeaks & 1 & Loon & \\
Needlefish & 1 & Greeb & \\
Goldspotted killifish & 1 & Pelican & \\
Rainwater killifish & 1 & Comorant & \\
Other Pelagic Fishes & 1 & Big Herons and Egrets & \\
Other Demersal Fishes & 1 & Small Herons and Egrets & \\
Benthic Particulate Organic Carbon (Benthic POC) & 1 & Ibis & \\
Free Bacteria & 2 & Roseate Spoonbill & \\
Water Flagellates & 2 & Herbivorous Ducks & \\
Water Cilitaes & 2 & Omnivorous Ducks & \\
Acartia Tonsa & 2 & Predatory Ducks & \\
Oithona nana & 2 & Raptors & \\
Paracalanus & 2 & Gruiformes & \\
Other Copepoda & 2 & Small Shorebirds & \\
Meroplankton & 2 & Gulls and Terns & \\
Other Zooplankton & 2 & Kingfisher & \\
Sponges & 2 & Crocodiles & \\
Water Particulate Organic Carbon (Water POC) & 2 & Loggerhead Turtle & \\
Input & 2 & Green Turtle & \\
& & Manatee & \\
& & Dissolved Organic Carbon (DOC) & \\
& & Output & \\
& & Respiration & \\

%% file: supplementary-tex/SM-foodweb-labels.tex
Free Bacteria             & Microbial microfauna         & Microfauna         & Green  \\
Water Flagellates         & Microbial microfauna         & Microfauna         & Green  \\
Water Cilitaes            & Microbial microfauna         & Microfauna         & Green  \\
Acartia Tonsa             & Zooplankton microfauna       & Microfauna         & Green  \\
Oithona nana              & Zooplankton microfauna       & Microfauna         & Green  \\
Paracalanus               & Zooplankton microfauna       & Microfauna         & Green  \\
Other Copepoda            & Zooplankton microfauna       & Microfauna         & Green  \\
Meroplankton              & Zooplankton microfauna       & Microfauna         & Green  \\
Other Zooplankton         & Zooplankton microfauna       & Microfauna         & Green  \\
Sponges                   & Macroinvertebrates           & Macroinvertebrates & Green  \\
Water POC                 & Detritus                     & Detritus           & Green  \\
Input                     & Detritus                     & Detritus           & Green  \\
Sardines                  & Pelagic Fishes               & Pelagic Fishes     & Yellow \\
Anchovy                   & Pelagic Fishes               & Pelagic Fishes     & Yellow \\
Bay Anchovy               & Pelagic Fishes               & Pelagic Fishes     & Yellow \\
Halfbeaks                 & Pelagic Fishes               & Pelagic Fishes     & Yellow \\
Needlefish                & Pelagic Fishes               & Pelagic Fishes     & Yellow \\
Goldspotted killifish     & Fishes Demersal              & Fishes Demersal    & Yellow \\
Rainwater killifish       & Fishes Demersal              & Fishes Demersal    & Yellow \\
Silverside                & Pelagic Fishes               & Pelagic Fishes     & Yellow \\
Other Pelagic Fishes      & Pelagic Fishes               & Pelagic Fishes     & Yellow \\
Detritivorous Crabs       & Macroinvertebrates           & Macroinvertebrates & Red    \\
Predatory Crabs           & Macroinvertebrates           & Macroinvertebrates & Red    \\
Callinectus sapidus       & Macroinvertebrates           & Macroinvertebrates & Red    \\
Lizardfish                & Benthic Fishes               & Benthic Fishes     & Red    \\
Eels                      & Fishes Demersal              & Fishes Demersal    & Red    \\
Code Goby                 & Benthic Fishes               & Benthic Fishes     & Red    \\
Clown Goby                & Benthic Fishes               & Benthic Fishes     & Red    \\
Herbivorous Shrimp        & Macroinvertebrates           & Macroinvertebrates & Red    \\
Benthic Phytoplankton     & Producer Demersal            & Producer           & Blue   \\
Thalassia                 & Producer Seagrass            & Producer           & Blue   \\
Halodule                  & Producer Seagrass            & Producer           & Blue   \\
Syringodium               & Producer Seagrass            & Producer           & Blue   \\
Drift Algae               & Producer Algae               & Producer           & Blue   \\
Epiphytes                 & Producer Algae               & Producer           & Blue   \\
Benthic Flagellates       & Sediment Organism microfauna & Microfauna         & Blue   \\
Benthic Ciliates          & Sediment Organism microfauna & Microfauna         & Blue   \\
Meiofauna                 & Sediment Organism microfauna & Microfauna         & Blue   \\
Other Cnidaridae          & Macroinvertebrates           & Macroinvertebrates & Blue   \\
Echinoderma               & Macroinvertebrates           & Macroinvertebrates & Blue   \\
Bivalves                  & Macroinvertebrates           & Macroinvertebrates & Blue   \\
Detritivorous Gastropods  & Macroinvertebrates           & Macroinvertebrates & Blue   \\
Predatory Gastropods      & Macroinvertebrates           & Macroinvertebrates & Blue   \\
Detritivorous Polychaetes & Macroinvertebrates           & Macroinvertebrates & Blue   \\
Predatory Polychaetes     & Macroinvertebrates           & Macroinvertebrates & Blue   \\
Suspension Feeding Polych & Macroinvertebrates           & Macroinvertebrates & Blue   \\
Macrobenthos              & Macroinvertebrates           & Macroinvertebrates & Blue   \\
Benthic Crustaceans       & Macroinvertebrates           & Macroinvertebrates & Blue   \\
Detritivorous Amphipods   & Macroinvertebrates           & Macroinvertebrates & Blue   \\
Herbivorous Amphipods     & Macroinvertebrates           & Macroinvertebrates & Blue   \\
Isopods                   & Macroinvertebrates           & Macroinvertebrates & Blue   \\
Predatory Shrimp          & Macroinvertebrates           & Macroinvertebrates & Blue   \\
Pink Shrimp               & Macroinvertebrates           & Macroinvertebrates & Blue   \\
Omnivorous Crabs          & Macroinvertebrates           & Macroinvertebrates & Blue   \\
Catfish                   & Benthic Fishes               & Benthic Fishes     & Blue   \\
Mullet                    & Pelagic Fishes               & Pelagic Fishes     & Blue   \\
Benthic POC               & Detritus                     & Detritus           & Blue   \\
Toadfish                  & Benthic Fishes               & Benthic Fishes     & Blue   \\
Brotalus                  & Fishes Demersal              & Fishes Demersal    & Blue   \\
Blennies                  & Benthic Fishes               & Benthic Fishes     & Blue   \\
Flatfish                  & Benthic Fishes               & Benthic Fishes     & Blue   \\
Other Demersal Fishes     & Fishes Demersal              & Fishes Demersal    & Blue   \\

%% file: supplementary-tex/SM-foodweb-eval.tex
& Evaluation & \multicolumn{1}{l}{Motif embedding} & \multicolumn{1}{l}{Motif recursive}
& \multicolumn{1}{l}{Edge embedding}  & \multicolumn{1}{l}{Edge recursive}
& \phantom{XX}Infomap\phantom{XX}         & Louvain \\
& 
& \multicolumn{1}{l}{+ k-means}       & \multicolumn{1}{l}{bi-partitioning}
& \multicolumn{1}{l}{+ k-means}       & \multicolumn{1}{l}{bi-partitioning} 
&                                     & \\ \midrule
\parbox[t]{2mm}{\multirow{4}{*}{\rotatebox[origin=c]{90}{Classification 1}}}
& ARI    & \textbf{0.3005} & 0.2156 & 0.1564 & 0.1226 & 0.1423 & 0.2207 \\
& F1     & \textbf{0.4437} & 0.3853 & 0.3180 & 0.2888 & 0.3100 & 0.4068 \\
& NMI    & \textbf{0.5040} & 0.4468 & 0.4112 & 0.3879 & 0.4035 & 0.4220 \\
& Purity & \textbf{0.5645} & 0.5323 & 0.4032 & 0.4194 & 0.4194 & 0.5323 \\ \midrule
\parbox[t]{2mm}{\multirow{4}{*}{\rotatebox[origin=c]{90}{Classification 2}}}
& ARI    & \textbf{0.3265} & 0.2356 & 0.1814 & 0.1190 & 0.1592 & 0.2207 \\
& F1     & \textbf{0.4802} & 0.4214 & 0.3550 & 0.3035 & 0.3416 & 0.4068 \\
& NMI    & \textbf{0.4822} & 0.4185 & 0.3533 & 0.3034 & 0.3471 & 0.4220 \\
& Purity & \textbf{0.6129} & 0.5806 & 0.4839 & 0.4355 & 0.4677 & 0.5323 \\

%% file: supplementary-tex/SM-yeast.tex

In this network, each node is an operon (a group of genes in a mRNA molecule),
and a directed edge from operon $i$ to operon $j$ means that $i$ is regulated by
a transcriptional factor encoded by $j$~\cite{alon2007network}.  Edges are
directed and signed.  A positive sign represents activation and a negative sign
represents repression.  The network data was downloaded from
\url{http://www.weizmann.ac.il/mcb/UriAlon/sites/mcb.UriAlon/files/uploads/NMpaper/yeastdata.mat}
and
\url{http://www.weizmann.ac.il/mcb/UriAlon/sites/mcb.UriAlon/files/uploads/DownloadableData/list_of_ffls.pdf}.

For this case study, we examine the coherent feedforward loop motif (see
Figure~\ref{fig:yeast}), which act as sign-sensitive delay elements in
transcriptional regulation
networks~\cite{mangan2003structure,mangan2003coherent}.  Formally, the
feedforward loop is represented by the following signed motifs
\begin{align}\label{eqn:cffls}
B_1 = \begin{bmatrix} 0 & + & + \\ 0 & 0 & + \\ 0 & 0 & 0 \end{bmatrix},\;
B_2 = \begin{bmatrix} 0 & - & - \\ 0 & 0 & + \\ 0 & 0 & 0 \end{bmatrix},\;
B_3 = \begin{bmatrix} 0 & + & - \\ 0 & 0 & - \\ 0 & 0 & 0 \end{bmatrix},\;
B_4 = \begin{bmatrix} 0 & - & + \\ 0 & 0 & - \\ 0 & 0 & 0 \end{bmatrix}.
\end{align}
These motifs have the same edge pattern and only differ in sign.  All of the
motifs are simple ($\anchorset = \{1, 2, 3\}$).  For our analysis, we consider
all coherent feedforward loops that are subgraphs on the induced subgraph of any
three nodes.  However, there is only one instance where the coherent feedforward
loop itself is a subgraph but not an induced subgraph on three nodes.
Specifically, the induced subgraph by DAL80, GAT1, and GLN3 contains a
bi-directional edge between DAL80 and GAT1, unidirectional edges from DAL80 and
GAT1 to GLN3.

\subsubsection{Connected components of the adjacency matrices}

\begin{table}[tb]
\centering
\caption{%
Connected components of size greater than one for the motif adjacency matrix in
the \emph{S.~cerevisiae} network for the coherent feedforward loop.
}
\scalebox{0.9}{
\begin{tabular}{c@{\hskip 0.5cm} l l}
\toprule
Size & operons  \\ \midrule
18 & ALPHA1, CLN1, CLN2, GAL11, HO, MCM1, MFALPHA1, PHO5, SIN3, \\
     & SPT16, STA1, STA2, STE3, STE6, SWI1, SWI4/SWI6, TUP1, SNF2/SWI1 \\
9 & HXT11, HXT9, IPT1, PDR1, PDR3, PDR5, SNQ2, YOR1, YRR1  \\
9 & GCN4, ILV1, ILV2, ILV5, LEU3, LEU4, MET16, MET17, MET4  \\
6 & CHO1, CHO2, INO2, INO2/INO4, OPI3, UME6 \\
6 & DAL80, DAL80/GZF3, GAP1, GAT1, GLN1, GLN3 &  \\
5 & CYC1, GAL1, GAL4, MIG1, HAP2/3/4/5 \\
3 & ADH2, CCR4, SPT6 \\
3 & CDC19, RAP1, REB1 \\
3 & DIT1, IME1, RIM101 \\
\bottomrule
\end{tabular}
}
\label{tab:yeast_ccs}
\end{table}

Again, we analyze the component structure of the motif adjacency matrix as a
pre-processing step.  The original network consists of 690 nodes and 1082 edges,
and its largest weakly connected component consists of 664 nodes and 1066 edges.
Every coherent feedforward loop in the network resides in the largest weakly
connected component, so we subsequently consider this sub-network in the
following analysis.  Of the 664 nodes in the network, only 62 participate in a
coherent feedforward loop.  Forming the motif adjacency matrix results in nine
connected components, of sizes 18, 9, 9, 6, 6, 5, 3, 3, and 3.  The operons for
the connected components consisting of more than one node is listed in
Table~\ref{tab:yeast_ccs}.

\subsubsection{Comparison against existing methods}

\begin{figure}[htb]
\centering
\includegraphics[width=1\textwidth]{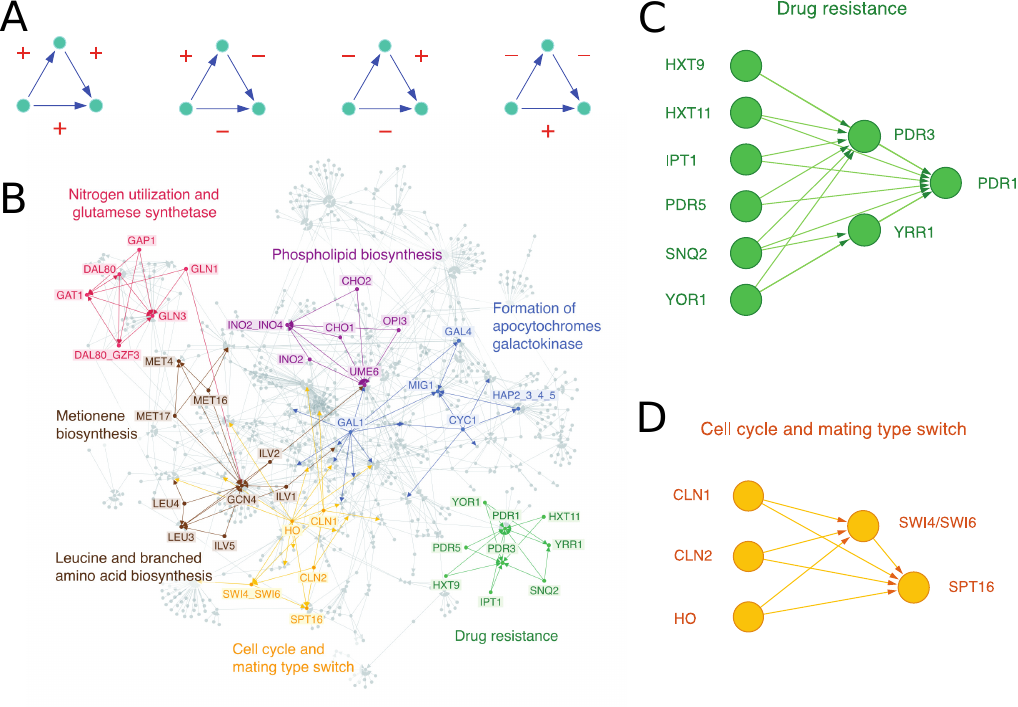}
\caption{%
Higher-order organization of the \emph{S.~cerevisiae} transcriptional regulation
network.
{\bf A:}
The four higher-order structures used by our higher-order clustering method,
which can model signed motifs. These are coherent feedfoward loop motifs, which
act as sign-sensitive delay elements in transcriptional regulation
networks~\cite{mangan2003coherent}.  The edge signs refer to activation
(positive) or repression (negative).
{\bf B:}
Six higher-order clusters revealed by the motifs in (A).  Clusters show
functional modules consisting of several motifs (coherent feedforward loops),
which were previously studied individually~\cite{mangan2003structure}.  The
higher-order clustering framework identifies the functional modules with higher
accuracy (97\%) than existing methods (68--82\%).
{\bf C--D:}
Two higher-order clusters from (B). In these clusters, all edges have positive
sign.  The functionality of the motifs in the modules correspond to drug
resistance (C) or cell cycle and mating type match (D).  The clustering suggests
that coherent feedforward loops function together as a single processing unit
rather than as independent elements.
}
\label{fig:yeast}
\end{figure}

We note that, although the original network is connected, the motif adjacency
matrix corresponds to a disconnected graph.  This already reveals much of the
structure in the network (Figure~\ref{fig:yeast}).  Indeed, this ``shattering''
of the graph into components for the feedforward loop has previously been
observed in transcriptional regulation networks~\cite{dobrin2004aggregation}.
We additionally used Algorithm~\ref{alg:motif_fiedler} to partition the largest
connected component of the motif adjacency matrix (consisting of 18 nodes).
This revealed the cluster $\{$CLN2, CLN1, SWI4/SWI6, SPT16, HO$\}$, which
contains three coherent feedforward loops (Figure~\ref{fig:yeast}).  All three
instances of the motif correspond to the function ``cell cycle and mating type
switch".  The motifs in this cluster are the only feedforward loops for which
the function is described in Reference~\cite{mangan2003structure}.  Using the
same procedure on the undirected version of the induced subgraph of the 18 nodes
(\emph{i.e.}, using motif $\medge$) results in the cluster $\{$CLN1, CLN2,
SPT16, SWI4/SWI6 $\}$.  This cluster breaks the coherent feedforward loop formed
by HO, SWI4/SWI6, and SPT16.

We also evaluated our method based on the classification of motif
functionality~\cite{mangan2003structure}.%
\footnote{The functionalities may be downloaded from our project web page: \projecturl.}
In total, there are 12 different
functionalities and 29 instances of labeled coherent feedforward loops.  We
considered the motif-based clustering of the graph to be the connected
components of the motif adjacency matrix with the additional partition of the
largest connected component.  To form an edge-based clustering, we used the
embedding + k-means algorithm on the undirected graph (\emph{i.e.}, motif
$\medge$) with $k = 12$ clusters.  We also clustered the graph using Infomap and
the Louvain method.  Table~\ref{tab:ffl_classification} summarizes the results.
We see that the motif-based clustering coherently labels all 29 motifs in the
sense that the three nodes in every instance of a labeled motif is placed in the
same cluster.  The edge-based spectral, Infomap, and Louvain clustering
coherently labeled 25, 23, and 23 motifs, respectively.

We measured the accuracy of each clustering method as the rand
index~\cite{manning2008introduction} on the coherently labeled motifs,
multiplied by the fraction of coherently labeled motifs.  The motif-based
clustering had the highest accuracy.  We conclude that motif-based clustering
provides an advantage over edge-based clustering methods in identifying
functionalities of coherent feedforward loops in the the \emph{S.~cerevisiae}
transcriptional regulation network.

\clearpage
\begin{table}[h]
\centering
\caption{%
Classification of coherent feedforward loop motifs by several clustering
methods.  In a given motif instance, we say that it is coherently labeled
if the nodes comprising the motif are in the same cluster.  If a motif is
not coherently labeled, a ``-1" is listed. 
The accuracy is the rand index on the labels and motif
functionality on coherently labeled motifs, multiplied by the fraction of
coherently labeled motifs.
}
\scalebox{0.56}{
\begin{tabular}{l l l l l c c c c}
\toprule
\input{supplementary-tex/SM-yeast-ffls}
\bottomrule
\end{tabular}
}
\label{tab:ffl_classification}
\end{table}

%% file: supplementary-tex/SM-yeast-ffls.tex
                           & \multicolumn{3}{c}{Motif nodes} & Function   & \multicolumn{4}{c}{Class label}                                                                          \\
                           &                                 &            &            &                                              & Motif-based & Edge-based & Infomap & Louvain \\ \midrule
                           & GAL11                           & ALPHA1     & MFALPHA1   & pheromone response                           & 1           & 1          & -1      & -1      \\
                           & GCN4                            & MET4       & MET16      & Metionine biosynthesis                       & 2           & 2          & 1       & -1      \\
                           & GCN4                            & MET4       & MET17      & Metionine biosynthesis                       & 2           & 2          & 1       & -1      \\
                           & GCN4                            & LEU3       & ILV1       & Leucine and branched amino acid biosynthesis & 2           & 2          & 1       & 1       \\
                           & GCN4                            & LEU3       & ILV2       & Leucine and branched amino acid biosynthesis & 2           & 2          & 1       & 1       \\          
                           & GCN4                            & LEU3       & ILV5       & Leucine and branched amino acid biosynthesis & 2           & 2          & 1       & 1       \\         
                           & GCN4                            & LEU3       & LEU4       & Leucine and branched amino acid biosynthesis & 2           & 2          & 1       & 1       \\
                           & GLN3                            & GAT1       & GAP1       & Nitrogen utilization                         & 3           & 3          & 1       & 2       \\
                           & GLN3                            & GAT1       & DAL80      & Nitrogen utilization                         & 3           & 3          & 1       & 2       \\
                           & GLN3                            & GAT1       & DAL80/GZF3 & Glutamate synthetase                         & 3           & 3          & 1       & 2       \\
                           & GLN3                            & GAT1       & GLN1       & Glutamate synthetase                         & 3           & 3          & 1       & 2       \\
                           & MIG1                            & HAP2/3/4/5 & CYC1       & formation of apocytochromes                  & 4           & 4          & -1      & -1      \\
                           & MIG1                            & GAL4       & GAL1       & Galactokinase                                & 4           & -1         & -1      & -1      \\
                           & PDR1                            & YRR1       & SNQ2       & Drug resistance                              & 5           & 5          & 2       & 3       \\
                           & PDR1                            & YRR1       & YOR1       & Drug resistance                              & 5           & 5          & 2       & 3       \\
                           & PDR1                            & PDR3       & HXT11      & Drug resistance                              & 5           & 5          & 2       & 3       \\
                           & PDR1                            & PDR3       & HXT9       & Drug resistance                              & 5           & 5          & 2       & 3       \\
                           & PDR1                            & PDR3       & PDR5       & Drug resistance                              & 5           & 5          & 2       & 3       \\
                           & PDR1                            & PDR3       & IPT1       & Drug resistance                              & 5           & 5          & 2       & 3       \\       
                           & PDR1                            & PDR3       & SNQ2       & Drug resistance                              & 5           & 5          & 2       & 3       \\
                           & PDR1                            & PDR3       & YOR1       & Drug resistance                              & 5           & 5          & 2       & 3       \\
                           & RIM101                          & IME1       & DIT1       & sporulation-specific                         & 6           & 6          & 3       & 4       \\
                           & SPT16                           & SWI4/SWI6  & CLN1       & Cell cycle and mating type switch            & 7           & -1         & 4       & 5       \\
                           & SPT16                           & SWI4/SWI6  & CLN2       & Cell cycle and mating type switch            & 7           & -1         & -1      & 5       \\
                           & SPT16                           & SWI4/SWI6  & HO         & Cell cycle and mating type switch            & 7           & -1         & -1      & -1      \\
                           & TUP1                            & ALPHA1     & MFALPHA1   & Mating factor alpha                          & 1           & 1          & -1      & 5       \\
                           & UME6                            & INO2/INO4  & CHO1       & Phospholipid biosynthesis                    & 8           & 6          & 5       & 4       \\
                           & UME6                            & INO2/INO4  & CHO2       & Phospholipid biosynthesis                    & 8           & 6          & 5       & 4       \\
                           & UME6                            & INO2/INO4  & OPI3       & Phospholipid biosynthesis                    & 8           & 6          & 5       & 4       \\ \midrule
Frac. coherently labeled &                                 &            &            &                                              & 29 / 29     & 25 / 29    & 23 / 29  & 23 / 29 \\
Accuracy                   &                                 &            &            &                                              & 0.97        & 0.82       & 0.68    & 0.76    \\

%% file: supplementary-tex/SM-enwiki.tex

The English Wikipedia
network~\cite{boldi2004webgraph,boldi2011layered,boldi2004ubicrawler} consists
of 4.21 million nodes (representing articles) and 101.31 million edges, where an
edge from node $i$ to node $j$ means that there is a hyperlink from the $i$th
article to the $j$th article.  The network data was downloaded from
\url{http://law.di.unimi.it/webdata/enwiki-2013/}.

We used Algorithm~\ref{alg:motif_fiedler} to find a motif-based cluster for
motif $\motiftype{3}{6}$ and $\medge$ (the algorithm was run on the largest
connected component of the motif adjacency matrix).  The clusters are shown in
Figure~\ref{fig:enwiki_comms}.  The nodes in the motif-based cluster are cities
and barangays (small administrative divisions) in the Philippines.  The cluster
has a set of nodes with many outgoing links that form the source node in motif
$\motiftype{3}{6}$.  In total, the cluster consists of 22 nodes and 338 edges.
The linking pattern appears anomalous and suggests that perhaps the pages
uplinking should receive reciprocated links.  On the other hand, the edge-based
cluster is much larger cluster and does not have too much structure.  The
cluster consists of several high-degree nodes and their neighbors.

\clearpage
\begin{figure}[htb]
\begin{centering}
\includegraphics[width=1\columnwidth]{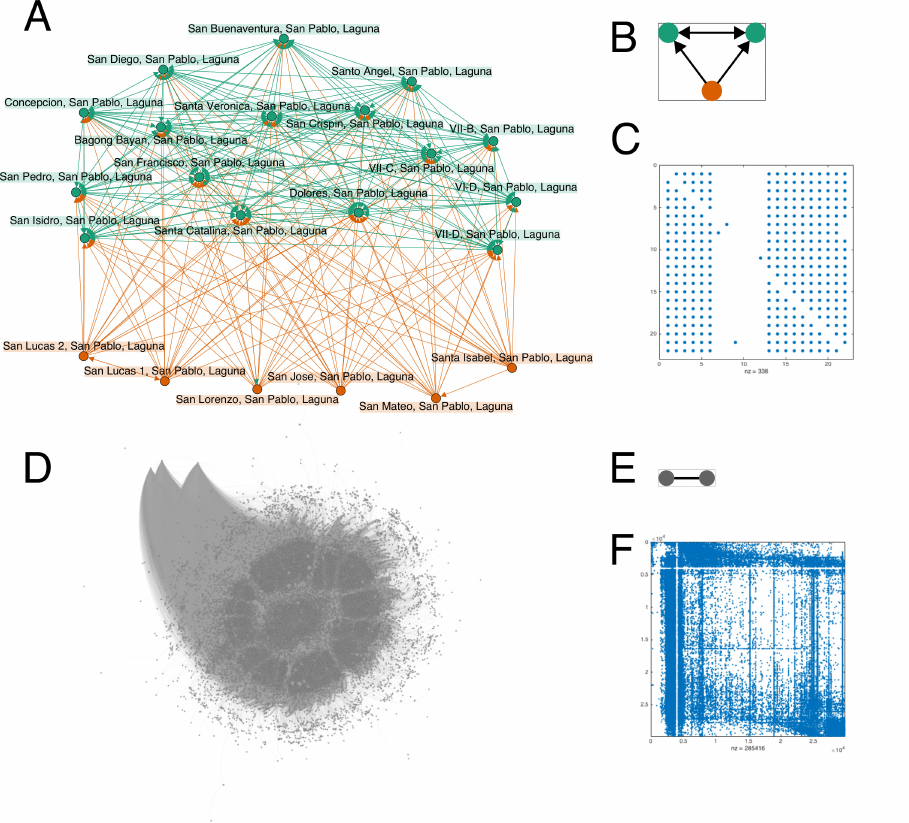}
\caption{%
Clusters from the English Wikipedia hyperlink
network~\cite{boldi2004webgraph,boldi2011layered,boldi2004ubicrawler}.
{\bf A--C}:
Motif-based cluster (A) for motif $\motiftype{3}{6}$ (B).  The cluster
consists of cities and small administrative divisions in the Philippines.  The
green nodes have many bi-direction links with each other and many incoming links
from orange nodes at the bottom of the figure.  The spy plot illustrates this
network structure (C).
{\bf D--F}:
Cluster (D) for undirected edges (E).  The cluster has a few very
high-degree nodes, as evidenced by the spy plot (F).
}
\label{fig:enwiki_comms}
\end{centering}
\end{figure}
\clearpage

%% file: supplementary-tex/SM-twitter.tex
We also analyzed the complete 2010 Twitter follower
graph~\cite{boldi2004webgraph,boldi2011layered,kwak2010twitter}.  The graph
consists 41.65 million nodes (users) and 1.47 billion edges, where an edge from
node $i$ to node $j$ signifies that user $i$ is followed by user $j$ on the
social network.  The network data was downloaded from
\url{http://law.di.unimi.it/webdata/twitter-2010/}.

We used Algorithm~\ref{alg:motif_fiedler} to find a motif-based cluster for
motif $\motiftype{3}{6}$ (the algorithm was run on the largest connected
component of the motif adjacency matrix).  The cluster contains 151 nodes and
consists of two disconnected components.  Here, we consider the smaller of the
two components, which consists of 38 nodes.  We also found an edge-based cluster
on the undirected graph (using Algorithm~\ref{alg:motif_fiedler} with motif
$\medge$).  This cluster consists of 44 nodes.

Figure~\ref{fig:twitter_comms} illustrates the motif-based and edge-based
clusters.  Both clusters capture anomalies in the graph.  The motif-based
cluster consists of holding accounts for a photography company.  The nodes that
form bi-directional links have completed profiles (contain a profile picture)
while several nodes with incomplete profiles (without a profile picture) are
followed by the completed accounts.  The edge-based cluster is a near clique,
where the user screen names all begin with ``LC\_''.  We suspect that the
similar usernames are either true social communities, holding accounts, or bots.
(For the most part, their tweets are protected, so we could not verify if any of
these scenarios are true).  Interestingly, both $\motiftype{3}{6}$ and $\medge$
find anomalous clusters.  However, their structures are quite different.  We
conclude that $\motiftype{3}{6}$ can lead to the detection of new anomalous
clusters in social networks.

\clearpage
\begin{figure}[htb]
\begin{centering}
\includegraphics[width=1.0\columnwidth]{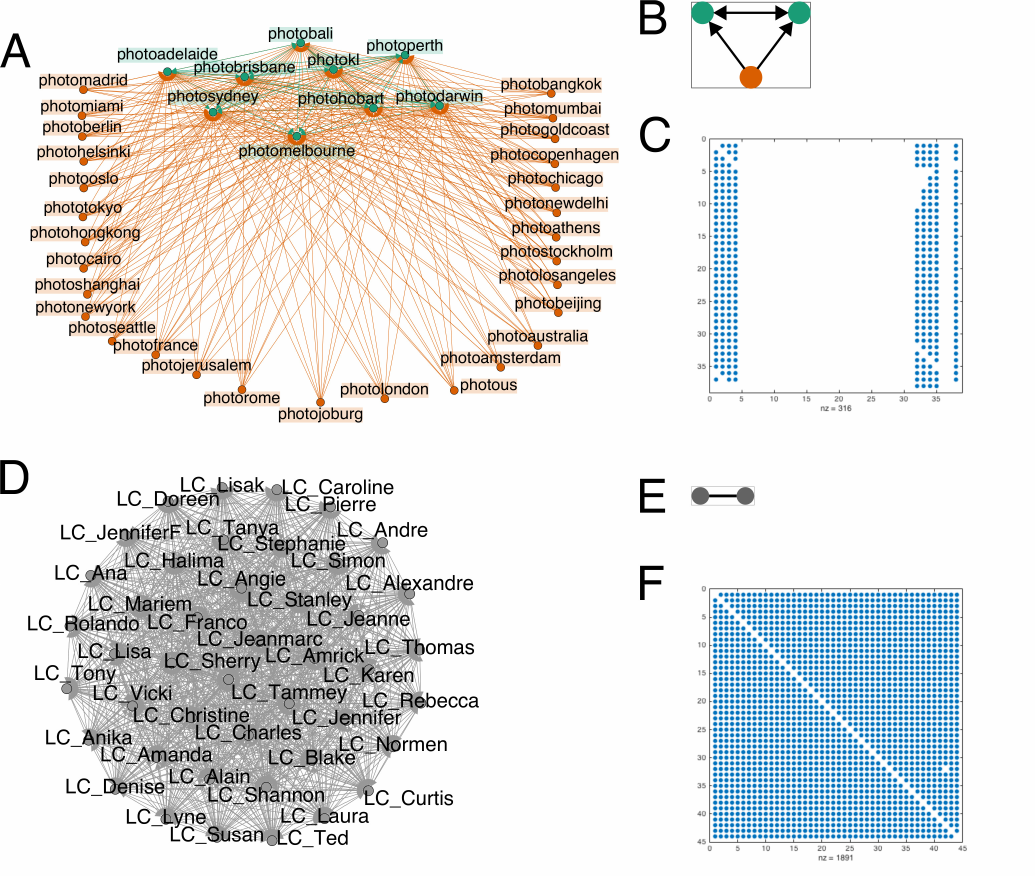}
\caption{%
Clusters in the 2010 Twitter follower
network~\cite{boldi2004webgraph,boldi2011layered,kwak2010twitter}.
{\bf A--C}:
Motif-based cluster (A) for motif $\motiftype{3}{6}$ (B).  All accounts are
holding accounts for a photography company.  The green nodes correspond to
accounts that have completed profiles, while the orange accounts have incomplete
profiles.  The spy plot illustrates how the cluster is formed around this motif
(C).
{\bf D--F}:
Cluster (D) for edge-based clustering (E).  The cluster consists of a
near-clique (F) where all users have the prefix ``LC\_''.
}
\label{fig:twitter_comms}
\end{centering}
\end{figure}
\clearpage

%% file: supplementary-tex/SM-stanford.tex

The Stanford web graph~\cite{leskovec2009community,snapnets} consists of 281,903
nodes and 2,312,497 edges, where an edge from node $i$ to node $j$ means that
there is a hyperlink from the $i$th web page to the $j$th web page.  Here, all
of the web pages come from the Stanford domain.  The network data was downloaded
from \url{http://snap.stanford.edu/data/web-Stanford.html}.

We used Algorithm~\ref{alg:motif_fiedler} to find a motif-based cluster for
motif $\motiftype{3}{7}$, a motif that is over-expressed in web
graphs~\cite{milo2002network}.  An illustration of the cluster and an edge-based
cluster (\emph{i.e.}, using Algorithm~\ref{alg:motif_fiedler} with $\medge$) are
in Figure~\ref{fig:web_stanford}.  Interestingly, both clusters exhibits a
core-periphery structure, albeit markedly different ones.  The motif-based
cluster contains several core nodes with large in-degree.  Such core nodes
comprise the sink node in motif $\motiftype{3}{7}$.  On the periphery are
several clusters within which are many bi-directional links (as illustrated by
the spy plot in Figure~\ref{fig:web_stanford}).  The nodes in these clusters
then up-link to the core nodes.  This type of organizational unit suggests an
explanation for why motif $\motiftype{3}{7}$ is over-expressed: clusters of
similar pages tend to uplink to more central pages.  The edge-based cluster also
has a few nodes with large in-degree, serving as a small core.  On the periphery
are the neighbors of these nodes, which themselves tend \emph{not} to be
connected (as illustrated by the spy plot).

\clearpage
\begin{figure}[htb]
\begin{centering}
\includegraphics[width=0.85\columnwidth]{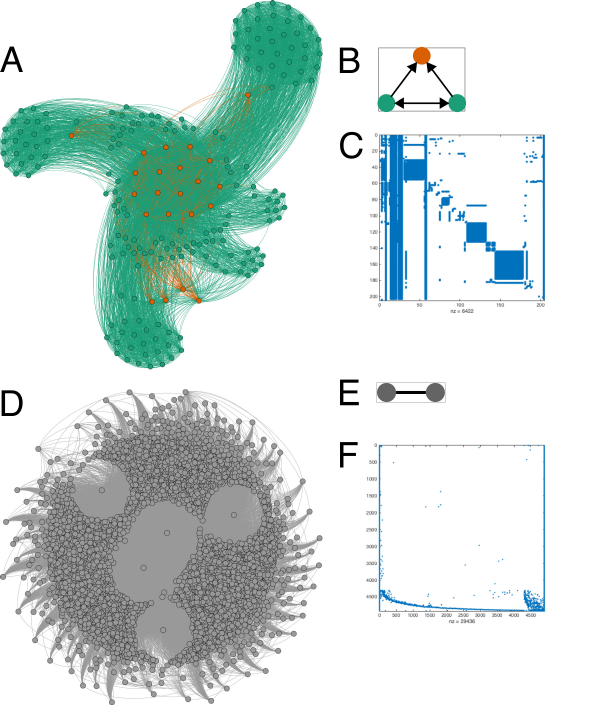}
\caption{%
Clusters in the Stanford web graph~\cite{leskovec2009community}.
{\bf A--C:}
Motif-based cluster (A) for motif $\motiftype{3}{7}$ (B).  The cluster has
a core group of nodes with many incoming links (serving as the sink node in
$\motiftype{3}{7}$; shown in orange) and several periphery groups that are tied
together (the bi-directional link in $\motiftype{3}{7}$; shown in green) and
also up-link to the core.  This is evident from the spy plot (C).
{\bf D--F:}
Cluster (C) for undirected edges (B).  The cluster contains a few
high-degree nodes and their neighbors, and the neighbors tend to not be
connected, as illustrated by the splot (F).
}
\label{fig:web_stanford}
\end{centering}
\end{figure}
\clearpage

%% file: supplementary-tex/SM-collaboration.tex

We used Algorithm~\ref{alg:motif_fiedler} to identify clusters of a four-node
motif (the semi-clique) that has been studied in conjunction with researcher
productivity in collaboration networks~\cite{chakraborty2014automatic} (see
Figure~\ref{fig:collab}).  We found a motif-based cluster in two different
collaboration networks.  Each one is derived from co-authorship in papers
submitted to the arXiv under a certain category; here, we analyze the "High
Energy Physics--Theory" (HepTh) and "Condensed Matter Physics" (CondMat)
categories~\cite{leskovec2007graph,snapnets}.  The HepTh network has 23,133
nodes and 93,497 edges and the CondMat network has 9,877 nodes and 25,998 edges.
The HepTh network data was downloaded from
\url{http://snap.stanford.edu/data/ca-HepTh.html}
and the CondMat network data was downloaded from
\url{http://snap.stanford.edu/data/ca-CondMat.html}.

Figure~\ref{fig:collab} shows the two clusters for each of the collaboration
networks.  In both networks, the motif-based cluster consists of a core group of
nodes and similarly-sized groups on the periphery.  The core group of nodes
correspond to the nodes of degree 3 in the motif and the periphery group nodes
correspond to the nodes of degree 2.  One explanation for this organization is
that there is a small small group of authors that writes papers with different
research groups.  Alternatively, the co-authorship could come from a single
research group, where senior authors are included on all of the papers and
junior authors on a subset of the papers.

On the other hand, the edge-based clusters (\emph{i.e.}, result of
Algorithm~\ref{alg:motif_fiedler} for $\medge$) are a clique in the HepTh
netowork and a clique with a few dangling nodes in the CondMat network.  The
dense clusters are quite different from the sparser clusters based on the
semi-clique.  Such dense clusters are not that surprising.  For example, a
clique could arise from a single paper published by a group of authors.

\clearpage
\begin{figure}[htb]
\begin{centering}
\includegraphics[width=1.0\columnwidth]{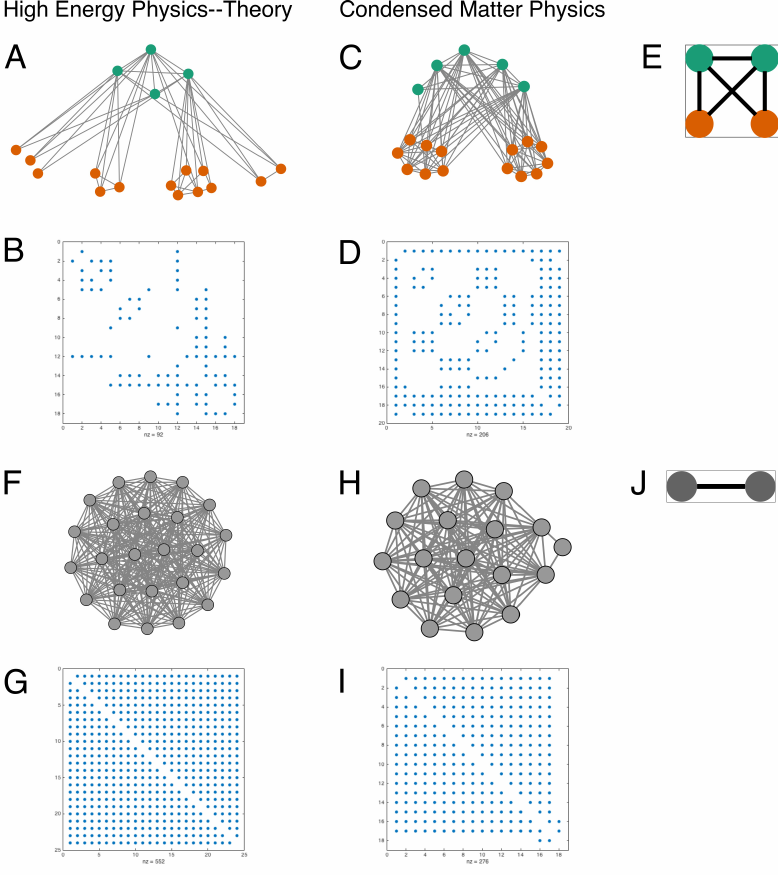}
\caption{%
Clusters in co-authorship networks~\cite{leskovec2007graph}.
{\bf A--E:}
Best motif-based cluster for the semi-clique motif (E) in the High Energy
Physics--Theory collaboration network (A) and the Condensed Matter Physics
collaboration network (C).  Corresponding spy plots are shown in (B) and (D).
{\bf F--I:}
Best edge-based (I) cluster in the High Energy Physics--Theory collaboration
network (F) and the Condensed Matter Physics collaboration network (H).
Corresponding spy plots are shown in (G) and (I).
}
\label{fig:collab}
\end{centering}
\end{figure}
\clearpage

%% file: supplementary-tex/SM-data.tex
All data is available at our project web site at \projecturl.  The web site
includes links to datasets used for experiments throughout the supplementary
material%
~\cite{west2014exploiting,snapnets,leskovec2005graphs,gehrke2003overview,albert1999internet,leskovec2007dynamics,leskovec2010governance,leskovec2012learning,boldi2004webgraph,boldi2011layered,boldi2014bubing,takac2012data,backstrom2006group,leskovec2009community,yang2012defining}.

%% file: combined.bbl
\begin{thebibliography}{10}
\bibitem{note-spectral-ordering}
The normalized motif Laplacian matrix is $\mathcal{L}_M = D^{-1/2} (D - W_M)
  D^{-1/2}$, where $D$ is a diagonal matrix with the row-sums of $W_M$ on the
  diagonal ($D_{ii} = \sum_j (W_M)_{ij}$), and $D^{-1/2}$ is the same matrix
  with the inverse square-roots on the diagonal ($D^{-1/2}_{ii} =
  1/\sqrt{\sum_{j} (W_M)_{ij}}$). The spectral ordering $\sigma$ is the
  by-value ordering of $D^{-1/2}z$, where $z$ is the eigenvector corresponding
  to the second smallest eigenvalue of $\mathcal{L}_M$, \emph{i.e.}, $\sigma_i$
  is the index of $D^{-1/2}z$ with the $i$th smallest value.
\end{thebibliography}
